\documentclass[12pt]{article}
\usepackage{graphicx}
\usepackage{setspace}
 \DeclareGraphicsExtensions{.eps, .ps}
\usepackage{soul,color}
\usepackage{amsmath}
\usepackage{amsthm}
\usepackage{amsfonts}
\usepackage{newfloat}
\usepackage{bbm}
\usepackage{natbib}
\usepackage{hyperref}
\usepackage{subcaption}
\usepackage{authblk}

\DeclareFloatingEnvironment[name={Supplementary Figure}]{suppfigure}

\usepackage{setspace}
\usepackage[margin=1in]{geometry}
\theoremstyle{plain}
\newtheorem{theorem}{Theorem}

\newtheorem{remark}{Remark}
\newtheorem{corollary}{Corollary}

\newtheorem{lemma}{Lemma}
\newtheorem{proposition}{Proposition}

\newcommand\independent{\protect\mathpalette{\protect\independenT}{\perp}}
\def\independenT#1#2{\mathrel{\rlap{$#1#2$}\mkern2mu{#1#2}}}

\newcommand{\psinum}{\psi_{h}^{\text{num}}}
\newcommand{\psiden}{\psi_{h}^{\text{den}}}

\begin{document}
\begin{singlespace}
\title{Causal Effect Estimation after Propensity Score Trimming with Continuous Treatments}
\author[1]{Zach Branson}
\author[1]{Edward H. Kennedy}
\author[1]{Sivaraman Balakrishnan}
\author[1]{Larry Wasserman}
\affil[1]{Department of Statistics and Data Science, Carnegie Mellon University, USA}
\date{}
\maketitle
\begin{abstract}
  Propensity score trimming, which discards subjects with propensity scores below a threshold, is a common way to address positivity violations that complicate causal effect estimation. However, most works on trimming assume treatment is discrete and models for the outcome regression and propensity score are parametric. This work proposes nonparametric estimators for trimmed average causal effects in the case of continuous treatments based on efficient influence functions. For continuous treatments, an efficient influence function for a trimmed causal effect does not exist, due to a lack of pathwise differentiability induced by trimming and a continuous treatment. Thus, we target a smoothed version of the trimmed causal effect for which an efficient influence function exists. Our resulting estimators exhibit doubly-robust style guarantees, with error involving products or squares of errors for the outcome regression and propensity score, which allows for valid inference even when nonparametric models are used. Our results allow the trimming threshold to be fixed or defined as a quantile of the propensity score, such that confidence intervals incorporate uncertainty involved in threshold estimation. These findings are validated via simulation and an application, thereby showing how to efficiently-but-flexibly estimate trimmed causal effects with continuous treatments.
\end{abstract}

\noindent%
\end{singlespace}
\vfill

\section{Introduction}

We consider estimating average treatment effects (ATEs) in an observational study with a continuous treatment. Under standard identification assumptions, the ATE can be estimated with some combination of an outcome model and a propensity score model. A ubiquitous assumption is that every subject has a non-zero probability of receiving any treatment conditional on covariates, i.e. ``positivity'' \citep{imbens2004nonparametric,hernan2006estimating,d2021overlap}. Positivity can be less tenable with continuous treatments, because it can be unlikely that every subject has a non-zero conditional density at every treatment value. Furthermore, small propensity scores adversely affect the bias and variance of most ATE estimators \citep{westreich2010invited,khan2010irregular,petersen2012diagnosing}. Thus, one common approach is to discard subjects whose propensity scores are below a threshold when estimating causal effects; this is known as propensity score trimming \citep{frolich2004finite,smith2005does,crump2009dealing}.

Many empirical studies have found that trimming improves the accuracy of causal effect estimators, such that it is common in practice \citep{cole2008constructing,sturmer2010treatment,lee2011weight,yang2016propensity}. However, most works assume treatment is discrete, such that there is little guidance on how to conduct trimming for continuous treatments \citep{wu2024matching}. Furthermore, most works assume that the outcome and propensity score model are parametric. Causal effect estimators based on parametric models are especially sensitive to model misspecification when there is limited positivity (e.g., \citealt[Chapter 14]{imbens2015causal}); and if the propensity score model is misspecified, then the trimmed population will be misspecified. This motivates using flexible models for the outcome and propensity score model, but in turn valid inference can be challenging, because the resulting models typically converge at slower-than-parametric rates.

To address these limitations, we derive estimators based on efficient influence functions (EIFs) for trimmed ATEs in the case of continuous treatments. Our resulting estimators have doubly-robust style guarantees, i.e. their error can be expressed as products or squares of errors for the outcome regression and propensity score. Thus, our estimators attain parametric convergence rates even when the outcome regression and propensity score are estimated at slower rates with flexible estimators. While it is known how to construct ATE estimators for discrete treatments via EIFs \citep{hines2022demystifying,kennedy2022semiparametric}, an EIF for a trimmed ATE with a continuous treatment does not exist, due to a lack of pathwise differentiability induced by the nonsmooth nature of trimming \citep{yang2018asymptotic} and a continuous treatment \citep{kennedy2017non}. Thus, we target a smoothed estimand that is close to the trimmed ATE but has an EIF. Our results allow the amount of trimming (the trimming threshold) to depend on the treatment value, because positivity violations may differ across treatments. We derive EIFs when the trimming threshold is fixed or defined as a quantile of the propensity score, such that confidence intervals incorporate uncertainty involved in threshold estimation.

In Section \ref{s:setup}, we outline notation to define trimmed causal estimands. In Section \ref{s:EIFs}, we derive estimators based on EIFs, either when the trimming threshold is fixed or estimated via a quantile of the propensity score. In Section \ref{s:smoothingParameters}, we describe a data-driven way to select smoothing parameters involved in defining our estimands. We then present simulations in Section \ref{s:simulations} and a real application in Section \ref{s:application} which validate that our estimators can flexibly and accurately estimate trimmed causal effects. We conclude with Section \ref{s:conclusion}.

\subsection{Related Work}

Many works have found empirical evidence that trimming can improve doubly robust estimators (e.g., \citealt{mao2019propensity,li2020comment,zhang2022stable}). However, these studies do not establish theoretical guarantees for these trimmed estimators, and usually consider a typical doubly robust estimator implemented on a trimmed sample that is assumed fixed. However, the proportion of trimmed subjects can be viewed as an estimand whose uncertainty is driven by uncertainty in the propensity score. We derive EIFs for the proportion of trimmed subjects, such that inference naturally incorporates uncertainty in this proportion.

More generally, our work adds to literature that tweaks inverse propensity score weighted estimators to be robust to extreme propensity scores \citep{cao2009improving,chaudhuri2014heavy,rothe2017robust,yang2018asymptotic,ma2020robust,heiler2021valid,sasaki2022estimation,khan2022doubly,liu2023biased,ma2023}. Our work is most similar to \cite{yang2018asymptotic} and \cite{khan2022doubly}, although they both focus on binary treatments. \cite{yang2018asymptotic}, like us, study a doubly robust estimator that smoothes the trimming indicator and quantifies uncertainty in the proportion of trimmed subjects; however, they assume the propensity score is estimated with a well-specified parametric generalized linear model. Meanwhile, we establish conditions for valid inference even when the propensity score and outcome regression are estimated with flexible nonparametric models. Meanwhile, \cite{khan2022doubly} also derive a trimmed doubly robust estimator that yields valid inference even when nonparametric models are used; however, their estimand is the ATE for the estimated trimmed set, instead of for the true trimmed set. In short, the estimated trimmed set is the subjects for whom $\widehat{\pi} \geq \widehat{t}$, instead of for whom $\pi \geq t$. \cite{khan2022doubly} note that deriving an EIF of the ATE for the true trimmed set may be fruitful; that is the approach we take here.

Finally, the aforementioned works assume the treatment is binary. Thus, our work contributes to the literature on dose-response estimation \citep{diaz2013targeted,kennedy2017non,fong2018covariate,bonvini2022fast,huling2023independence,wu2024matching} in the case of positivity violations. In the
Supplementary Material we show how our results simplify when the treatment is discrete, thereby providing a way to nonparametrically estimate trimmed causal effects in
the canonical discrete case.

\section{Setup and Causal Estimands} \label{s:setup}

\subsection{Notation}

For a possibly random function $f$, we denote sample averages with $\mathbb{P}_n\{ f(Z)\} = n^{-1} \sum_{i=1}^n f(Z_i)$ and sample variances with $\text{Var}_n\{ f(Z) \} = (n-1)^{-1} \sum_{i=1}^n [f(Z_i) - \mathbb{P}_n\{ f(Z)\} ]^2$. We write the squared $L_2$-norm as $\|f\|^2 = \int f(z)^2 dP(z)$. We use $\mathbb{I}(B)$ to denote the indicator function for an event $B$. We use $\Phi_{\epsilon}(z)$ and $\phi_{\epsilon}(z)$ to denote the cumulative distribution function (CDF) and probability density function (PDF) of the distribution $N(0, \epsilon^2)$. 

\subsection{Causal Estimands for Continuous Treatments}

Let $X$ denote a vector of covariates, $A$ a continuous treatment, and $Y$ a continuous outcome. Let $Y(a)$ denote the outcome we observe if $A = a$, i.e. the potential outcome. We denote the average treatment effect (ATE) at $A = a$ as $\psi(a) \equiv \mathbb{E}[Y(a)]$. Although $\psi(a)$ is not traditionally considered an ``effect'' because it is not a contrast, for simplicity we write effects in terms of one value $a$. In practice, one would estimate effects for many values.

Three assumptions are sufficient for identifying $\psi(a)$ \citep{gill2001causal}:
\begin{itemize}
  \item[(A1)] \textbf{No Unmeasured Confounding}: $Y(a) \independent A | X$, for all $a$.
  \item[(A2)] \textbf{Consistency}: If $A = a$, then $Y = Y(a)$, for all $a$.
  \item[(A3)] \textbf{Positivity}: There is an $c > 0$ such that $c \leq \pi(a|x)$ for all $a$ and $x$, where $\pi(a|x)$ is the conditional density of treatment $a$ (the propensity score).
\end{itemize}
Under these three assumptions, $\psi(a) = \int \mu(x, a) dP(x)$, where $\mu(x,a) = \mathbb{E}[Y | X = x, A = a]$. This motivates estimating the outcome regression $\mu(x,a)$; e.g., construct $\widehat{\mu}(x,a)$ and average across subjects. If we want to be robust to model misspecification for $\mu(x,a)$, one can develop estimators that utilize the propensity score $\pi(a|x)$. Furthermore, to avoid parametric model specification for $\mu(x,a)$ and $\pi(a|x)$, we can consider estimators based on the efficient influence function (EIF) of their corresponding estimand \citep{van2000asymptotic,tsiatis2006semiparametric,kennedy2022semiparametric}. However, an EIF for $\psi(a)$ does not exist, due to a lack of pathwise differentiability at $A = a$ \citep{diaz2013targeted}. One option is to derive an EIF for an estimand that smooths across $A$. Given a smooth, symmetric kernel $K_h(\cdot)$ with bandwidth $h$, we can define and identify a kernel-smoothed version of $\psi(a)$:
\begin{align}
  \psi_h(a) = \int \int K_h(a_0 - a) \mu(x, a_0) d a_0 dP(x) . \label{eqn:psih}
\end{align}
We discuss an EIF-based estimator for $\psi_h(a)$ in Section \ref{s:EIFs}.

The above identification results assume positivity holds, and estimation can be difficult if propensity scores are small \citep{westreich2010invited,khan2010irregular,petersen2012diagnosing}. Furthermore, for continuous treatments, it's likely that $\pi(a|x)$ will be small for some $a$. This motivates trimming.

\subsection{Trimmed Causal Estimands}

Propensity score trimming targets the ATE among subjects with non-extreme propensity scores. Given a trimming threshold $t$, the trimmed ATE is defined and identified as
\begin{align}
  \psi(a; t) = \mathbb{E}[Y(a) | \pi(a|X) > t] = \frac{\int \mathbb{I}(\pi(a|x) > t) \mu(x,a) dP(x)}{\int \mathbb{I}(\pi(a|x) > t) dP(x)} . \label{eqn:simpleIdentification}
\end{align}
When $A$ is binary, it is common to also trim subjects whose propensity scores are close to one, because then $P(A = 0 | X)$ is close to zero. However, for continuous treatments, we only have to trim subjects whose propensity scores are close to zero at $A = a$.

Equation (\ref{eqn:simpleIdentification}) suggests the following plug-in estimator:
\begin{align}
  \widehat{\psi}(a; t) = \frac{\mathbb{P}_n\{ \mathbb{I}(\widehat{\pi}(a|X) > t) \widehat{\mu}(X,a) \}}{\mathbb{P}_n\{ \mathbb{I}(\widehat{\pi}(a|X) > t) \}}, \label{eqn:plugin}
\end{align}
where $\widehat{\pi}(a|X)$ and $\widehat{\mu}(X,a)$ denote estimators for the propensity score and outcome regression. Intuitively, this estimator would require both $\mu(x,a)$ and $\pi(a|x)$ to be well-specified, thereby motivating flexible models. However, the resulting estimator would typically inherit slow nonparametric convergence rates, which would also adversely affect inference.

Instead, we use efficiency theory to derive EIF-based estimators, such that parametric rates can still be achieved with nonparametric models. An EIF for the trimmed ATE does not exist, due to a lack of differentiability at $A = a$ and from the indicator $\mathbb{I}(\pi(a|x) > t)$. To address the former, we use kernel smoothing with $K_h(\cdot)$; and to address the latter, use a smoothed version of $\mathbb{I}(\pi(a|x) > t)$, denoted $S(\pi(a|x), t)$. One example is $S(\pi(a|x), t) = \Phi_{\epsilon} \left\{ \pi(a|x) - t \right\}$ \citep{yang2018asymptotic}. We notationally suppress dependence on the tuning parameter $\epsilon$ to remain agnostic to the smoothing indicator one may use.

Given $K_h(\cdot)$ and $S(\pi(a|x), t)$, the smoothed trimmed ATE (STATE) is identified as:
\begin{align}
  \psi_{h}(a; t) = \frac{\int \int K_h(a_0 - a) S(\pi(a_0|x), t) \mu(x,a_0) da_0 dP(x)}{\int \int K_h(a_0 - a) S(\pi(a_0|x), t) da_0 dP(x)} . \label{eqn:causalEstimandSmoothed}
\end{align}
The STATE is a weighted average of $\psi_h(a)$ in (\ref{eqn:psih}), where subjects with $\pi(a|x) > t$ are weighed close to one, and subjects with $\pi(a|x) < t$ are weighed close to zero. The amount $\psi_{h}(a; t)$ differs from $\psi(a; t)$ depends on $h$ and $\epsilon$. In Section \ref{s:EIFs}, we derive EIF-based estimators for the STATE when $t$ is fixed or estimated as a quantile of the propensity score; in both cases we assume $h$ and $\epsilon$ are fixed. We discuss how to choose $h$ and $\epsilon$ in Section \ref{s:smoothingParameters}.

Note that the trimmed estimands $\psi(a; t)$ and $\psi_h(a; t)$ differ from their non-trimmed counterparts $\psi(a)$ and $\psi_h(a)$ to the extent that the trimmed population differs from the full population \citep{crump2009dealing,sturmer2010treatment,yang2018asymptotic}. Because trimming changes the target population, interpreting trimmed effects can be challenging. In Sections \ref{s:EIFs}-\ref{s:simulations} we focus on estimation of trimmed effects, and in Section \ref{s:application} we discuss how to interpret trimmed effects. In short, in addition to estimating trimmed effects, our approach allows one to examine how trimming changes the population across treatment values, which aids interpretation.

\section{Efficient Estimators for Trimmed Causal Effects} \label{s:EIFs}

Here we derive EIF-based estimators for the STATE $\psi_{h}(a; t)$ in (\ref{eqn:causalEstimandSmoothed}). Given a kernel $K_h(\cdot)$ and fixed bandwidth $h$, the uncentered EIF of the smoothed ATE $\psi_h(a)$ in (\ref{eqn:psih}) is
  \begin{align}
    \varphi_h(a) = K_h(A - a) \frac{Y - \mu(X,A)}{\pi(A|X)} + \int K_h(a_0 - a) \mu(X,a_0) da_0  \label{eqn:eifPsih}
  \end{align}
\citep{bibaut2017data}. By ``uncentered,'' we mean the EIF is $\varphi_h(a) - \psi_h(a)$, such that $\mathbb{E}[\varphi_h(a)] = \psi_h(a)$. The one-step estimator based on this EIF is:
\begin{align}
  \widehat{\psi}_h(a) = \mathbb{P}_n\{ \widehat{\varphi}_h(a) \} = \mathbb{P}_n \left\{ K_h(A - a) \frac{Y - \widehat{\mu}(X,A)}{\widehat{\pi}(A|X)} + \int K_h(a_0 - a) \widehat{\mu}(X,a_0) da_0 \right\} . \label{eqn:psihHat}
\end{align}
Others have proposed similar estimators that use kernel smoothing, an estimated outcome regression, and inverse propensity score weights \citep{kennedy2017non,kallus2018policy,su2019non,colangelo2020double}. If positivity holds, then, with additional assumptions on the convergence rates of $\widehat{\mu}(a,x)$ and $\widehat{\pi}(a|x)$, one can show $\widehat{\psi}_h(a)$ is $n^{-1/2}$ consistent, asymptotically Normal, and semiparametrically efficient \citep{kennedy2022semiparametric}. However, $\widehat{\psi}_h(a)$ can be unstable when propensity scores are small. To establish EIF-based estimators for the STATE, it's helpful to write $\psi_{h}(a; t) = \psi_{h}^{\text{num}}(a; t)/\psi_{h}^{\text{den}}(a; t)$, where
\begin{align}
  \psi_{h}^{\text{num}}(a; t) &= \int \int K_h(a_0 - a) S(\pi(a_0|x), t) \mu(x,a_0) da_0 dP(x), \text{ and} \label{eqn:funcNum} \\
  \psi_{h}^{\text{den}}(a; t) &= \int \int K_h(a_0 - a) S(\pi(a_0|x), t) da_0 dP(x)  . \label{eqn:funcDenom}
\end{align}
We derive EIFs for $\psi_{h}^{\text{num}}(a; t)$ and $\psi_{h}^{\text{den}}(a; t)$ and corresponding one-step estimators $\widehat{\psi}_{h}^{\text{num}}(a; t)$ and $\widehat{\psi}_{h}^{\text{den}}(a; t)$. Then, we take their ratio.
An asymptotically equivalent approach is to derive the EIF for $\psi_{h}(a; t)$ itself; for completeness we include results from this approach in the Supplementary Material. We focus on the first approach because $\widehat{\psi}_{h}^{\text{den}}(a; t)$ represents the proportion of non-trimmed subjects, which is useful for estimating $t$ in Section \ref{ss:estimatedThreshold}.

\subsection{When the Trimming Threshold is Fixed} \label{ss:fixedThreshold}

A centered EIF $\xi$ of an estimand $\psi$ acts as its functional pathwise derivative and satisfies the von Mises expansion for probability measures $P$ and $\bar{P}$ \citep{newey1994asymptotic,tsiatis2006semiparametric}:
\begin{align}
  \psi(\bar{P}) - \psi(P) = \int \xi(\bar{P}) d(\bar{P} - P)(x) + R_2(\bar{P}, P), \label{eqn:vonMises}
\end{align}
where $R_2(\bar{P}, P)$ is a second-order remainder of a functional Taylor expansion \citep{hines2022demystifying,kennedy2022semiparametric}. In our proofs in the Supplementary Material, we derive a potential EIF $\xi$, verify that the von Mises expansion (\ref{eqn:vonMises}) holds, and establish conditions when $R_2(\bar{P}, P) = o_P(n^{-1/2})$ and thus the estimator is characterized by the distribution of the first term, which is asymptotically Normal with mean zero for the estimators we derive.

Theorem \ref{thm:numDenomEIFs} establishes the EIFs of $\psi_{h}^{\text{num}}(a; t)$ and $\psi_{h}^{\text{den}}(a; t)$ when $t$ is fixed. The EIFs depend on the derivative of $S(\pi(a|x), t)$ with respect to $\pi(a|x)$, denoted $\partial S(\pi(a|x),t)/\partial \pi$. For example, when $S(\pi(a|x),t) = \Phi_{\epsilon}\{\pi(a|x) - t\}$, $\partial S(\pi(a|x),t)/\partial \pi = \phi_{\epsilon}\{\pi(a|x) - t\}$.

\begin{theorem}
\label{thm:numDenomEIFs}
  If $t$ is fixed, the uncentered EIFs for $\psi_{h}^{\text{num}}(a; t)$ and $\psi_{h}^{\text{den}}(a; t)$ are, respectively:
  \begin{align*}
    \varphi_{h}^{\text{num}}(a; t) &= K_h(A - a) \frac{\{Y - \mu(X,A)\} S(\pi(A|X), t)}{\pi(A|X)} + \int K_h(a_0 - a) S(\pi(a_0|X),t) \mu(X,a_0) da_0 \\ &+ K_h(A - a) \mu(X,A) \frac{\partial S(\pi(A | X), t) }{\partial \pi} - \int K_h(a_0 - a) \mu(X,a_0) \frac{\partial S(\pi(a_0 | X), t) }{\partial \pi} \pi(a_0|X) da_0 \\
  \varphi_{h}^{\text{den}}(a; t) &= \int K_h(a_0 - a) S(\pi(a_0 | X), t) da_0 \\ &+ K_h(A - a) \frac{\partial S(\pi(A | X), t) }{\partial \pi} - \int K_h(a_0 - a) \frac{\partial S(\pi(a_0 | X), t) }{\partial \pi} \pi(a_0|X) da_0  .
  \end{align*}
\end{theorem}
\noindent
The first two terms of $\varphi^{\text{num}}_h(a; t)$ are analogous to those that define $\varphi_h(a)$ in (\ref{eqn:eifPsih}). Meanwhile, the last two terms reflect the fact that $S(\pi(a|x), t)$ depends $\pi(a|x)$. These terms also appear in $\varphi^{\text{den}}_h(a; t)$, which represents the proportion of non-trimmed subjects. Thus, the resulting estimators incorporate uncertainty in the estimated trimmed population. Let $\widehat{\varphi}^{\text{num}}_{h}(a; t)$ equal $\varphi^{\text{num}}_{h}(a; t)$, but with $\mu(x,a)$ and $\pi(a|x)$ replaced by estimators $\widehat{\mu}(x,a)$ and $\widehat{\pi}(a|x)$. Define $\widehat{\varphi}^{\text{den}}_{h}(a; t)$ analogously. One-step estimators based on Theorem \ref{thm:numDenomEIFs} are:
\begin{align}
  \widehat{\psi}_{h}^{\text{num}}(a; t) = \mathbb{P}_n \left\{ \widehat{\varphi}^{\text{num}}_{h}(a; t) \right\}, \hspace{0.1in} \text{and} \hspace{0.1in}
  \widehat{\psi}_{h}^{\text{den}}(a; t) = \mathbb{P}_n \left\{ \widehat{\varphi}^{\text{den}}_{h}(a; t) \right\}  , \label{eqn:eifNumDenEst}
\end{align}
and an estimator for the estimand of interest, $\psi_{h}(a; t)$, is
\begin{align}
  \widehat{\psi}_{h}(a; t) = \widehat{\psi}_{h}^{\text{num}}(a; t)/\widehat{\psi}_{h}^{\text{den}}(a; t)  . \label{eqn:eifRatioTwoEsts}
\end{align} 

To establish the asymptotic distribution of $\widehat{\psi}_{h}(a; t)$, we must make further assumptions about $\widehat{\mu}(x,a)$ and $\widehat{\pi}(a|x)$. One could employ Donsker-type conditions, but this can restrict which methods are used for $\widehat{\mu}(x,a)$ and $\widehat{\pi}(a|x)$ \citep{zheng2010asymptotic}. Instead, we assume the sample used to estimate $\mu(x,a)$ and $\pi(a|x)$ is independent of the sample used to compute $\widehat{\psi}_h(a; t)$, e.g. via sample-splitting. The following result establishes rate conditions on $\widehat{\mu}(x,a)$ and $\widehat{\pi}(a|x)$ that are sufficient for $\widehat{\psi}_{h}(a; t)$ to be asymptotically Normal.
\begin{proposition}
\label{prop:inferenceFixedT}
  If $t$ is fixed, $\mu(x,a)$ and $\pi(a|x)$ are estimated on an independent sample of size $n$, $\big(\|\widehat{\pi}(a|X) - \pi(a|X)\| + \|\widehat{\mu}(X,a) - \mu(X,a)\| \big) \|\widehat{\pi}(a|X) - \pi(a|X)\| = o_P(n^{-1/2})$, and $\| \widehat{\varphi}_h^{\text{num}}(a; t) - \varphi_h^{\text{num}}(a; t) \| = o_p(1)$ and $\| \widehat{\varphi}_h^{\text{den}}(a; t) - \varphi_h^{\text{den}}(a; t) \| = o_p(1)$,  then
  \begin{align}
    \sqrt{n} \left( \widehat{\psi}_{h}(a; t) - \psi_{h}(a; t) \right) \rightarrow N\left( 0, \text{Var} \left\{ \frac{\varphi^{\text{num}}_{h}(a; t) - \psi_{h}(a; t) \varphi^{\text{den}}_{h}(a; t)}{\psiden(a; t)} \right\} \right),
  \end{align}
  where $\varphi^{\text{num}}_{h}(a; t)$ and $\varphi^{\text{den}}_{h}(a; t)$ are the EIFs of $\psinum(a; t)$ and $\psiden(a; t)$ in Theorem \ref{thm:numDenomEIFs}.
\end{proposition}
The independent-sample assumption and consistency assumption for $\widehat{\varphi}_h^{\text{num}}(a; t)$ and $\widehat{\varphi}_h^{\text{den}}(a; t)$ ensures that the empirical process terms of the von Mises expansions are $o_P(n^{-1/2})$. The rate condition on $\widehat{\mu}(x,a)$ and $\widehat{\pi}(a|x)$ ensures the remainder terms are also $o_P(n^{-1/2})$. This condition allows for flexible models that converge at slower-than-parametric rates, e.g. random forests and neural networks \citep{chernozhukov2018double,farrell2021deep}. This condition also suggests that the rate for $\widehat{\pi}(a|x)$ is more crucial than that for $\widehat{\mu}(x,a)$: We require $\|\widehat{\pi}(a|x) - \pi(a|x)\| = o_P(n^{-1/4})$, whereas $\widehat{\mu}(x,a)$ can converge more slowly if $\widehat{\pi}(a|x)$ converges faster than $o_P(n^{-1/4})$. This is intuitive: If the set of trimmed subjects cannot be well-estimated, then we cannot accurately estimate effects among those subjects.

The independent-sample assumption may seem that we need to use half our data to model $\mu(x,a)$ and $\pi(a|x)$. In practice, we recommend cross-fitting to utilize the whole sample \citep{schick1986asymptotically,robins2008higher}. Specifically, $\widehat{\mu}(x,a)$ and $\widehat{\pi}(a|x)$ are constructed on one half of the sample, $\widehat{\varphi}_h^{\text{num}}(a; t)$ and $\widehat{\varphi}_h^{\text{den}}(a; t)$ are constructed on the other half, and the process is repeated with the halves swapped to obtain $\widehat{\varphi}_h^{\text{num}}(a; t)$ and $\widehat{\varphi}_h^{\text{den}}(a; t)$ for the whole sample. More than two folds are also possible. For a review of cross-fitting for doubly robust estimators, see \cite{chernozhukov2018double} and \cite{kennedy2022semiparametric}.

Note that $\psiden(a; t)$ and $\psi_{h}(a; t)$ can be consistently estimated by $\widehat{\psi}^{\text{den}}_{h}(a; t)$ in (\ref{eqn:eifNumDenEst}) and $\widehat{\psi}_{h}(a; t)$ in (\ref{eqn:eifRatioTwoEsts}). Then, Proposition \ref{prop:inferenceFixedT} suggests the following $(1-\alpha)$-level confidence interval:
\begin{align}
  \widehat{\psi}_{h}(a; t) \pm z_{\alpha/2} \sqrt{ \frac{\text{Var}_n \left\{ \frac{\widehat{\varphi}^{\text{num}}_{h}(a; t) - \widehat{\psi}_{h}(a; t) \widehat{\varphi}^{\text{den}}_{h}(a; t)}{\widehat{\psi}^{\text{den}}_{h}(a; t)} \right\}}{n} }, \label{eqn:fixedTCI}
\end{align}
where $z_{\alpha/2}$ is the $\alpha/2$ quantile of a standard Normal distribution.

\subsection{When the Trimming Threshold is an Unknown Parameter} \label{ss:estimatedThreshold}

It may be difficult to specify a trimming threshold \textit{a priori}, and thus it is common to frame $t$ as a parameter that must be estimated. One common choice is the $\gamma$-quantile of the propensity score \citep{cole2008constructing,sturmer2010treatment,lee2011weight}: $t_0 = F^{-1}_{a}(\gamma) \label{eqn:empQuant}$, where $F^{-1}_{a}(\cdot)$ is the inverse CDF of $\pi(a|X)$. Here, $t_0$ denotes the true threshold, which must be estimated. The causal effect estimand is then $\psi_h(a; t_0)$.

An EIF for this choice of $t_0$ does not exist, because it corresponds to the quantile of an unknown density. To admit an EIF, we instead define $t_0$ as $t_0 = F_{a,h}^{-1}(\gamma)$, where $F_{a,h}(t) = 1 - \int \int K_h(a_0 - a) S(\pi(a_0|x), t) da_0 dP(x)$ is the smoothed CDF of $\pi(a|X)$. In this case, $t_0$ is defined such that $\psi_{h}^{\text{den}}(a; t_0) = 1-\gamma$. Thus, a natural estimator for $t_0$ is:
\begin{align}
  \widehat{t} = \inf \{t:\widehat{\psi}_{h}^{\text{den}}(a; t) \leq 1-\gamma\}, \label{eqn:estimatedT}
\end{align}
where $\widehat{\psi}_{h}^{\text{den}}(a; t)$ is defined in (\ref{eqn:eifNumDenEst}). In practice, $\widehat{t}$ can be computed via a line search. Then, a natural estimator for $\psi_{h}(a; t_0)$ is $\widehat{\psi}_{h}(a; \widehat{t}) = \widehat{\psi}_{h}^{\text{num}}(a; \widehat{t})/(1-\gamma)$, where $\widehat{\psi}_{h}^{\text{num}}(a; t)$ is defined in (\ref{eqn:eifNumDenEst}). To study the behavior of $\widehat{\psi}_{h}(a; \widehat{t})$, we can decompose its error by writing
\begin{align*}
  \widehat{\psi}_{h}(a; \widehat{t}) - \psi_{h}(a; t_0) &= (1-\gamma)^{-1} \left( \widehat{\psi}_{h}^{\text{num}}(a; \widehat{t}) - \psi_{h}^{\text{num}}(a; \widehat{t}) + \psi_{h}^{\text{num}}(a; \widehat{t}) - \psi_{h}^{\text{num}}(a; t_0) \right)  .
\end{align*}
The first term represents the error in estimating $\psi_{h}^{\text{num}}(a; t)$, and the second term represents the error in estimating $t_0$, which depends on the error in estimating $\psiden(a; t)$. Thus, the error of $\widehat{\psi}_{h}(a; \widehat{t})$ can be decomposed into the error in estimating $\psinum(a; t)$ and $\psiden(a; t)$ for any fixed $t$. This is stated formally in Theorem \ref{thm:eifErrorEstT}. For the following theorem, we write the uncentered EIF $\varphi^{\text{num}}_{h}(a; t)$ as $\varphi_{h}^{\text{num}}(a; t, \eta)$ for $\eta = (\pi, \mu)$ to emphasize its dependency on $\pi(a|x)$ and $\mu(x,a)$. For example, we can write $\widehat{\psi}_{h}^{\text{num}}(a; t) = \mathbb{P}_n\{ \varphi_{h}^{\text{num}}(a; t, \widehat{\eta}) \}$.

\begin{theorem}
\label{thm:eifErrorEstT}
  Let $t_0$ denote the true threshold for $\gamma \in (0,1)$. Let $\eta_0 = (\pi_0, \mu_0)$ denote the true propensity score and true outcome regression. Finally, let $\varphi_{h}^{\text{num}}(a; t, \eta)$ and $\varphi_{h}^{\text{den}}(a; t, \eta)$ denote the uncentered EIFs of $\psi_{h}^{\text{num}}(a; t)$ and $\psi_{h}^{\text{den}}(a; t)$, defined in Theorem \ref{thm:numDenomEIFs}. Assume that
  \begin{enumerate}
    \item The estimators $\widehat{\eta} = (\widehat{\pi}, \widehat{\mu})$ are estimated on an independent sample of size $n$.
    \item The function classes $\{\varphi_{h}^{\text{num}}(a; t, \eta): t \in \mathbb{R}^+\}$ and $\{\varphi_{h}^{\text{den}}(a; t, \eta): t \in \mathbb{R}^+\}$ are Donsker in $t$ for any fixed $\eta$.
    \item $|\widehat{t} - t_0| = o_P(1)$, $\|\widehat{\eta} - \eta_0\| = o_P(1)$, and $\| \widehat{\varphi}_h^{\text{num}}(a; t) - \varphi_h^{\text{num}}(a; t) \| = o_p(1)$ and $\| \widehat{\varphi}_h^{\text{den}}(a; t) - \varphi_h^{\text{den}}(a; t) \| = o_p(1)$ for any fixed $t$.
    \item The maps $t \mapsto \psi^{\text{num}}_{h}(a; t, \eta)$ and $t \mapsto \psi^{\text{den}}_{h}(a; t, \eta)$ are differentiable at $t_0$ uniformly in $\eta$, where $\frac{\partial}{\partial t} \psi^{\text{num}}_{h}(a; t_0, \widehat{\eta}) \xrightarrow{p} \frac{\partial}{\partial t} \psi^{\text{num}}_{h}(a; t_0, \eta_0)$ and $\frac{\partial}{\partial t} \psi^{\text{den}}_{h}(a; t_0, \widehat{\eta}) \xrightarrow{p} \frac{\partial}{\partial t} \psi^{\text{den}}_{h}(a; t_0, \eta_0)$. Furthermore, $\big| \frac{\partial}{\partial t} \psi^{\text{num}}_{h}(a; t_0, \eta_0) \big|$ and $|1/\frac{\partial}{\partial t} \psi^{\text{den}}_{h}(a; t_0, \eta_0)|$ are bounded.
  \end{enumerate}
  Then,
   \begin{align*}
  (1-\gamma) \{\widehat{\psi}_{h}(a; \widehat{t}) - \psi_{h}(a; t_0)\} &= (\mathbb{P}_n - \mathbb{P})\{ \varphi^{\text{num}}_{h}(a; t_0, \eta_0) \} - \frac{\partial \psi^{\text{num}}_{h}(a; t_0, \eta_0)/\partial t}{\partial \psi^{\text{den}}_{h}(a; t_0, \eta_0)/\partial t}  (\mathbb{P}_n - \mathbb{P})\{ \varphi^{\text{den}}_{h}(a; t_0, \eta_0) \} \\ &+ O_P(R^{\text{num}}_2) + O_P (R_2^{\text{den}})  + o_P(n^{-1/2}),
 \end{align*}
  where $R^{\text{num}}_2 = \mathbb{P}\{ \varphi_{h}^{\text{num}}(a; t_0, \widehat{\eta}) - \varphi_{h}^{\text{num}}(a; t_0, \eta_0) \}$ and $R^{\text{den}}_2 = \mathbb{P}\{ \varphi_{h}^{\text{den}}(a; t_0, \widehat{\eta}) - \varphi_{h}^{\text{den}}(a; t_0, \eta_0) \}$.
\end{theorem}
\noindent
Similar to Proposition \ref{prop:inferenceFixedT}, the first assumption ensures the numerator and denominator estimators behave like sample averages. The second assumption requires that the EIFs are not overly complex functions of $t$, but they can be arbitrarily complex functions of $\mu(x,a)$ and $\pi(a|x)$. The third assumption only requires that the threshold, nuisance functions, and EIFs are consistently estimated at any rate. The last assumption requires some smoothness in $t$, such that we can use a Taylor expansion to characterize the behavior of $\widehat{t} - t_0$.  

Theorem \ref{thm:eifErrorEstT} allows us to establish that $\widehat{\psi}_{h}(a; \widehat{t})$ is $n^{-1/2}$ consistent and asymptotically Normal if $\psinum(a; t)$ and $\psiden(a; t)$ are estimated at $n^{-1/2}$ rates; this is the case when $\|\widehat{\pi}(a|x) - \pi(a|x)\|^2 = o_P(n^{-1/2})$ and $\|\widehat{\pi}(a|x) - \pi(a|x)\| \cdot \|\widehat{\mu}(x,a) - \mu(x, a)\| = o_P(n^{-1/2})$. Thus, we can still conduct valid inference even when the propensity score and outcome regression are flexibly estimated at nonparametric rates. This is stated formally below.
\begin{corollary}
\label{cor:inferenceEstT}
  Assume Assumptions 1-4 in Theorem \ref{thm:eifErrorEstT} hold. Then, if $\widehat{t}$ is estimated as in (\ref{eqn:estimatedT}), and $\big(\|\widehat{\pi}(a|X) - \pi(a|X)\| + \|\widehat{\mu}(X,a) - \mu(X,a)\| \big) \|\widehat{\pi}(a|X) - \pi(a|X)\| = o_P(n^{-1/2})$,
  then,
  \begin{align*}
    \sqrt{n} (1-\gamma) \left( \widehat{\psi}_{h}(a; \widehat{t}) - \psi_{h}(a; t_0) \right) \rightarrow N\left( 0, \text{Var} \left\{ \varphi^{\text{num}}_{h}(a; t_0) - \frac{\partial \psi^{\text{num}}_{h}(a; t_0)/\partial t}{\partial \psi^{\text{den}}_{h}(a; t_0)/\partial t} \varphi^{\text{den}}_{h}(a; t_0) \right\} \right).
  \end{align*}
\end{corollary}
Note that, when $t$ is fixed, Theorem \ref{thm:numDenomEIFs} implies that when $\big(\|\widehat{\pi}(a|X) - \pi(a|X)\| + \|\widehat{\mu}(X,a) - \mu(X,a)\| \big) \|\widehat{\pi}(a|X) - \pi(a|X)\| = o_P(n^{-1/2})$, the asymptotic variance of $\widehat{\psi}_h^{\text{num}}(a; t)$ is $\text{Var} \left\{ \varphi^{\text{num}}_{h}(a; t) \right\}$. Thus, the second term in the variance of Corollary \ref{cor:inferenceEstT} comes from the fact that the trimming threshold $t$ is estimated; this is reflected by the presence of $\varphi^{\text{den}}_{h}(a; t_0)$, which is used to estimate the threshold.

Corollary \ref{cor:inferenceEstT} suggests that we must estimate $\partial \psi^{\text{num}}_{h}(a; t_0)/\partial t$ and $\partial \psi^{\text{den}}_{h}(a; t_0)/\partial t$. Define $\varphi'^{\text{num}}_{h}(a; t)$ as the derivative of $\varphi^{\text{num}}_{h}(a; t)$ with respect to $t$; define $\varphi'^{\text{den}}_{h}(a; t)$ analogously. Then, natural estimators for these derivatives are
\begin{align*}
  \widehat{\psi}'^{\text{num}}_{h}(a; \widehat{t}) = \mathbb{P}_n\left\{ \widehat{\varphi}'^{\text{num}}_{h}(a; \widehat{t}) \right\}, \hspace{0.1in} \text{and } \widehat{\psi}'^{\text{den}}_{h}(a; \widehat{t}) = \mathbb{P}_n\left\{ \widehat{\varphi}'^{\text{den}}_{h}(a; \widehat{t}) \right\},
 \end{align*}
 where $\widehat{\varphi}'^{\text{num}}_{h}(a; \widehat{t})$ and $\widehat{\varphi}'^{\text{den}}_{h}(a; \widehat{t})$ are equal to $\varphi'^{\text{num}}_{h}(a; t)$ and $\varphi'^{\text{den}}_{h}(a; t)$, but with $\mu(x,a)$ and $\pi(a|x)$ replaced by $\widehat{\mu}(x,a)$ and $\widehat{\pi}(a|x)$, and with $t$ replaced by $\widehat{t}$ in (\ref{eqn:estimatedT}). Then, Corollary \ref{cor:inferenceEstT} suggests the following $(1-\alpha)$-level confidence interval for $\psi_{h}(a; t)$:
\begin{align}
  \widehat{\psi}_{h}(a; \widehat{t}) \pm z_{\alpha/2} \sqrt{ \frac{\text{Var}_n\left\{ \widehat{\varphi}^{\text{num}}_{h}(a; \widehat{t}) - \frac{\widehat{\psi}'^{\text{num}}_{h}(a; \widehat{t})}{\widehat{\psi}'^{\text{den}}_{h}(a; \widehat{t})} \widehat{\varphi}^{\text{den}}_{h}(a; \widehat{t}) \right\}}{n(1-\gamma)^2} }. \label{eqn:estTCI}
\end{align}

\section{Choosing Bandwidth $h$ and Smoothing Parameter $\epsilon$} \label{s:smoothingParameters}

In Section \ref{s:EIFs}, we assumed the bandwidth $h$ for the kernel $K_h(\cdot)$ and the parameter $\epsilon$ for the smoothed indicator $S(\pi(a|x), t)$ were prespecified. Here we provide a risk-minimization procedure to select $h$ and a rule-of-thumb based on entropy to select $\epsilon$. For simplicity we assume the threshold $t$ is fixed, but our procedures would be the same if $t$ were estimated.

\subsection{Risk Minimization Procedure to Select Bandwidth $h$} \label{ss:bandwidth}

Here we provide a data-driven procedure for selecting $h$ by minimizing the expected squared difference between our  estimator $\widehat{\psi}_h(a; t)$ in (\ref{eqn:eifRatioTwoEsts}) and a non-smoothed analog. The non-smoothed analog of our estimand $\psi_h(a; t)$, defined in (\ref{eqn:causalEstimandSmoothed}), is
\begin{align*}
  \widetilde{\psi}(a; t) = \frac{\int S(\pi(a|x), t) \mu(x,a) dP(x)}{\int S(\pi(a|x), t) dP(x)},
\end{align*}
where we still view $\epsilon$ as fixed. We would like $\psi_h(a; t)$ to be close to $\widetilde{\psi}(a; t)$, and we can estimate $\psi_h(a; t)$ with $\widehat{\psi}_h(a; t)$. A common bandwidth selection approach in nonparametric regression is to minimize mean squared error \citep{fan1996local,wasserman2006all}. Thus, we consider minimizing $\int \{ \widehat{\psi}_h(a; t) - \widetilde{\psi}(a; t) \}^2 w(a) da$, where $w(a)$ is a weight function. Because $\widetilde{\psi}(a; t)$ does not depend on $h$, this is equivalent to minimizing
\begin{align}
  R \left( \widehat{\psi}_h \right) = \int \left\{ \widehat{\psi}^2_h(a; t) - 2 \widehat{\psi}_h(a; t) \widetilde{\psi}(a; t) \right\} w(a) da, \label{eqn:risk}
\end{align}
which we refer to as the ``risk.'' Thus, we do not let $h$ vary across $a$. Our procedure involves estimating the risk (\ref{eqn:risk}) and, among candidate bandwidths $\mathcal{H}$, selecting the bandwidth that minimizes it. Equation (\ref{eqn:risk}) depends on $w(a)$. A natural choice is $w(a) = \int S(\pi(a|x), t) dP(x)$, because it will be small for treatment values with many small propensity scores. Given this choice for $w(a)$, the risk (\ref{eqn:risk}) becomes
\begin{align*}
  R \left( \widehat{\psi}_h \right) = \int \int \widehat{\psi}^2_h(a; t) S(\pi(a|x), t) da \hspace{0.025in} dP(x) - 2 \int \int \widehat{\psi}_h(a; t) S(\pi(a|x), t) \mu(x, a) da \hspace{0.025in} dP(x).
\end{align*}
Thus, once $\widehat{\psi}_h(a; t)$ is computed for a given bandwidth, a natural risk estimator is
\begin{align}
    \widehat{R} \left( \widehat{\psi}_h \right) = \mathbb{P}_n\left\{ \int \widehat{\psi}^2_h(a; t) S(\widehat{\pi}(a|x), t) da \right\} - 2 \mathbb{P}_n\left\{ \int \widehat{\psi}_h(a; t) S(\widehat{\pi}(a|x), t) \widehat{\mu}(x, a) da \right\} \label{eqn:estimatedRisk}.
\end{align}
As in Section \ref{ss:fixedThreshold}, in practice, the data is split into halves, where one half is used to estimate $\pi(a|x)$ and $\mu(x,a)$, and the other is used to compute $\widehat{\psi}_h(a; t)$ and $\widehat{R} \left( \widehat{\psi}_h \right)$. Cross-fitting can also be used to utilize the full sample size. This process can be computed for all candidate bandwidths $h \in \mathcal{H}$, and finally the bandwidth is chosen as the one that minimizes (\ref{eqn:estimatedRisk}).

\subsection{Entropy-based Procedure to Select Smoothing Parameter $\epsilon$} \label{ss:epsilon}

Unlike kernel smoothing, smoothing $\mathbb{I}(\pi(a|x) > t)$ via $S(\pi(a|x), t)$ does not naturally suggest a data-driven selection procedure for the parameter $\epsilon$. \cite{yang2018asymptotic} recommend conducting several analyses for different $\epsilon$ values, and assess if the analysis is sensitive to $\epsilon$. We also recommend this in practice, but researchers may nonetheless want some guidance for how to choose $\epsilon$. We consider a way to interpret $\epsilon$ that suggests a rule-of-thumb selection procedure. Although our parameter $\epsilon$ refers to $S(\pi(a|x), t) = \Phi_{\epsilon}\{\pi(a|x) - t\}$, the following interpretation would apply for any smoothed indicator with its own parameter.

A smoothed trimmed estimand is a weighted average with weights proportional to $S(\pi(a|x), t)$. Thus, we can interpret the estimand as an average among subjects randomly selected via a selection variable $S | X \sim \text{Bern}\{s(a|X)\}$, where we write $s(a|x) = S(\pi(a|x), t)$ for notational simplicity. Consider the entropy of this selection rule, defined as:
\begin{align}
  H(S|X) = - \mathbb{E} \left[ s(a|X) \log_2\{ s(a|X) \} + \{1 - s(a|X)\} \log_2\{ 1 - s(a|X) \} \right]. \label{eqn:entropy}
\end{align}
For example, consider two extremes of $\epsilon$. When $\epsilon \rightarrow \infty$, $s(a|X) = 0.5$; thus, every subject contributes equally to the average, and $H(S|X) = 1$. Meanwhile, when $\epsilon \rightarrow 0$, $s(a|X) = \mathbb{I}(\pi(a|x) > t)$; thus, only subjects with propensity scores above $t$ contribute to the average, and $H(S|X) = 0$. In practice, intermediate values of $\epsilon$ that lead to small, non-zero entropy will be most useful. For example, one could choose $\epsilon$ such that $H(S|X) \approx 0.05$; this would have the same entropy as a coin flip whose probability of heads is 0.9944.

This introduces the following selection procedure. First, estimate the propensity score $\pi(a|x)$. Then, for each choice of $\epsilon$ (e.g., $10^{-1}, 10^{-2}, \dots$), compute the estimated entropy:
\begin{align}
  \widehat{H}(S|X) = - \mathbb{P}_n[\widehat{s}(a|X) \log_2\{\widehat{s}(a|X)\} + \{1 - \widehat{s}(a|X)\} \log_2\{1 - \widehat{s}(a|X)\} ] \label{eqn:estimatedEntropy}
\end{align}
where $\widehat{s}(a|X) = S(\widehat{\pi}(a|X), t)$. Finally, choose the $\epsilon$ such that $\widehat{H}(S|X)$ is close to an intermediate value, e.g. 0.05. While such a choice does not optimize a measure of error (as in the previous subsection), it nonetheless provides a way to interpret $\epsilon$ in practice.

In Section \ref{s:simulations}, we present simulations where, for simplicity, $h$ and $\epsilon$ are fixed. In Section \ref{s:application}, we present an application where we choose $h$ and $\epsilon$ using the above selection procedures.

\section{Simulations} \label{s:simulations}

We consider an illustrative example with one covariate $X$, a continuous outcome $Y$, and a continuous treatment $A$. We generated 1000 datasets, each with $n = 1000$ subjects, as follows. First, the covariate was generated as $X \sim \text{Unif}(0,1)$. Then, the treatment was generated as $A | X \sim N(m(X), 0.2^2)$, where $m(X)$ is shown in Figure \ref{fig:simMeanA}. Finally, the outcome was generated as $Y | X, A \sim N(\mu(X), 0.5^2)$, where $\mu(X)$ is shown in Figure \ref{fig:simMeanY}. The exact specification of $m(X)$ and $\mu(X)$ are in the Supplementary Material. In this example there are positivity violations for some treatment values, motivating trimming, and the propensity score and outcome regression are complex functions of $X$, motivating nonparametric models. We emphasize that although our example only involves one covariate, we present results based on the error rate involved in estimating the outcome regression and propensity score, as we discuss shortly. Thus, our results reflect the performance practitioners should expect if these quantities are estimated at slow rates, e.g., as in high-dimensional settings or with flexible models. Using a single covariate allows us to more easily illustrate where there are positivity violations that adversely affect estimators' performance.

Here, $\mathbb{E}[A|X]$ ranges from 0 to 1. Thus, we consider estimating $\psi(a) = \mathbb{E}[Y(a)]$ for $a \in [0,1]$. If consistency, unconfoundedness, and positivity hold, $\psi(a) = \mathbb{E}[\mu(X,a)]$, where $\mu(X,a) = \mathbb{E}[Y | X, A = a]$. In this example, $\mu(X,a) = \mu(X)$ for all $a$. Thus, $\psi(a)$ and its smoothed counterpart $\psi_h(a) = \mathbb{E}[\int K_h(a_0 - a) \mu(x,a_0) da_0]$ do not vary across treatment values. Figure \ref{fig:simTruePS} displays the true propensity scores $\pi(a|x)$ for $(x,a) \in [0,1]^2$; areas where $\pi(a|x) < 0.1$ are in black. The propensity scores become smaller for $a$ near 0 or 1, such that $\psi_h(a)$ is difficult to estimate for extreme treatment values.

\begin{figure}
  \centering
  \begin{subfigure}[b]{0.32\textwidth}
  \includegraphics[scale=0.3]{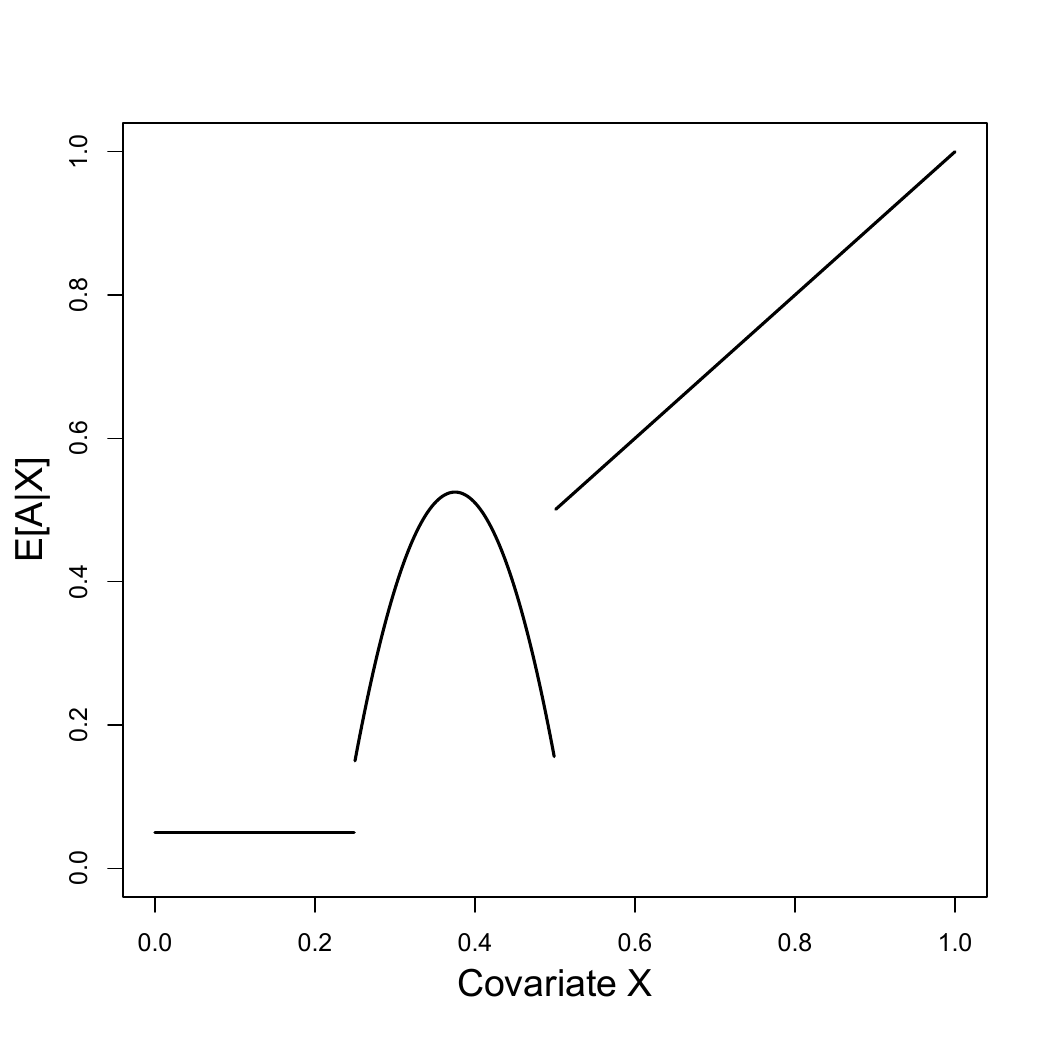}
  \caption{$m(X) = \mathbb{E}[A|X]$, which defines the propensity score.}
  \label{fig:simMeanA}
  \end{subfigure}
  \hfill
  \begin{subfigure}[b]{0.32\textwidth}
  \includegraphics[scale=0.3]{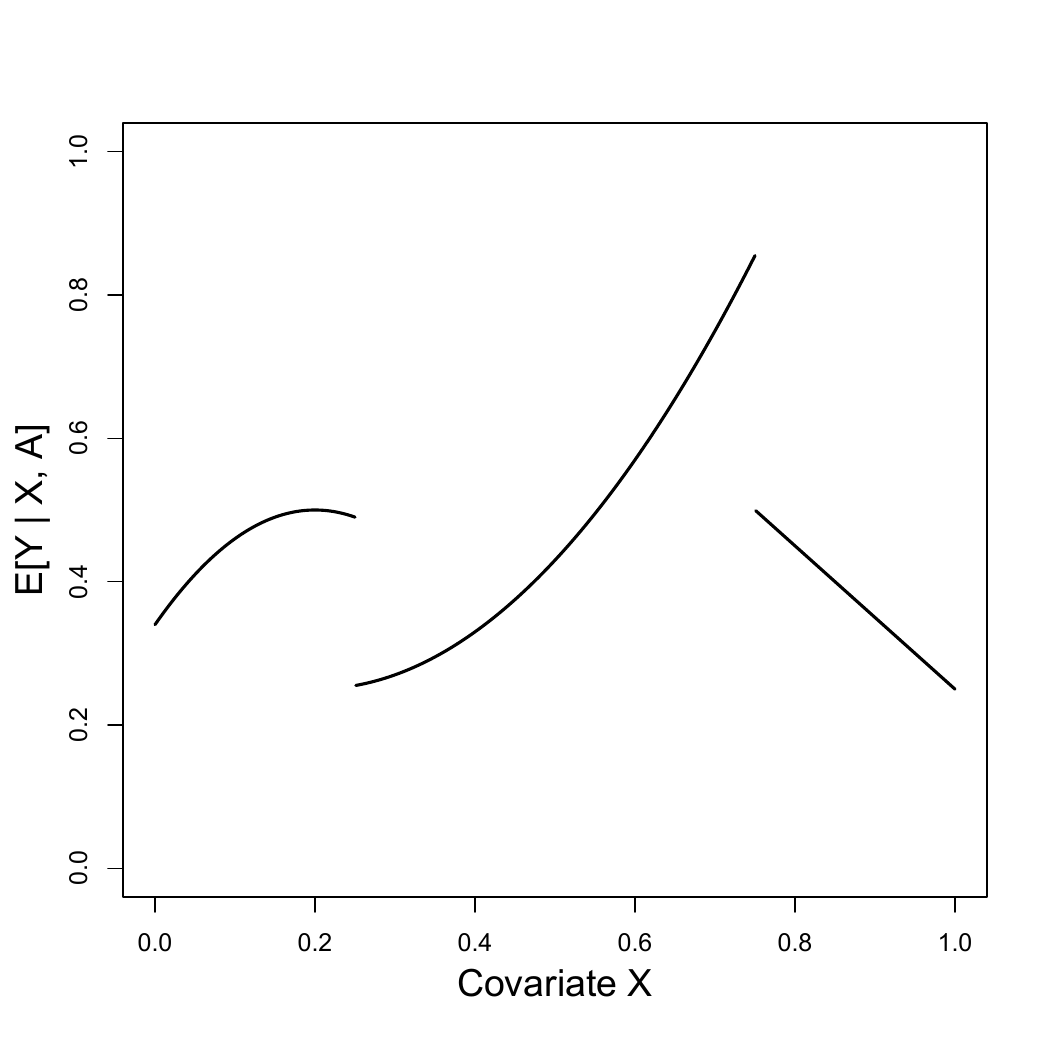}
  \caption{$\mu(X) = \mathbb{E}[Y| X, A]$, which doesn't depend on $A$.}
  \label{fig:simMeanY}
  \end{subfigure}
  \hfill
  \begin{subfigure}[b]{0.32\textwidth}
  \includegraphics[scale=0.45]{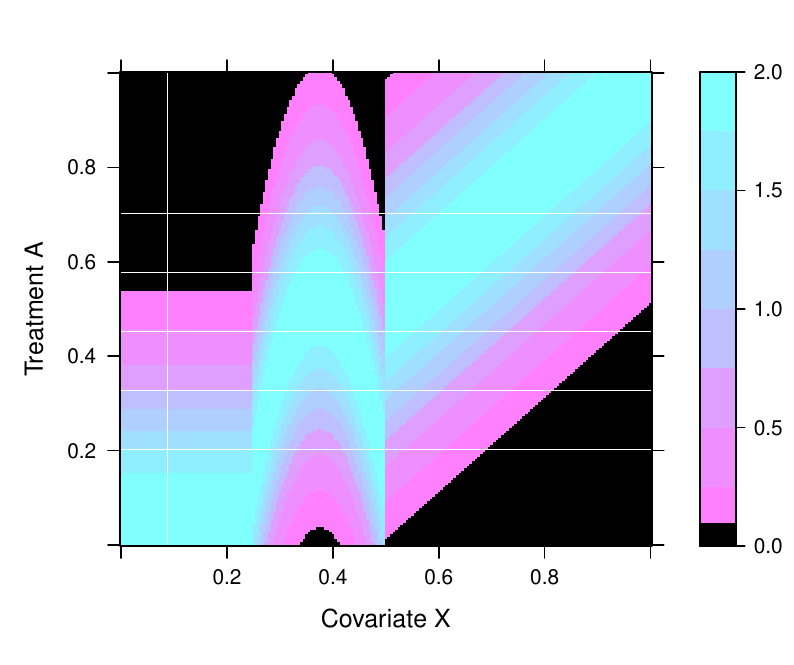}
  \caption{$\pi(a|x)$ for $(a,x) \in [0,1]^2$. Black denotes $\pi(a|x) < 0.1$.}
  \label{fig:simTruePS}
  \end{subfigure}
  \caption{Nuisance functions used in the data-generating process.}
  \label{fig:simFunctions}
\end{figure}

We consider three estimands, shown in Table \ref{tab:ests}: The smoothed ATE (SATE) $\psi_h(a)$, trimmed ATE (TATE) $\psi(a; t)$, and smoothed trimmed ATE (STATE) $\psi_h(a; t)$, which were defined in (\ref{eqn:psih}), (\ref{eqn:simpleIdentification}), and (\ref{eqn:causalEstimandSmoothed}), respectively. Table \ref{tab:ests} also lists the estimators we consider for each estimand. The SATE estimator averages the estimated efficient influence function (EIF) values $\widehat{\varphi}_h(a)$, defined in (\ref{eqn:psihHat}). Meanwhile, for the TATE, we consider two plug-in estimators. The first estimator is $\widehat{\psi}(a; t)$ in (\ref{eqn:plugin}), which averages $\widehat{\mu}(X,a)$ among subjects whose $\widehat{\pi}(a|X) > t$. The second estimator, $\widehat{\psi}^{\text{alt}}(a; t)$, is similar, but replaces $\widehat{\mu}(X,a)$ with $\widehat{\varphi}_h(a)$; thus, it computes a doubly robust estimator on the trimmed sample, which is common in practice for binary treatments \citep{mao2019propensity,li2020comment,zhang2022stable}. Thus, we call this the ``EIF-based trimmed plug-in'' estimator. Finally, the STATE estimator, $\widehat{\psi}_h(a; t)$ in (\ref{eqn:eifRatioTwoEsts}), is the estimator we derived in Section \ref{s:EIFs}.

\begin{table}
  \begin{tabular}{|c|c|c|}
  \hline
    \textbf{Estimand} & \textbf{Definition} & \textbf{Estimators} \\
    \hline
    SATE & $\psi_{h}(a) = \int \int K_h(a_0 - a) \mu(x,a_0) da_0 dP(x)$ & $\widehat{\psi}_h(a) = \mathbb{P}_n\{ \widehat{\varphi}_h(a)\}$ \\
    \hline
    TATE & $\psi(a; t) = \frac{\int \mathbb{I}(\pi(a|x) > t) \mu(x,a) dP(x)}{\int \mathbb{I}(\pi(a|x) > t) dP(x)}$ & \begin{tabular}{@{}c@{}} $\widehat{\psi}(a; t) = \frac{\mathbb{P}_n\{ \mathbb{I}(\widehat{\pi}(a|X) > t) \widehat{\mu}(X,a) \}}{\mathbb{P}_n\{ \mathbb{I}(\widehat{\pi}(a|X) > t) \}}$ \\ $\widehat{\psi}^{\text{alt}}(a; t) = \frac{\mathbb{P}_n[\mathbb{I}\{ \widehat{\pi}(a|X) > t \} \widehat{\varphi}_h(a) ]}{\mathbb{P}_n[\mathbb{I}\{ \widehat{\pi}(a|X) > t \}]}$ \end{tabular} \\
    \hline
    STATE & $\psi_{h}(a; t) = \frac{\int \int K_h(a_0 - a) S(\pi(a_0|x), t) \mu(x,a_0) da_0 dP(x)}{\int \int K_h(a_0 - a) 
S(\pi(a_0|x), t) da_0 dP(x)}$ & $\widehat{\psi}_{h}(a; t) = \frac{\widehat{\psi}^{\text{num}}_{h}(a; t)}{\widehat{\psi}^{\text{den}}_{h}(a; t)} = \frac{\mathbb{P}_n\{ \widehat{\varphi}^{\text{num}}_{h}(a; t) \} }{\mathbb{P}_n\{ \widehat{\varphi}^{\text{den}}_{h}(a; t) \} }$ \\
\hline
  \end{tabular}
  \caption{The estimands, their definitions, and their respective estimators that we consider. These estimands were defined in (\ref{eqn:psih}), (\ref{eqn:simpleIdentification}), and (\ref{eqn:causalEstimandSmoothed}), respectively.}
  \label{tab:ests}
\end{table}

All these estimators depend on estimates $\widehat{\pi}(a|x)$ and $\widehat{\mu}(x,a)$. To remain agnostic to the type of estimators one may use, for each dataset we simulated estimators as, for each $a$,\footnote{To compute the causal effect estimators for any given $a$, we simulated $\widehat{\pi}(A|X)$, $\widehat{\pi}(a|X)$, and $\widehat{\pi}(a_0|X)$ for $a_0 \in \{-0.5, -0.45, \dots, 1.45, 1.50\}$, and analogously for the outcome regression. The last quantity is needed to evaluate integrals over treatment values $a_0$.}
\begin{align*}
  \widehat{\pi}(a|X) &\sim \pi(a | X) + 2 \cdot \text{expit}[ \text{logit}\{\pi(a | X)/2\} + N(n^{-\alpha}, n^{-2 \alpha})] \\
  \widehat{\mu}(X, a) &\sim \mu(X, a) + N(n^{-\alpha}, n^{-2 \alpha}),
\end{align*}
such that the root mean squared error (RMSE) of $\widehat{\pi}(a | X)$ and $\widehat{\mu}(X, a)$ are $O_P(n^{-\alpha})$, and thus we can control the estimators' convergence rate via the rate parameter $\alpha$.\footnote{We used the transformation $\text{logit}\{\pi(a|X)/2\}$ to place propensity scores on the real line, such that Normal errors are sensible (in this case, $\max_{a,x} \pi(a|x) \leq 2$, such that $\pi(a|x)/2 \in [0,1]$ for all $a$ and $x$).} This follows previous simulation studies that evaluate causal effect estimators when nuisance functions are estimated at different rates \citep{kennedy2023towards,zeng2023efficient,mcclean2024nonparametric}. We considered convergence rates $\alpha \in \{0.1, 0.2, 0.3, 0.4, 0.5\}$. Setting $\alpha = 0.5$ corresponds to the rate one would expect from a well-specified parametric model without positivity violations, whereas $\alpha < 0.5$ corresponds to what one would expect from a nonparametric model, or more generally settings with high-dimensional covariates \citep[Chapter 4]{wasserman2006all} or positivity violations \citep{khan2010irregular}. In this case, because the nuisance functions are complex and there are positivity violations, $\alpha = 0.5$ is very unrealistic.

For these estimands and estimators, we consider a Gaussian kernel for $K_h(\cdot)$ with $h = 0.1$ and smoothed indicator $S(\pi(a|x),t) = \Phi_{\epsilon}\{\pi(a|x) - t\}$ with $\epsilon = 0.01$. We first consider the threshold $t = 0.1$, and then consider estimating $t$ as a quantile of the propensity score.

\subsection{Results When the Threshold is Fixed to $t = 0.1$} \label{ss:simFixedT}

On each dataset, we implemented the four estimators for treatment values $a \in \{0, 0.05, \dots, 0.95, 1\}$. Following previous simulation studies for estimating causal effects with continuous treatments \citep{kennedy2017non,wu2024matching}, for a given estimand $\theta(a)$ and estimator $\widehat{\theta}_s(a)$ on simulated dataset $s$, we computed its root mean squared error (RMSE), averaged across the distribution of $A$:
\begin{align*}
  \text{RMSE} &= \int \left( \frac{1}{1000} \sum_{s=1}^{1000} \{\widehat{\theta}_s(a) - \theta(a)\}^2 \right)^{1/2} dP(a) .
\end{align*}
Figure \ref{fig:rmse} displays the RMSE for each estimator for different convergence rates. Our trimmed doubly robust estimator $\widehat{\psi}_h(a; t)$ and the EIF-based trimmed plug-in estimator $\widehat{\psi}^{\text{alt}}(a; t)$ exhibit low RMSE. Meanwhile, the RMSE is high for the doubly robust estimator $\widehat{\psi}_h(a)$, due to small propensity scores inflating its variance. Finally, the plug-in estimator $\widehat{\psi}(a; t)$ inherits the convergence rate of nuisance function estimation, and thus also exhibits high RMSE unless the nuisance functions are estimated at an $\alpha = 0.5$ parametric rate.

\begin{figure}
  \centering
  \begin{subfigure}[b]{0.475\textwidth}
    \includegraphics[scale=0.4]{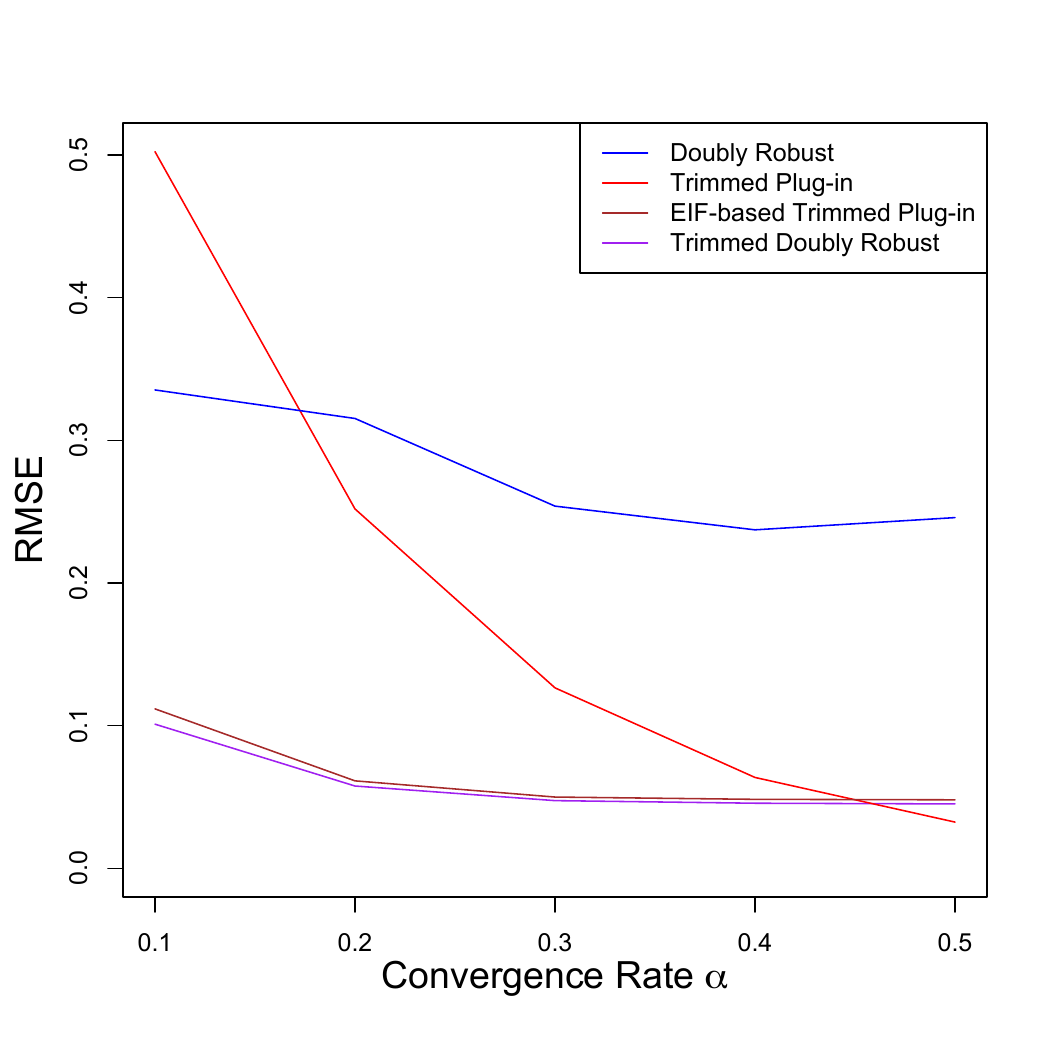}
    \caption{Estimators' average RMSE for $t = 0.1$.}
    \label{fig:rmse}
  \end{subfigure}
  \begin{subfigure}[b]{0.475\textwidth}
    \includegraphics[scale=0.4]{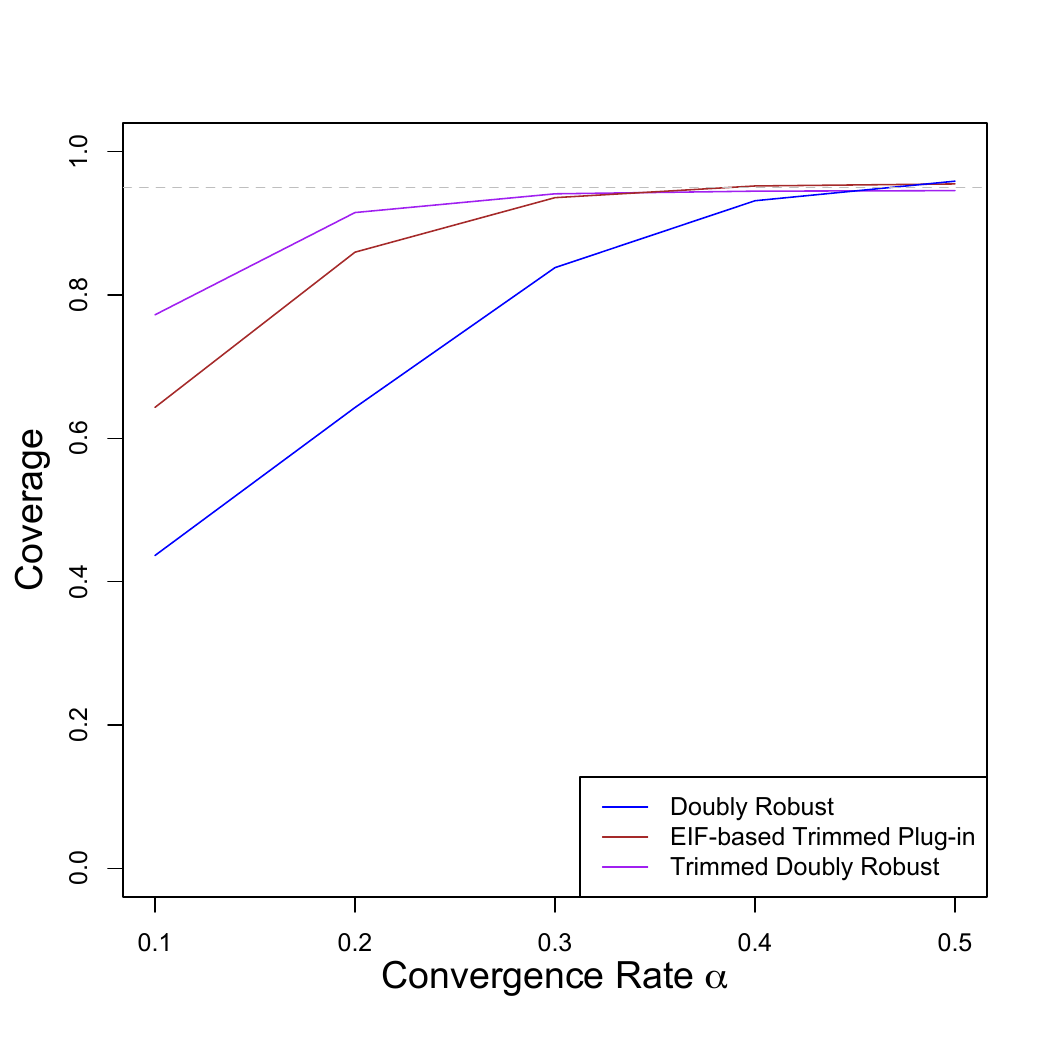}
    \caption{Average coverage for $t = 0.1$.}
    \label{fig:coverage}
  \end{subfigure}
  \begin{subfigure}[b]{0.475\textwidth} 
  \centering
  \includegraphics[scale=0.425]{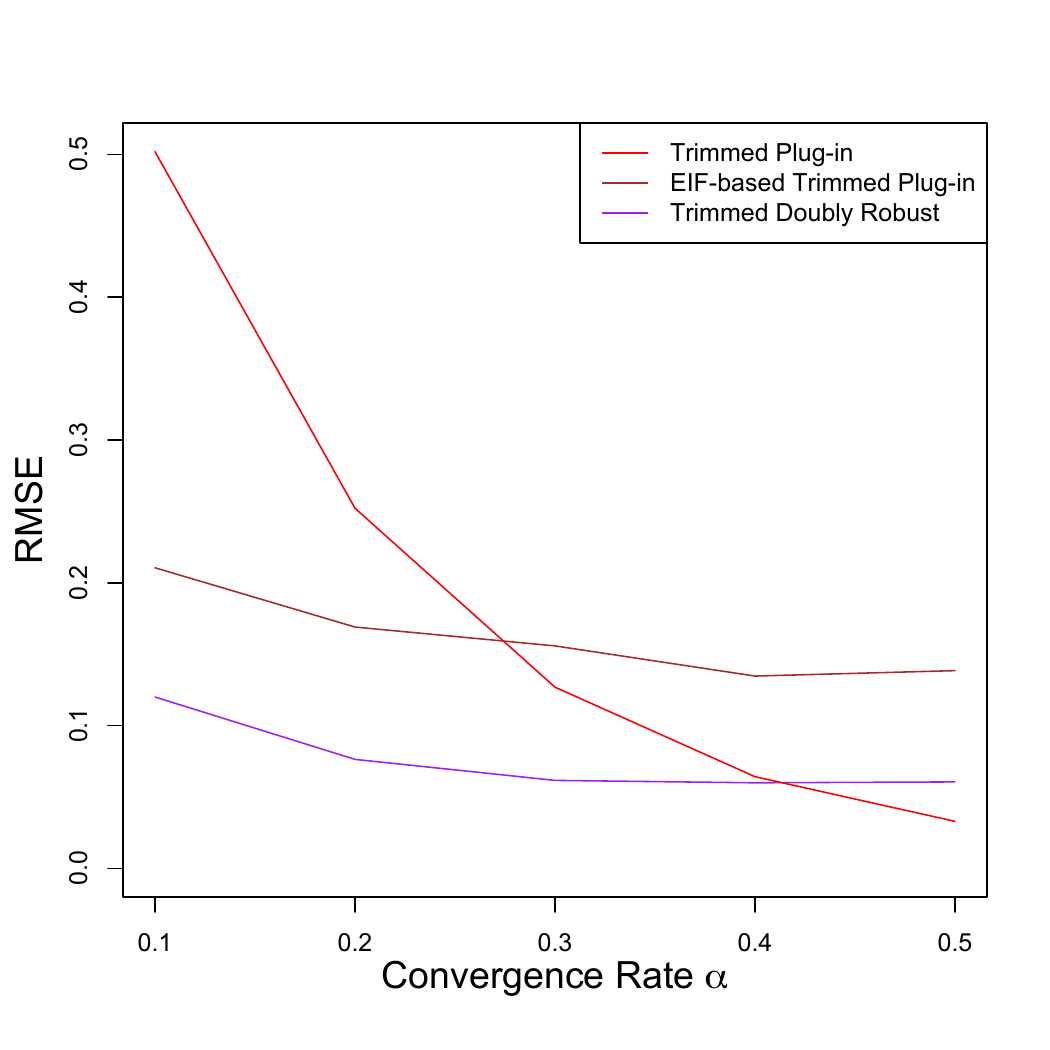}
  \caption{Estimators' average RMSE for estimated $t$.}
  \label{fig:simRMSEEstT}
  \end{subfigure}
  \begin{subfigure}[b]{0.475\textwidth} 
  \centering
  \includegraphics[scale=0.425]{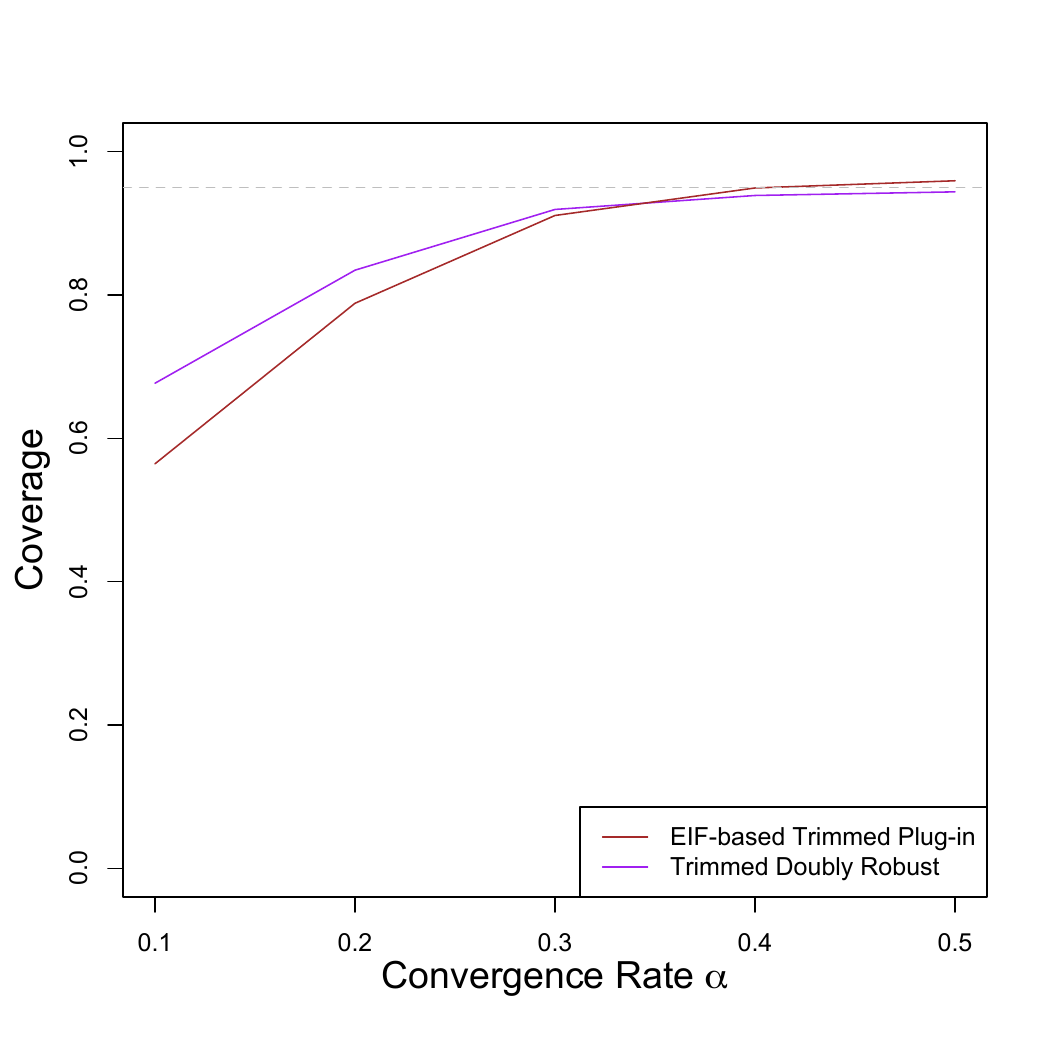}
  \caption{Average coverage for estimated $t$.}
  \label{fig:simCoverageEstT}
  \end{subfigure}
  \caption{Estimators' average RMSE (left) and average coverage of 95\% CIs (right) across convergence rates $\alpha \in \{0.1,0.2,0.3,0.4,0.5\}$, averaged across treatment values $a \in \{0, 0.05,\dots, 0.95, 1\}$, for $t = 0.1$ (top) and estimated $t$ with $\gamma = 0.2$ (bottom).}
\end{figure}

We also considered confidence interval (CI) coverage. The 95\% CI for the doubly robust estimator is $\widehat{\psi}_h(a) \pm z_{0.025} \sqrt{ \text{Var}_n\{ \widehat{\varphi}_h(a) \} / n}$. Meanwhile, for the EIF-based plug-in $\widehat{\psi}^{\text{alt}}(a; t)$, we computed an analogous interval within the trimmed sample. Finally, for our trimmed doubly robust estimator, we computed the interval (\ref{eqn:fixedTCI}). For a given estimand $\theta(a)$ with CI $[\hat{\theta}_{\ell,s}(a), \hat{\theta}_{u,s}(a)]$ for simulated dataset $s$, we computed its average coverage:
\begin{align*}
  \text{coverage} &= \int \left( \frac{1}{1000} \sum_{s=1}^{1000} \mathbb{I} \left\{\theta(a) \in [\hat{\theta}_{\ell,s}(a), \hat{\theta}_{u,s}(a)]\right\} \right)  dP(a).
\end{align*}
We did not compute CIs for the plug-in $\widehat{\psi}(a; t)$---which would require a bootstrap procedure---but we expect coverage to be low, given its poor RMSE. Figure \ref{fig:coverage} displays the coverage results. Coverage is worst for $\widehat{\psi}_h(a)$, due to instability from small propensity scores. Meanwhile, $\widehat{\psi}_h(a; t)$ has coverage close to the nominal level once $\alpha \geq 0.3$, confirming that we can conduct valid inference even when nuisance functions are estimated at nonparametric rates. Coverage for $\widehat{\psi}^{\text{alt}}(a; t)$ is slightly lower, but still preferable to that for $\widehat{\psi}_h(a)$.

\subsection{Results When the Threshold is Estimated} \label{ss:simEstT}

Although the threshold was fixed at $t = 0.1$ above, the amount of trimming varied across treatment values: Figure \ref{fig:simTruePS} suggests that approximately 50\% of subjects were trimmed for extreme treatment values close to 0 or 1, whereas almost no subjects were trimmed for moderate values. Instead, researchers may set $t$ such that a fixed proportion of subjects are trimmed. A common choice is the $\gamma$-quantile of the estimated propensity scores, $\widehat{t}_{\text{pi}} = \widehat{F}_a^{-1}(\gamma)$, where $\widehat{F}_a^{-1}(\gamma)$ is the inverse empirical CDF of $\widehat{\pi}(a|X)$. Thus, $\widehat{t}_{\text{pi}}$ ensures the denominator of the plug-in estimators $\widehat{\psi}(a; t)$ and $\widehat{\psi}^{\text{alt}}(a; t)$ are equal to $1-\gamma$. Similarly, in Section \ref{ss:estimatedThreshold} we studied an estimator, $\widehat{t}$ in (\ref{eqn:estimatedT}), that ensures the denominator of $\widehat{\psi}_h(a; t)$ is equal to $1-\gamma$. We consider these choices for $t$ and the properties of the resulting estimators.

We consider $\gamma = 0.2$, such that 20\% of subjects are trimmed. We simulated $\widehat{\pi}(a|X)$ and $\widehat{\mu}(X, a)$ as in the previous subsection, and then computed $\widehat{t}_{\text{pi}}$ accordingly. Meanwhile, to compute $\widehat{t}$, we computed $\widehat{\psi}_h^{\text{den}}(a; t)$ for $t \in [0, 0.5]$ in increments of 0.005 and selected the smallest $t$ such that $\widehat{\psi}_h^{\text{den}}(a; t) \leq 1-\gamma$. These threshold estimators target true thresholds $t_0^{\text{pi}}$ and $t_0$, respectively. Thus, the TATE and STATE estimands are $\psi(a; t_0^{\text{pi}})$ and $\psi_h(a; t_0)$.

Figure \ref{fig:simRMSEEstT} displays the RMSE for $\widehat{\psi}(a; \widehat{t}_{\text{pi}})$, $\widehat{\psi}^{\text{alt}}(a; \widehat{t}_{\text{pi}})$, and $\widehat{\psi}_h(a; \widehat{t})$. Again our estimator performs well across convergence rates, whereas the plug-in $\widehat{\psi}(a; \widehat{t}_{\text{pi}})$ performs poorly for slow rates. Meanwhile, the EIF-based plug-in $\widehat{\psi}^{\text{alt}}(a; \widehat{t}_{\text{pi}})$ performs worse than our estimator, whereas it performed similarly in Section \ref{ss:simFixedT}. The reason is that, compared to Section \ref{ss:simFixedT}, here we are performing a more moderate amount of trimming for extreme treatment values with positivity violations. Thus, small propensity scores still remain when estimating causal effects, which can adversely affect the doubly robust estimator involved in $\widehat{\psi}^{\text{alt}}(a; \widehat{t}_{\text{pi}})$.

Finally, Figure \ref{fig:simCoverageEstT} displays the coverage results. For the trimmed doubly robust estimator, we used our interval (\ref{eqn:estTCI}) from Section \ref{ss:estimatedThreshold}. For the EIF-based trimmed plug-in, we computed the typical CI for the doubly robust estimator, but within the trimmed sample. Similar to Section \ref{ss:simFixedT}, the EIF-based trimmed plug-in exhibits slightly less coverage, but coverage for both estimators approaches the nominal level as the convergence rate increases.

In summary, our estimators from Section \ref{s:EIFs} performed well in terms of RMSE and coverage whether $t$ was fixed or estimated. Meanwhile, the typical doubly robust estimator performed poorly when propensity scores were small, and the trimmed plug-in estimator performed poorly when nuisance functions' convergence rate was slow. Finally, computing a doubly robust estimator on a trimmed sample tended to perform well, but could still be susceptible to small propensity scores if trimming was not extensive.

\section{The Effect of Smoking on Medical Expenditure} \label{s:application}

In order to demonstrate our trimming approach with a real application, we consider data from the 1987 National Medical Expenditure Survey (NMES), which has been used to study the causal impact of smoking on medical expenses among people who smoke \citep{johnson2003disease,imai2004causal,zhao2020propensity}. In these works, the outcome variable is log annual medical expenditure (in United States dollars). Meanwhile, the treatment variable is a continuous measure of smoking use, defined as
\begin{align}
  A = \log \left( 0.05 \times \text{number of cigarettes per day} \times \text{number of years smoked} \right), \label{eqn:packyear}
\end{align}
where 0.05 is used to measure packs of cigarettes. Finally, the covariates $X$ include: categorical measures of gender, race, marriage status, education, geographic region, income, and seatbelt use; and quantitative measures of age when surveyed, age when the subject started smoking, and their quadratic transformations. Following previous works, we focus on the subset of observations without missing data whose medical expenditure is greater than zero.\footnote{\cite{johnson2003disease} found that addressing missing data via multiple imputation did not notably change results. Thus, \cite{imai2004causal} and \cite{zhao2020propensity} conducted complete-case analyses. Meanwhile, \cite{zhao2020propensity} noted that accounting for subjects with zero expenditure would require modeling the probability $P(Y > 0 | A, X)$ and distribution $p(Y | A, X, Y > 0)$. Like them, we only consider this second model to focus on the challenges of dose-response estimation.} This results in a final set of 8263 subjects. In our analyses we also incorporate subjects' survey-sampling weights provided within the NMES dataset. Replication materials for our analyses can be found at \href{https://github.com/zjbranson/contTreatments}{https://github.com/zjbranson/contTreatments}.

The goal is to determine whether a change in smoking use causes a change in medical expenditure. \cite{imai2004causal} estimated a marginal version of the dose-response curve, and \cite{zhao2020propensity} extended their methodology to estimate the full curve. \cite{zhao2020propensity} noted that if consistency, no unmeasured confounding, and positivity hold, then $\mathbb{E}[Y(a)] = \mathbb{E}[\mathbb{E}\{Y | \pi(A|X), A = a\}]$. \cite{zhao2020propensity} assumed $\pi(A|X)$ was parameterized by $\theta$ and proposed regressing $Y$ onto $\widehat{\theta}$ and $A$. They used linear regression to estimate the propensity score and splines for the second-stage regression. Instead, we consider flexible models for both the propensity score and outcome regression.

We first estimated $m(X) = \mathbb{E}[A | X]$ and $s(X) = \text{Var}(A |X)$ with the \texttt{SuperLearner} R package \citep{van2007super} using an ensemble of adaptive regression splines, generalized linear models, generalized additive models, regression trees, and random forests. Then, we computed $\widehat{\pi}(a|X)$ as the density at $a$ for the distribution $N(\widehat{m}(X), \widehat{s}(X))$. A similar procedure was done in \cite{zhao2020propensity}, but with linear regression for $m(X)$ and a homoskedasticity assumption for $s(X)$. Figure \ref{fig:psEsts} displays the estimated propensity scores across treatment values. Many estimated propensity scores are close to 0 for low treatments, suggesting that this area of the dose-response curve is difficult to estimate precisely. Although \cite{zhao2020propensity} did not focus on positivity violations for this application, this validates their emphasis to ``advise caution in that all available methods [for causal dose-response estimation] can be biased by model misspecification and extrapolation.''

\begin{figure}
\centering
  \includegraphics[scale=0.5]{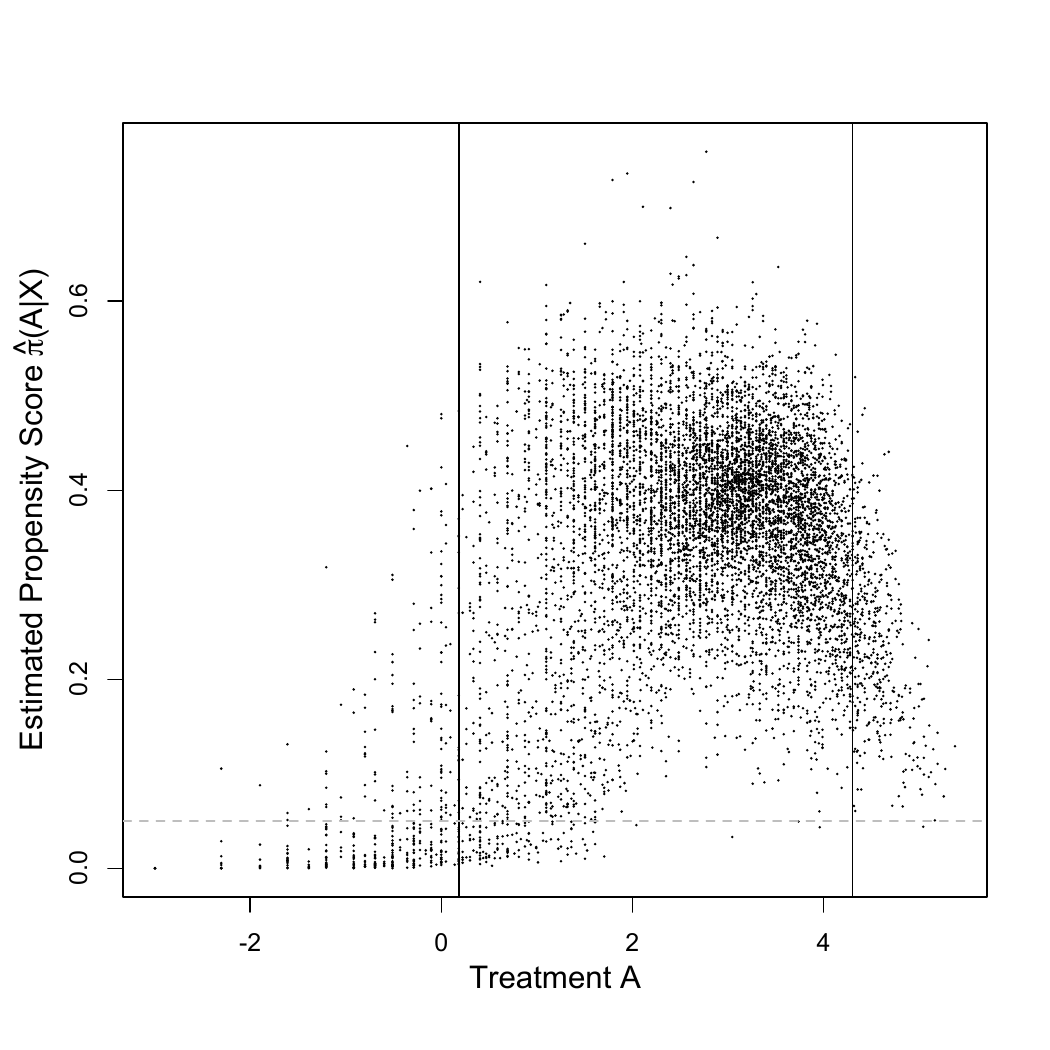}
  \caption{Scatterplot of $A$ and estimated propensity scores $\widehat{\pi}(A|X)$. Horizontal dashed line denotes $t = 0.05$; vertical lines denote the 5\% and 95\% quantiles of $A$.}
  \label{fig:psEsts}
\end{figure}

Similar to \cite{zhao2020propensity}, we estimate causal effects at 10 equally-spaced points between the 5\% and 95\% quantiles of $A$. We consider the doubly robust estimator without trimming, $\widehat{\psi}_h(a)$ in (\ref{eqn:psihHat}), and with trimming in (\ref{eqn:eifRatioTwoEsts}), $\widehat{\psi}_h(a; t)$. For both estimators, we again used the ensemble approach described above to estimate $\mu(x,a) = \mathbb{E}[Y | X = x, A = a]$, and used the cross-fitting procedure described in Section \ref{ss:fixedThreshold} with two splits of equal size. Figure \ref{fig:psEsts} suggests a fixed trimming threshold $t = 0.05$, which will result in some trimming for low treatments. To choose $h$, we implemented the procedure in Section \ref{ss:bandwidth}, where we computed the estimated risk (\ref{eqn:estimatedRisk}) for $h \in \{0.05, 0.06, \dots, 1.99, 2\}$. This resulted in $\widehat{h} = 0.92$ for $\widehat{\psi}_h(a; t)$. We followed an analogous process for $\widehat{\psi}_h(a)$, which resulted in $\widehat{h} = 0.21$. We present results for both bandwidths below. To choose the smoothing parameter $\epsilon$, we implemented the procedure in Section \ref{ss:epsilon}, where we computed the estimated entropy (\ref{eqn:estimatedEntropy}) for $\epsilon = 10^{-c}$ where $c \in \{1, 1.2, \dots, 4.8, 5\}$; we set $\epsilon = 10^{-2}$ because it results in an entropy close to 0.05. Following recommendations from \cite{yang2018asymptotic}, we also implemented our analyses for several choices $\epsilon$, and found that results were quite similar to those presented here. Visualizations of the estimated risk and entropy are in the Supplementary Material.

Figure \ref{fig:appEstH21} displays point estimates and 95\% confidence intervals for $\widehat{\psi}_h(a)$ and $\widehat{\psi}_h(a; t)$ when $h = 0.21$, i.e., the optimal bandwidth for $\widehat{\psi}_h(a)$. The non-trimmed estimator is very unstable in areas where propensity scores are close to 0, whereas the trimmed estimator exhibits notably narrower confidence intervals. Figure \ref{fig:appEstH92} displays results when $h = 0.92$, i.e., the optimal bandwidth for $\widehat{\psi}_h(a; t)$. The non-trimmed estimator is completely unstable, because the larger bandwidth makes small propensity scores influence more areas of the dose-response curve. Meanwhile, results for the trimmed estimator are quite similar for both bandwidths, although $h = 0.92$ leads to slightly tighter confidence intervals.

\begin{figure}
\centering
\begin{subfigure}[b]{0.475\textwidth}
\centering
  \includegraphics[scale=0.4]{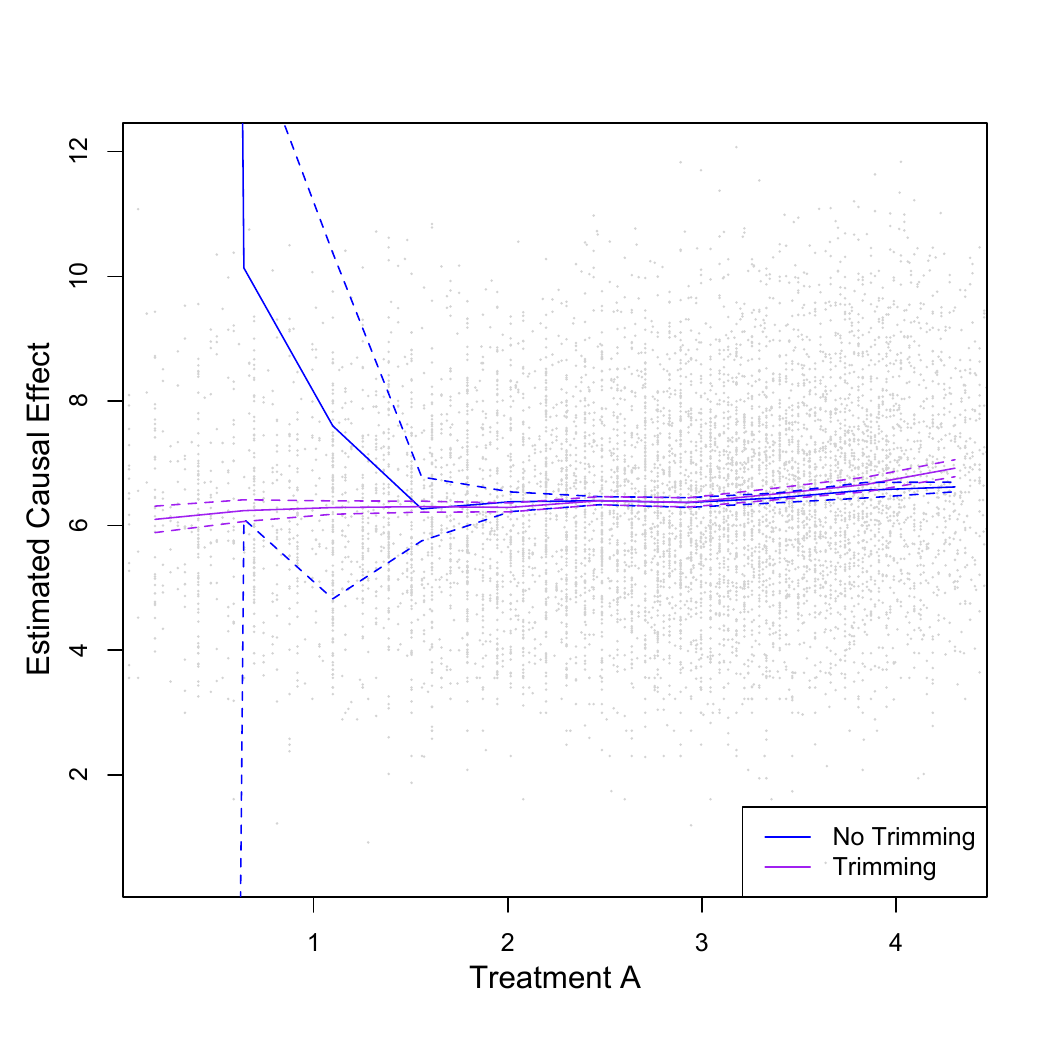}
  \caption{Results when $h = 0.21$.}
  \label{fig:appEstH21}
\end{subfigure}
\hspace{0.1in}
\begin{subfigure}[b]{0.475\textwidth}
\centering
  \includegraphics[scale=0.4]{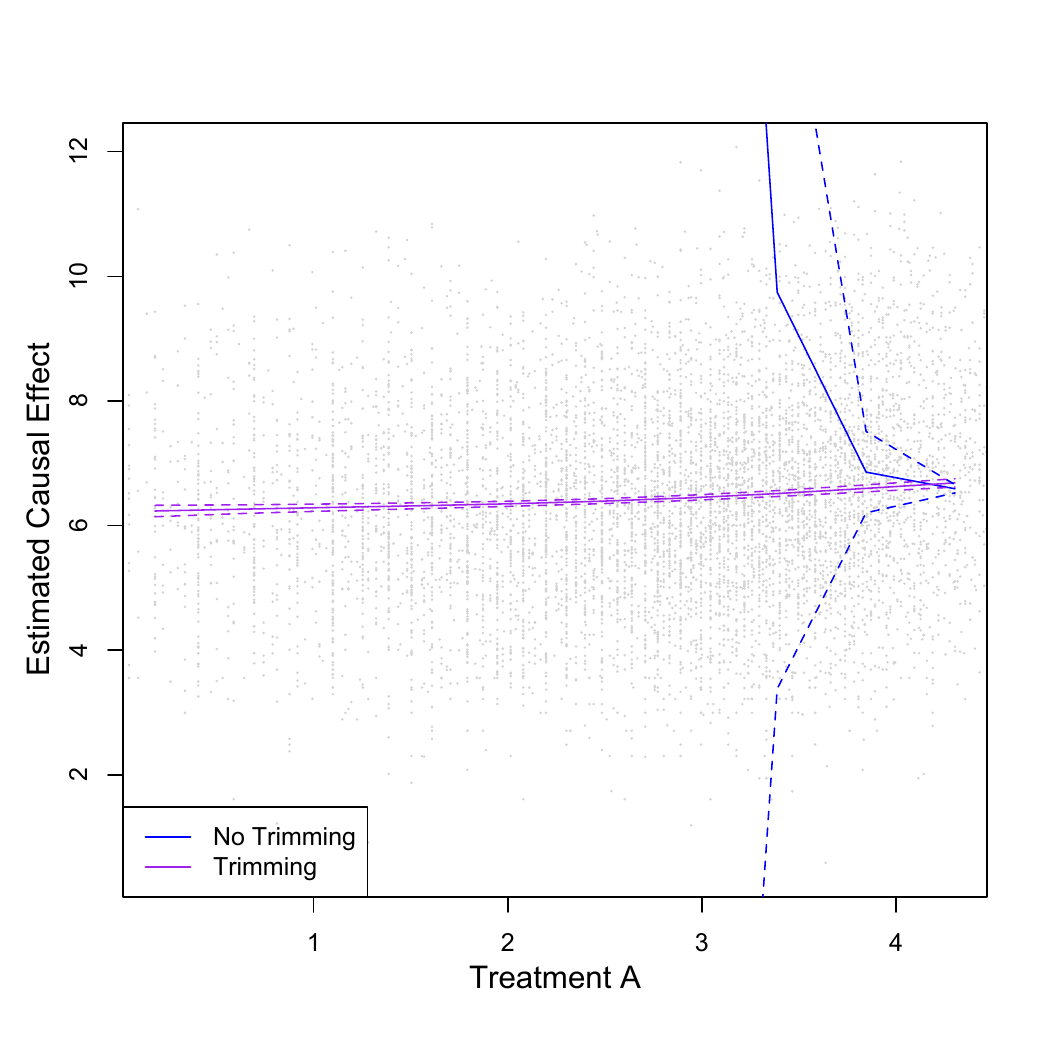}
  \caption{Results when $h = 0.92$.}
  \label{fig:appEstH92}
\end{subfigure}
\begin{subfigure}[b]{0.475\textwidth}
  \centering
  \includegraphics[scale=0.4]{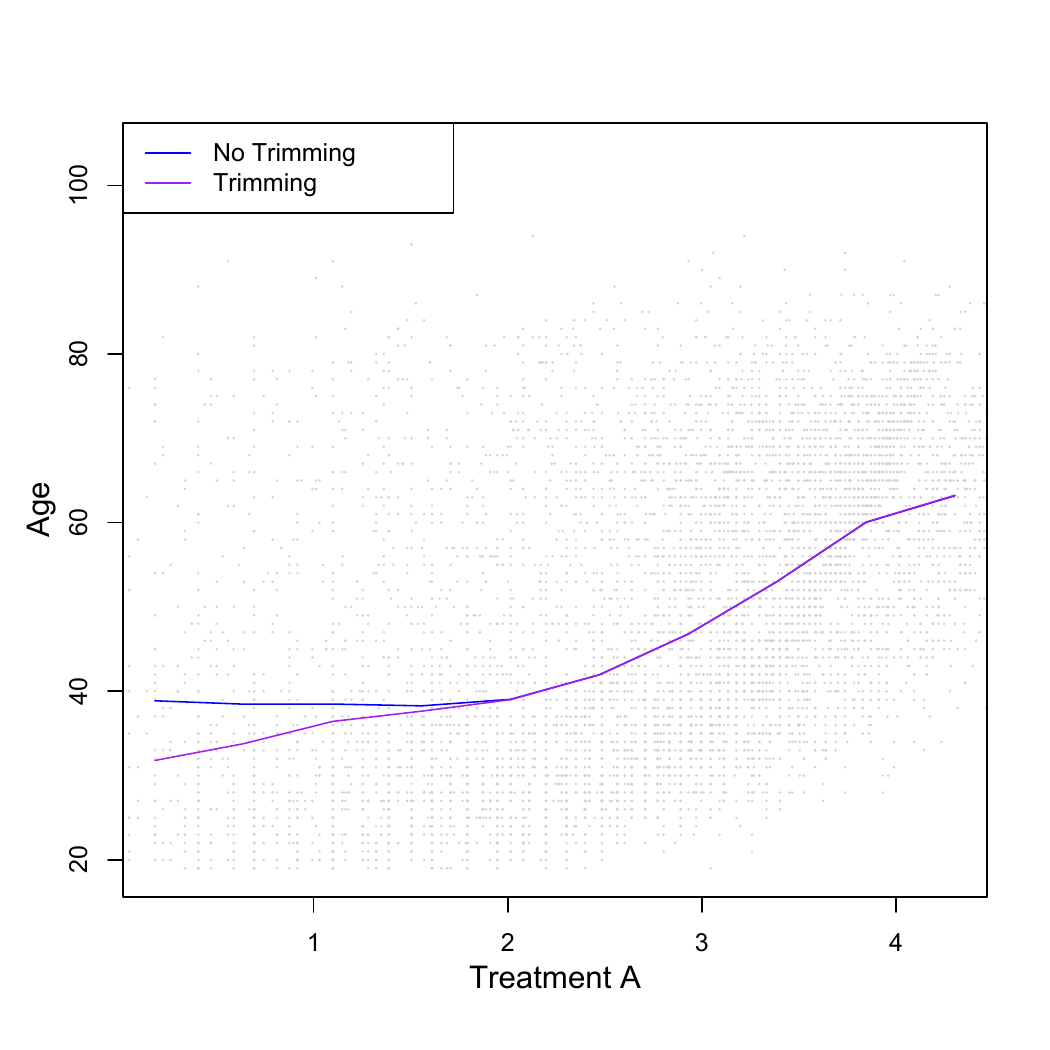}
    \caption{Smoothed age when $h = 0.21$.}
  \label{fig:smoothedAgeH21}
  \end{subfigure}
  \begin{subfigure}[b]{0.475\textwidth}
  \centering
  \includegraphics[scale=0.4]{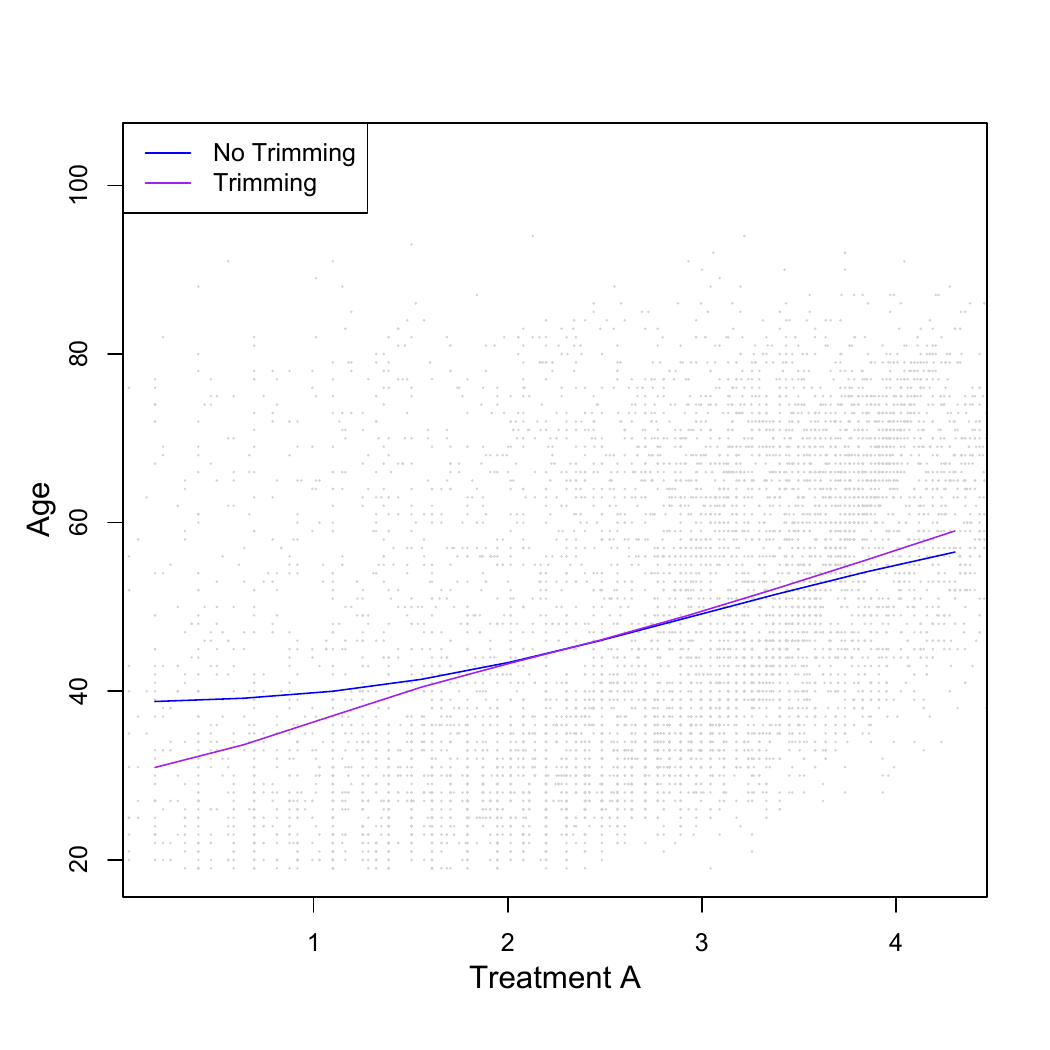}
  \caption{Smoothed age when $h = 0.92$.}
  \label{fig:smoothedAgeH92}
  \end{subfigure}
\caption{Estimated causal effects (top) and smoothed average age (bottom) when $h = 0.21$ (left) and $h = 0.92$ (right) for the doubly robust estimator without trimming (blue) and with trimming (purple). Point estimates are solid lines, and 95\% confidence intervals are dashed lines. Dots denote $(A,Y)$ observations (top) and $(A, X)$ observations (bottom).}
\label{fig:appEst}
\end{figure}

The trimmed estimator suggests that estimated average medical expenditure decreases very slightly as smoking decreases, but plateaus around $a = 1$. For example, two treatment values considered were $a = 4.30$ and $a = 2.01$, which corresponds to $\exp(4.30) \approx 73.70$ versus $\exp(2.01) \approx 7.46$ pack-years. For the trimmed estimator with $h = 0.92$, the corresponding confidence intervals were $(\$748.59, \$855.64)$ and $(\$549.40, \$596.59)$, suggesting a statistically significant decrease, with a point estimate of $\exp\{\widehat{\psi}_h(4.30; t)\} - \exp\{\widehat{\psi}_h(2.01; t)\} = \$227.82$. That said, this estimate (in 1987 US dollars) may be considered modest. Meanwhile, confidence intervals for $\exp\{\widehat{\psi}_h(a; t)\}$ overlap substantially for $a < 2$.

Results using the method from \cite{zhao2020propensity} are similar; indeed, the estimated propensity scores in Figure \ref{fig:psEsts} are similar to those estimated with linear regression, as in \cite{zhao2020propensity}. However, unlike $\widehat{\psi}_h(a)$, their estimator was not sensitive to positivity violations, because they studied a plug-in estimator based on the identification result $\mathbb{E}[Y(a)] = \mathbb{E}[\mathbb{E}\{Y | \pi(A|X), A = a\}]$. In Section \ref{s:simulations}, we saw that plug-in estimators' variance can be insensitive to positivity violations, but bias can be severe when nuisance functions are estimated at slow rates. The fact that our results are similar may suggest that propensity scores could be adequately modeled with parametric models for this dataset.

Figures \ref{fig:appEstH21} and \ref{fig:appEstH92} demonstrate that the trimmed estimator $\widehat{\psi}_h(a; t)$ is more precise than $\widehat{\psi}_h(a)$ for most treatment values, especially those where there are positivity violations. However, these estimators target different estimands: $\widehat{\psi}_h(a; t)$ and $\widehat{\psi}_h(a)$ target causal effects for the trimmed sample and full sample, respectively. Thus, to aid interpretation, it's useful to understand how the trimmed sample differs from the full sample. We computed kernel-smoothed covariate means before and after trimming, defined as
\begin{align*}
  \bar{X}_h(a) = \frac{\sum_{i=1}^n K_h(A_i - a) X_i}{\sum_{i=1}^n K_h(A_i - a)}, \hspace{0.1in} \text{and } \bar{X}_h(a; t) = \frac{\sum_{i=1}^n K_h(A_i - a) S(\widehat{\pi}(a|X_i), t) X_i}{\sum_{i=1}^n K_h(A_i - a) S(\widehat{\pi}(a|X_i), t)}.
\end{align*}
Figures \ref{fig:smoothedAgeH21} and \ref{fig:smoothedAgeH92} displays $\bar{X}_h(a)$ and $\bar{X}_h(a; t)$ for subjects' age when surveyed, which previous works found is the covariate most associated with the outcome. The trimmed sample is much younger for low treatments, and this difference between the trimmed sample and full sample increases as the treatment decreases. This complicates interpretation, because it indicates that the trimmed population changes across treatment values, which is a by-product of trimming separately for each treatment value. Nonetheless, because the discrepancy between the trimmed population and full population is monotonic in age, we can still interpret our analysis as precisely estimating causal effects for a younger sample.

To our knowledge, this complication in interpreting trimmed causal effects for continuous treatments has not been previously acknowledged, perhaps because trimming has focused on binary treatments. One option is to trim subjects whose propensity score is small for any treatment \citep{colangelo2020double}, but when we attempted this for this dataset, this resulted in trimming almost all subjects. In any case, when estimating a dose-response curve after trimming, we recommend visualizing how the trimmed population compares to the full population across treatment values, as we have done here. As we discuss below, one valuable direction for future work is to explore how to define a readily interpretable but also estimable trimmed dose-response curve.

\section{Discussion and Conclusion} \label{s:conclusion}

Although the positivity assumption is commonly invoked to identify and estimate causal effects, violations are also common, thereby destabilizing many estimators. When the treatment is continuous, positivity violations are likely a greater concern. We derived estimators for trimmed causal effects for continuous treatments, which incorporate kernel smoothing and a smoothed trimming indicator. These estimators exhibit doubly robust style guarantees, such that we can conduct valid inference even when the propensity score and outcome regression are estimated at slower-than-parameteric rates. We allowed the trimming threshold to be fixed or estimated as a quantile of the propensity score, and we provided data-driven ways to select smoothing parameters. Through a simulation study and an application, we found we could precisely estimate trimmed causal effects even when some propensity scores are small and nonparametric models are used.

That said, this work has several limitations. First, our estimators only target a trimmed dose-response curve, but one may be interested in the full-population dose-response curve. This discrepancy in estimands can be framed as ``trimming bias,'' and some works have employed bias-correction approaches to target full-population estimands \citep{chaudhuri2014heavy,ma2020robust}. One promising approach is that of \cite{sasaki2022estimation}, who use sieve estimators to model the trimmed causal effect as a function of the proportion of trimmed subjects, and then conduct a bias correction when this proportion is zero. Although this approach requires additional modeling assumptions, one could potentially use it in conjunction with our estimators to target the full-population dose-response curve.

Second, interpreting the dose-response curve after trimming can be challenging, because the trimmed population may change across treatment values in complex ways, especially for high-dimensional covariates. One could create a single trimmed population by trimming any subject whose propensity score is small for any treatment \citep{colangelo2020double}, but this may result in trimming most subjects, as was the case in our application. Alternatively, one could only estimate effects for treatment values where positivity is not violated---i.e., trimming treatment values instead of subjects---but this is unsatisfactory if one is interested in effects for all treatment values. It would be useful to develop tools that diagnose how a trimmed sample differs from a full population across treatments to aid interpretation.

Finally, our trimmed estimand depends on the threshold $t$ and smoothing parameters $h$ and $\epsilon$. We allowed $t$ to be estimated and derived confidence intervals which reflect uncertainty in estimating the threshold. Meanwhile, we provided data-driven procedures for selecting $h$ and $\epsilon$, but our confidence intervals do not reflect uncertainty in estimating these parameters. Researchers likely want $h$ and $\epsilon$ to be small, such that smoothing does not drastically change the estimand, but making $h$ and $\epsilon$ very small can adversely affect the variance of our estimators. It would be useful to derive an optimal way to choose these parameters, while also having inference reflect the uncertainty involved in this choice.

\begin{singlespace}
\bibliographystyle{apa}

\bibliography{trimmingBib}
\end{singlespace}

\newpage

\begin{center}
{\large\bf SUPPLEMENTARY MATERIAL}
\end{center}

\begin{description}

\item Section \ref{s:appendixEIFsHow}: Description of how efficient influence functions are derived throughout proofs in the supplementary materials.

\item Section \ref{as:lemmaNumDenomEIF}: Proof of Theorem \ref{thm:numDenomEIFs}.

\item Section \ref{s:remainderTermsFixedT}: Analysis of the remainder terms involved in the von Mises expansions for Theorem \ref{thm:numDenomEIFs}. This establishes sufficient conditions for the remainder terms to be $o_p(n^{-1/2})$, which are used to prove Proposition \ref{prop:inferenceFixedT}.

\item Section \ref{as:propositionProof}: Proof of Proposition \ref{prop:inferenceFixedT}.

\item Section \ref{as:proofTheorem2}: Proof of Theorem \ref{thm:eifErrorEstT}.

\item Section \ref{s:appendixEIFRatio}: Efficient influence function for the ratio $\psi_h(a; t)$, defined in (\ref{eqn:causalEstimandSmoothed}).

\item Section \ref{s:appendixSimDetails}: Details about the simulation study. Specifically, descriptions of the functions $m(X)$ and $\mu(X)$ visualized in Figure \ref{fig:simFunctions}, and visualizations of dose-response curve estimands for the simulation.

\item Section \ref{as:smoothingParameters}: Visualizations for the estimated risk and entropy used to select the smoothing parameters $h$ and $\epsilon$ for the application in Section \ref{s:application}.

\item Section \ref{s:appendixBinaryTreatment}: Results for discrete and binary treatments. Specifically, efficient influence functions for trimmed causal effects when the treatment is discrete or binary, and simulation results for the binary treatment case.

\end{description}

\section{How Efficient Influence Functions are Derived} \label{s:appendixEIFsHow}

In this Supplementary Material we will derive several efficient influence functions (EIFs). In general, the EIF of a parameter or estimand $\psi$ (e.g., the STATE $\psi_{h}(a; t)$ in (\ref{eqn:causalEstimandSmoothed})) acts as the pathwise derivative of the functional $\psi$ \citep{newey1994asymptotic,tsiatis2006semiparametric}. Importantly, a (centered) EIF $\xi$ satisfies the von Mises expansion \citep{mises1947asymptotic}, for two probability measures $P$ and $\bar{P}$:
\begin{align}
  \psi(\bar{P}) - \psi(P) = \int \xi(\bar{P}) d(\bar{P} - P)(x) + R_2(\bar{P}, P) \label{eqn:vonMises2}
\end{align}
where $R_2(\bar{P}, P)$ is a second-order remainder of a functional Taylor expansion \citep{hines2022demystifying,kennedy2022semiparametric}. Thus, one can first derive a potential EIF $\xi$, and then verify that the von Mises expansion (\ref{eqn:vonMises2}) holds by establishing that the remainder term $R_2(\bar{P}, P)$ is indeed second-order in the probability measures $\bar{P}$ and $P$. One common way to derive potential EIFs is via Gateaux derivatives \citep{ichimura2022influence,hines2022demystifying}. Instead, we will follow ``tricks'' for deriving influence functions, proposed in \cite{kennedy2022semiparametric}, which recognize that influence fuctions---as a pathwise derivative---obey differentiation rules (e.g. product and chain rules). Specifically, the ``tricks'' we use are:
\begin{enumerate}
  \item \textit{Discretization trick}: Momentarily treat the data as discrete when deriving influence functions.
  \item \textit{Product rule}: For two estimands $\psi_1$ and $\psi_2$, the EIF of their product is $\mathbb{IF}(\psi_1 \psi_2) = \mathbb{IF}(\psi_1) \psi_2 + \psi_1 \mathbb{IF}(\psi_2)$.
  \item \textit{Chain rule}: For an estimand $\psi$ and function $f(\psi)$, its EIF is $\mathbb{IF}(f(\psi)) = f'(\psi) \mathbb{IF}(\psi)$, where $f'(\psi)$ is the derivative of $f$ with respect to $\psi$.
\end{enumerate}
See \cite{kennedy2022semiparametric} for details why these ``tricks'' are justified for deriving EIFs. After deriving a proposed EIF $\xi$, we confirm that the von Mises expansion (\ref{eqn:vonMises2}) holds, thereby verifying that our proposed $\xi$ is indeed the EIF for the estimand $\psi$ of interest. Thus, the ``tricks'' are simply useful tools for finding potential EIFs, which are then verified rigorously. Throughout, we will use the notation $\mathbb{IF}(\psi)$ to denote the EIF of $\psi$.

\section{Proof of Theorem \ref{thm:numDenomEIFs}} \label{as:lemmaNumDenomEIF}

First, let

\begin{align*}
  \psi_{h}^{\text{num}}(a; t) &= \int \int K_h(a_0 - a) S(\pi(a_0|x), t) \mu(x,a_0) da_0 dP(x) \\
  \psi_{h}^{\text{den}}(a; t) &= \int \int K_h(a_0 - a) S(\pi(a_0|x), t) da_0 dP(x)
\end{align*}
Our goal is to derive the EIFs for these two quantities. First, note that when $A$ and $X$ are discrete, we have the following influence functions:
\begin{align*}
  \mathbb{IF}\{\mu(x,a)\} &= \frac{\mathbb{I}(X=x, A = a)}{\pi(a|x)p(x)}(Y - \mu(x,a)) \\
  \mathbb{IF}\{S(\pi(a|x), t)\} &= \frac{\partial S(\pi(a|x), t)}{\partial \pi} \mathbb{IF}\{\pi(a|x)\} \\
   &= \frac{\partial S(\pi(a|x), t)}{\partial \pi} \frac{\mathbb{I}(X = x)}{P(X=x)}(\mathbb{I}(A = a) - \pi(a|x)) \\
   \mathbb{IF}\{p(x)\} &= \mathbb{I}(X = x) - p(x)
\end{align*}
First we'll derive the EIF for $\psi_{h}^{\text{den}}(a; t)$. We have the following:
\begin{align*}
  \mathbb{IF}\{\psi_{h}^{\text{den}}(a; t)\} &= \mathbb{IF} \left\{ \sum_x \sum_{a_0} K_h(a_0 - a) S(\pi(a_0|x), t) p(x) \right\} \tag{by discretization trick} \\
  &= \sum_x \sum_{a_0} K_h(a_0 - a) \mathbb{IF} \left\{ S(\pi(a_0|x), t) p(x)\right\} \tag{\text{$K_h(\cdot)$ is constant}} \\
  &= \sum_x \sum_{a_0} K_h(a_0 - a) \left[ \mathbb{IF} \{ S(\pi(a_0|x), t) \} p(x) + S(\pi(a_0 | x)) \mathbb{IF}\{p(x)\} \right] \tag{product rule} \\
  &= \sum_x \sum_{a_0} K_h(a_0 - a) \frac{\partial S(\pi(a_0|x), t)}{\partial \pi} \mathbb{I}(X = x)(\mathbb{I}(A = a_0) - \pi(a_0|x)) \\ &+ \sum_x \sum_{a_0} K_h(a_0 - a) S(\pi(a_0 | x)) ( \mathbb{I}(X = x) - p(x) ) \tag{plugging in IFs} \\
  &= K_h(A - a) \frac{\partial S(\pi(A | X), t) }{\partial \pi} - \sum_{a_0} K_h(a_0 - a) \frac{\partial S(\pi(a_0 | X), t) }{\partial \pi} \pi(a_0|X) \\ &+ \sum_{a_0} K_h(a_0 - a) S(\pi(a_0 | X), t) - \psi_{h}^{\text{den}}(a; t) \tag{simplifying} \\
  &= K_h(A - a) \frac{\partial S(\pi(A | X), t) }{\partial \pi} - \int K_h(a_0 - a) \frac{\partial S(\pi(a_0 | X), t) }{\partial \pi} \pi(a_0|X) da_0 \\ &+ \int K_h(a_0 - a) S(\pi(a_0 | X), t) da_0 - \psi_{h}^{\text{den}}(a; t) \tag{sums to integrals}
\end{align*}
The above quantity is the EIF of $\psi_h^{\text{den}}(a; t)$, which we'll verify in the next subsection. Note that, in Theorem \ref{thm:numDenomEIFs}, $\varphi_h^{\text{den}}(a; t) = \mathbb{IF}\{\psiden(a; t)\} + \psiden(a; t)$ is the uncentered EIF.

Now we'll derive the EIF for $\psi_{h}^{\text{num}}(a; t)$. First, note that we have the following, by the discretization trick and product rule:
\begin{align*}
  \mathbb{IF}\{\psi_{h}^{\text{num}}(a; t)\} &= \mathbb{IF} \left\{ \sum_x \sum_{a_0} K_h(a_0 - a) S(\pi(a_0|x), t) \mu(x, a_0) p(x) \right\} \\ 
  &=  \sum_x \sum_{a_0} K_h(a_0 - a) \mathbb{IF}\{\mu(x,a_0)\} S(\pi(a_0|x), t) p(x) \tag{1} \\
  &+ \sum_x \sum_{a_0} K_h(a_0 - a) \mu(x,a_0) \mathbb{IF}\{S(\pi(a_0|x), t)\} p(x) \tag{2} \\
  &+ \sum_x \sum_{a_0} K_h(a_0 - a) \mu(x,a_0)S(\pi(a_0|x), t)\mathbb{IF}\{p(x)\} \tag{3}
\end{align*}
Now we will compute terms (1), (2), and (3). First, for (1), we have:
\begin{align*}
  (1) &= \sum_x \sum_{a_0} K_h(a_0 - a) \frac{\mathbb{I}(X=x, A = a_0)}{\pi(a_0|x)}\{Y - \mu(x,a_0)\} S(\pi(a_0|x), t) \\
  &= K_h(A - a) \frac{\{Y - \mu(X,A)\} S(\pi(A|X), t)}{\pi(A|X)}
\end{align*}
Meanwhile, for (2) and (3), note that the same terms appeared in our work for the EIF of $\psi_{h}^{\text{den}}(a; t)$ (where now we multiply these two terms by $\mu(x,a_0)$). Thus, for (2), we have:
\begin{align*}
  (2) =  K_h(A - a) \mu(X,A) \frac{\partial S(\pi(A | X), t) }{\partial \pi} - \int K_h(a_0 - a) \mu(X,a_0) \frac{\partial S(\pi(a_0|x), t)}{\partial \pi} \pi(a_0|X)) da_0
\end{align*}
and for (3), we have:
\begin{align*}
  (3) = \int K_h(a_0 - a) S(\pi(a_0 | X), t) \mu(X, a_0) da_0 - \psi_{h}^{\text{num}}(a; t)
\end{align*}
Plugging these all back into our original derivation of $\mathbb{IF}\{\psi_{h}^{\text{num}}(a; t)\}$, we have:
\begin{align*}
  \mathbb{IF}\{\psi_{h}^{\text{num}}(a; t)\} &= K_h(A - a) \frac{\{Y - \mu(X,A)\} S(\pi(A|X), t)}{\pi(A|X)} \\
  &+ K_h(A - a) \mu(X,A) \frac{\partial S(\pi(A|X), t)}{\partial \pi} - \int K_h(a_0 - a) \mu(X,a_0) \frac{\partial S(\pi(a_0|X), t)}{\partial \pi} \pi(a_0|X)) da_0 \\
  &+ \int K_h(a_0 - a) S(\pi(a_0 | X), t) \mu(X, a_0) da_0 - \psi_{h}^{\text{num}}(a; t)
\end{align*}
The above quantity is the EIF of $\psi_h^{\text{num}}(a; t)$. Again, note that, in Theorem \ref{thm:numDenomEIFs}, $\varphi_h^{\text{num}}(a; t) = \mathbb{IF}\{\psinum(a; t)\} + \psinum(a; t)$ is the uncentered EIF.

This completes the proof. To formally verify that the above quantities are indeed EIFs, we must establish their second-order bias, which we do in the following section.

\section{Analyzing the Remainder Terms of Theorem \ref{thm:numDenomEIFs}} \label{s:remainderTermsFixedT}

First we must extend our notation such it can incorporate probability distributions $P$ and $\bar{P}$. Our estimands, for a particular probability distribution $P$, are defined as:
\begin{align*}
  \psi_{h}^{\text{num}}(a; t; P) &= \int \int K_h(a_0 - a) S(\pi(a_0|x), t) \mu(x,a_0) da_0 dP(x) \\
  \psi_{h}^{\text{den}}(a; t; P) &= \int \int K_h(a_0 - a) S(\pi(a_0|x), t) da_0 dP(x)
\end{align*}
Note that, implicitly, $\mu$ and $\pi$ are also indexed by $P$. In the previous section (and in Theorem \ref{thm:numDenomEIFs}), we posit that the corresponding (centered) EIFs, indexed by $P$, are:
\begin{align*}
  \xi_{h}^{\text{num}}(a; t; P) &= K_h(A - a) \frac{\{Y - \mu(X,A)\} S(\pi(A|X), t)}{\pi(A|X)} \\
  &+ K_h(A - a) \mu(X,A) \frac{\partial S(\pi(A|X), t)}{\partial \pi} - \int K_h(a_0 - a) \mu(X,a_0) \frac{\partial S(\pi(a_0|X), t)}{\partial \pi} \pi(a_0|X) da_0 \\
  &+ \int K_h(a_0 - a) S(\pi(a_0 | X), t) \mu(X, a_0) da_0 - \psi_{h}^{\text{num}}(a; t; P) \\ \\
  \xi_{h}^{\text{den}}(a; t; P) &= K_h(A - a) \frac{\partial S(\pi(A | X), t) }{\partial \pi} - \int K_h(a_0 - a) \frac{\partial S(\pi(a_0 | X), t) }{\partial \pi} \pi(a_0|X) da_0 \\ &+ \int K_h(a_0 - a) S(\pi(a_0 | X), t) da_0 - \psi_{h}^{\text{den}}(a; t; P)
\end{align*}
Note that in order for the von Mises expansion (\ref{eqn:vonMises2}) to be valid for these quantities, it must be the case that $\mathbb{E}_P[\xi_{h}^{\text{num}}(a; t; P)] = 0$ and $\mathbb{E}_P[\xi_{h}^{\text{den}}(a; t; P)] = 0$. For completeness, we'll first confirm this, and then we will study the remainder terms for these quantities based on the von Mises expansion.

First let's start with $\xi_{h}^{\text{den}}(a; t; P)$. Note that
\begin{align*}
  \mathbb{E}_P \left[ K_h(A - a) \frac{\partial S(\pi(A | X), t) }{\partial \pi} \right] = \int \int K_h(a_0 - a) \frac{\partial S(\pi(a_0 | x), t) }{\partial \pi} \pi(a_0|x) da_0 dP(x)
\end{align*}
Thus, we can readily see that $\mathbb{E}_P[\xi_{h}^{\text{den}}(a; t; P)] = 0$. Meanwhile, by iterated expectations and the definition of $\mu(X,A) = \mathbb{E}[Y | X, A]$, we have
\begin{align*}
  \mathbb{E}_P \left[ K_h(A - a) \frac{\{Y - \mu(X,A)\} S(\pi(A|X), t)}{\pi(A|X)} \right] &= \mathbb{E}_P \left[ \mathbb{E}_P \left[ K_h(A - a) \frac{\{Y - \mu(X,A)\} S(\pi(A|X), t)}{\pi(A|X)} \bigg| X,A \right] \right] \\
  &= \mathbb{E}_P \left[ K_h(A - a) \frac{\{\mu(X,A) - \mu(X,A)\} S(\pi(A|X), t)}{\pi(A|X)} \right] \\
  &= 0
\end{align*}
Thus, we can also readily see that $\mathbb{E}_P[\xi_{h}^{\text{num}}(a; t; P)] = 0$.

\subsection{Second-Order Remainder for the Denominator Estimand}

We have the following lemma, which establishes that the remainder term for the EIF of $\psi_{h}^{\text{den}}(a; t)$ based on the von Mises expansion (\ref{eqn:vonMises2}) is second-order.

\begin{lemma} \label{lemma:secondOrderDenomFixedT}
  Consider the estimand $\psi_{h}^{\text{den}}(a; t)$ and two probability measures $P$ and $\bar{P}$. Let $\pi(a|x)$ and $\bar{\pi}(a|x)$ denote the propensity score under $P$ and $\bar{P}$, respectively. The remainder term is:
\begin{align*}
  R_2^{\text{den}}(\bar{P}, P) &= \psi_{h}^{\text{den}}(a; t; \bar{P}) - \psi_{h}^{\text{den}}(a; t; P) + \mathbb{E}_P\left[ \xi_{h}^{\text{den}}(a; t; \bar{P}) \right] \\
  &= -\int \int K_h(a_0 - a) \left[ \frac{1}{2}\frac{\partial^2 S(\bar{\pi}(a_0 | x), t)}{\partial \bar{\pi}^2} \left\{ \pi(a_0|x) - \bar{\pi}(a_0|x) \right\}^2 + R_3(\pi - \bar{\pi}) \right] da_0 dP(x)
 \end{align*} 
 where $R_3(\bar{\pi} - \pi)$ denotes the higher-order remainder (third-order and above) of the Taylor expansion of $S(\pi(a_0 | x), t)$ around $\bar{\pi}(a_0|x)$.
\end{lemma}

\begin{proof}
  
First, note that the expectation of $\xi_{h}^{\text{den}}(a; t; \bar{P})$ is
 \begin{align*}
  \mathbb{E}_P \left[ \xi_{h}^{\text{den}}(a; t; \bar{P}) \right] &= \mathbb{E}_P \left[ K_h(A - a) \frac{\partial S(\bar{\pi}(A | X), t) }{\partial \bar{\pi}} \right] - \mathbb{E}_P \left[ \int K_h(a_0 - a) \frac{\partial S(\bar{\pi}(a_0 | X), t) }{\partial \bar{\pi}} \bar{\pi}(a_0|X) da_0 \right] \\ &+ \mathbb{E}_P \left[ \int K_h(a_0 - a) S(\bar{\pi}(a_0|X), t) da_0 \right] - \mathbb{E}_P[\psi_{h}^{\text{den}}(a; t; \bar{P})] \\
  &= \int \int K_h(a_0 - a) \frac{\partial S(\bar{\pi}(a_0 | x), t) }{\partial \bar{\pi}} \pi(a_0|x) da_0 dP(x) \\ &- \int \int K_h(a_0 - a) \frac{\partial S(\bar{\pi}(a_0 | x), t) }{\partial \bar{\pi}} \bar{\pi}(a_0|x) da_0 dP(x) \\ &+ \int \int K_h(a_0 - a) S(\bar{\pi}(a_0|x), t) da_0 dP(x) - \psi_{h}^{\text{den}}(a; t; \bar{P}) \\
  &= \int \int K_h(a_0 - a) \frac{\partial S(\bar{\pi}(a_0 | x), t) }{\partial \bar{\pi}} \left\{ \pi(a_0|x) - \bar{\pi}(a_0|x) \right\} da_0 dP(x) \\ &+ \int \int K_h(a_0 - a) S(\bar{\pi}(a_0|x), t) da_0 dP(x) - \psi_{h}^{\text{den}}(a; t; \bar{P})
 \end{align*}
 Plugging this back into the decomposition of $R_2^{\text{den}}(\bar{P},P)$ above, we have
 \begin{align*}
  R_2^{\text{den}}(\bar{P},P) &= - \psi_{h}^{\text{den}}(a; t; P) + \int \int K_h(a_0 - a) \frac{\partial S(\bar{\pi}(a_0 | x), t) }{\partial \bar{\pi}} \left\{ \pi(a_0|x) - \bar{\pi}(a_0|x) \right\} da_0 dP(x) \\ &+ \int \int K_h(a_0 - a) S(\bar{\pi}(a_0|x), t) da_0 dP(x) \\
  &= \int \int K_h(a_0 - a) \frac{\partial S(\bar{\pi}(a_0 | x), t) }{\partial \bar{\pi}} \left\{ \pi(a_0|x) - \bar{\pi}(a_0|x) \right\} da_0 dP(x) \\ &+ \int \int K_h(a_0 - a) S(\bar{\pi}(a_0|x), t) da_0 dP(x) \\ &- \int \int K_h(a_0 - a) S(\pi(a_0|x), t) da_0 dP(x) \\
  &= \int \int K_h(a_0 - a) \frac{\partial S(\bar{\pi}(a_0 | x), t) }{\partial \bar{\pi}} \left\{ \pi(a_0|x) - \bar{\pi}(a_0|x) \right\} da_0 dP(x) \\ &+ \int \int K_h(a_0 - a) \left\{ S(\bar{\pi}(a_0|x), t) - S(\pi(a_0|x), t) \right\} da_0 dP(x)
 \end{align*}
 Now note that we can do a Taylor's expansion of $S(\pi(a_0 | x), t)$ around $\bar{\pi}(a_0|x)$ to obtain:
 \begin{align*}
  S(\pi(a_0 | x), t) &= S(\bar{\pi}(a_0 | x), t) + \frac{\partial S(\bar{\pi}(a_0 | x), t)}{\partial \bar{\pi}} \left\{ \pi(a_0|x) - \bar{\pi}(a_0|x) \right\} \\&+ \frac{1}{2}\frac{\partial^2 S(\bar{\pi}(a_0 | x), t)}{\partial \bar{\pi}^2} \left\{ \pi(a_0|x) - \bar{\pi}(a_0|x) \right\}^2 + R_3(\pi - \bar{\pi})
 \end{align*}
where $R_3(\bar{\pi} - \pi)$ denotes the higher-order remainder of the Taylor expansion of $S(\pi(a_0 | x), t)$ around $\bar{\pi}(a_0|x)$. Thus:
 \begin{align*}
  \frac{\partial S(\bar{\pi}(a_0 | x), t)}{\partial \bar{\pi}} \left\{ \pi(a_0|x) - \bar{\pi}(a_0|x) \right\} &= S(\pi(a_0|x), t) - S(\bar{\pi}(a_0|x), t) \\&- \frac{1}{2}\frac{\partial^2 S(\bar{\pi}(a_0 | x), t)}{\partial \bar{\pi}^2} \left\{ \pi(a_0|x) - \bar{\pi}(a_0|x) \right\}^2 - R_3(\pi - \bar{\pi})
 \end{align*}
Then, plugging this into the above for $R_2^{\text{den}}(\bar{P},P)$, we have:
\begin{align*}
  R_2^{\text{den}}(\bar{P},P) &= -\int \int K_h(a_0 - a) \left[ \frac{1}{2}\frac{\partial^2 S(\bar{\pi}(a_0 | x), t)}{\partial \bar{\pi}^2} \left\{ \pi(a_0|x) - \bar{\pi}(a_0|x) \right\}^2 + R_3(\pi - \bar{\pi}) \right] da_0 dP(x)
\end{align*}
which completes the proof.
\end{proof}

\subsection{Second-Order Remainder for the Numerator Estimand}

We have the following lemma, which establishes that the remainder term for the EIF of $\psi_{h}^{\text{num}}(a; t)$ based on the von Mises expansion (\ref{eqn:vonMises2}) is second-order.

\begin{lemma} \label{lemma:secondOrderNumFixedT}
  Consider the estimand $\varphi_{h}^{\text{num}}(a; t)$ and two probability measures $P$ and $\bar{P}$. Let $\pi(a|x)$ and $\bar{\pi}(a|x)$ denote the propensity score under $P$ and $\bar{P}$, respectively, and let $\mu(x,a)$ and $\bar{\mu}(x,a)$ denote the outcome regression function under $P$ and $\bar{P}$, respectively. The remainder term is:
  \footnotesize{
 \begin{align*}
  R_2^{\text{num}}(\bar{P},P) &= \psi_{h}^{\text{num}}(a; t; \bar{P}) - \psi_{h}^{\text{num}}(a; t; P) + \mathbb{E}_P\left[ \xi_{h}^{\text{num}}(a; t; \bar{P}) \right] \\ &= \int \int K_h(a_0 - a) \left\{ \frac{S(\bar{\pi}(a_0|x),t)}{\bar{\pi}(a_0|x)} - \frac{\partial S(\bar{\pi}(a_0|x),t)}{\partial \bar{\pi}} \right\} \{ \pi(a_0|x) - \bar{\pi}(a_0|x) \} \{\mu(x,a_0) - \bar{\mu}(x,a_0)\} da_0 dP(x) \\ &-\int \int K_h(a_0 - a) \mu(x,a_0) \left[ \frac{1}{2}\frac{\partial^2 S(\bar{\pi}(a_0 | x), t)}{\partial \bar{\pi}^2} \left\{ \pi(a_0|x) - \bar{\pi}(a_0|x) \right\}^2 + R_3(\pi - \bar{\pi}) \right] da_0 dP(x) 
 \end{align*}
 }
 where $R_3(\bar{\pi} - \pi)$ denotes the higher-order remainder (third-order and above) of the Taylor expansion of $S(\pi(a_0 | x), t)$ around $\bar{\pi}(a_0|x)$.
\end{lemma}

\begin{proof}
 First, it's useful to note that the $\psi_{h}^{\text{num}}(a; t; \bar{P})$ will cancel with the same term contained in $\xi_{h}^{\text{num}}(a; t; \bar{P})$, in the same way as the previous remainder term. Thus,
 \begin{align*}
  \mathbb{E}_P \left[ \xi_{h}^{\text{num}}(a; t; \bar{P})\right] + \psi_{h}^{\text{num}}(a; t; \bar{P}) &= \mathbb{E}_P \left[ K_h(A - a) \frac{(Y - \bar{\mu}(X,A)) S(\bar{\pi}(A|X), t)}{\bar{\pi}(A|X)} \right] \\ &+ \mathbb{E}_P \left[ K_h(A - a) \bar{\mu}(X,A) \frac{\partial S(\bar{\pi}(A|X), t)}{\partial \bar{\pi}} \right] \\ &- \mathbb{E}_P \left[ \int K_h(a_0 - a) \bar{\mu}(X,a_0) \frac{\partial S(\bar{\pi}(a_0|X), t)}{\partial \bar{\pi}} \bar{\pi}(a_0|X) da_0  \right] \\ &+ \mathbb{E}_P \left[ \int K_h(a_0 - a) S(\bar{\pi}(a_0, X), t) \bar{\mu}(X, a_0) da_0 \right] \\
  &= \mathbb{E}_P \left[ K_h(A - a) \frac{(Y - \bar{\mu}(X,A)) S(\bar{\pi}(A|X), t)}{\bar{\pi}(A|X)} \right] \\ &+ \int \int K_h(a_0 - a) \bar{\mu}(x, a_0) \frac{\partial S(\bar{\pi}(a_0, x), t)}{\partial \bar{\pi}} \{ \pi(a_0|x) - \bar{\pi}(a_0|x) \} da_0 dP(x)  \\ &+ \mathbb{E}_P \left[ \int K_h(a_0 - a) S(\bar{\pi}(a_0, X), t) \bar{\mu}(X, a_0) da_0 \right]
 \end{align*}
 Now let's look at the first expectation in the above decomposition. Note that, by iterated expectations:
 \begin{align*}
  \mathbb{E}_P \left[ K_h(A - a) \frac{(Y - \bar{\mu}(X,A)) S(\bar{\pi}(A|X), t)}{\bar{\pi}(A|X)} \right] &= \mathbb{E}_P \left[ \mathbb{E}_P \left[ K_h(A - a) \frac{(Y - \bar{\mu}(X,A)) S(\bar{\pi}(A|X), t)}{\bar{\pi}(A|X)} \bigg| X, A \right] \right] \\
  &= \mathbb{E}_P \left[ K_h(A - a) \frac{(\mu(X,A) - \bar{\mu}(X,A)) S(\bar{\pi}(A|X), t)}{\bar{\pi}(A|X)} \right]
 \end{align*}
 Thus, given this decomposition, we have:
 \begin{align*}
  R_2^{\text{num}}(\bar{P}, P) &= - \psi_{h}^{\text{num}}(a; t; P) + \mathbb{E}_P \left[ \xi_{h}^{\text{num}}(a; t; \bar{P}) \right] + \psi_{h}^{\text{num}}(a; t; \bar{P}) \\
  &= -\int \int K_h(a_0 - a) S(\pi(a_0|x), t) \mu(x,a_0) da_0 dP(x) \\ &+ \mathbb{E}_P \left[ K_h(A - a) \frac{(\mu(X,A) - \bar{\mu}(X,A)) S(\bar{\pi}(A|X), t)}{\bar{\pi}(A|X)} \right] \\ &+ \int \int K_h(a_0 - a) \bar{\mu}(x, a_0) \frac{\partial S(\bar{\pi}(a_0, x), t)}{\partial \bar{\pi}} \{ \pi(a_0|x) - \bar{\pi}(a_0|x) \} da_0 dP(x)  \\ &+ \mathbb{E}_P \left[ \int K_h(a_0 - a) S(\bar{\pi}(a_0, X), t) \bar{\mu}(X, a_0) da_0 \right] \\
  &= -\int \int K_h(a_0 - a) S(\pi(a_0|x), t) \mu(x,a_0) da_0 dP(x) \\ &+ \int \int K_h(a_0 - a) \frac{(\mu(x,a_0) - \bar{\mu}(x,a_0)) S(\bar{\pi}(a_0|x), t)}{\bar{\pi}(a_0|x)} \pi(a_0|x) da_0 dP(x) \\ &+ \int \int K_h(a_0 - a) \bar{\mu}(x, a_0) \frac{\partial S(\bar{\pi}(a_0, x), t)}{\partial \bar{\pi}} \{ \pi(a_0|x) - \bar{\pi}(a_0|x) \} da_0 dP(x)  \\ &+ \int \int K_h(a_0 - a) S(\bar{\pi}(a_0, x), t) \bar{\mu}(x, a_0) da_0 dP(x)
 \end{align*}
 Now let's look specifically at the first and fourth terms. We have:
 \begin{align*}
  &\int \int K_h(a_0 - a) S(\bar{\pi}(a_0, x), t) \bar{\mu}(x, a_0) da_0 dP(x) - \int \int K_h(a_0 - a) S(\pi(a_0|x), t) \mu(x,a_0) da_0 dP(x) \\
  &= \int \int K_h(a_0 - a) S(\pi(a_0|x), t) \{ \bar{\mu}(x,a_0) - \mu(x,a_0) \} da_0 dP(x) \\ &+ \int \int K_h(a_0 - a) \bar{\mu}(x,a_0) \{ S(\bar{\pi}(a_0,x),t) - S(\pi(a_0,x),t) \} da_0 dP(x)
 \end{align*}
 Plugging this into the above decomposition, we have:
 \begin{align*}
  R_2^{\text{num}}(\bar{P}, P) &= \int \int K_h(a_0 - a) \frac{(\mu(x,a_0) - \bar{\mu}(x,a_0)) S(\bar{\pi}(a_0|x), t)}{\bar{\pi}(a_0|x)} \pi(a_0|x) da_0 dP(x) \\ &+ \int \int K_h(a_0 - a) \bar{\mu}(x, a_0) \frac{\partial S(\bar{\pi}(a_0, x), t)}{\partial \bar{\pi}} \{ \pi(a_0|x) - \bar{\pi}(a_0|x) \} da_0 dP(x) \\ &+ \int \int K_h(a_0 - a) S(\pi(a_0|x), t) \{ \bar{\mu}(x,a_0) - \mu(x,a_0) \} da_0 dP(x) \\ &+ \int \int K_h(a_0 - a) \bar{\mu}(x,a_0) \{ S(\bar{\pi}(a_0,x),t) - S(\pi(a_0,x),t) \} da_0 dP(x)
 \end{align*}
 Now let's look at the first and third terms. For simplicity, let's ignore the integrals and $K_h(a_0 - a)$ term; furthermore, let's write $\mu, \bar{\mu}$ and $\pi, \bar{\pi}$ for notational simplicity. We have that these first and third terms equal:
 \begin{align*}
  \frac{(\mu - \bar{\mu}) S(\bar{\pi},t)}{\bar{\pi}} \pi - S(\pi,t)(\mu - \bar{\mu}) &= \left( \frac{1}{\bar{\pi}} - \frac{1}{\pi} \right) (\mu - \bar{\mu}) S(\bar{\pi},t) \pi + (\mu - \bar{\mu})S(\bar{\pi},t) - S(\pi,t)(\mu - \bar{\mu}) \\
  &= \left( \frac{1}{\bar{\pi}} - \frac{1}{\pi} \right) (\mu - \bar{\mu}) S(\bar{\pi},t) \pi + (\mu - \bar{\mu}) \{S(\bar{\pi},t) - S(\pi,t)\} \\
  &= (\pi - \bar{\pi}) (\mu - \bar{\mu}) \frac{S(\bar{\pi},t)}{\bar{\pi}} + (\mu - \bar{\mu}) \{S(\bar{\pi},t) - S(\pi,t)\}
 \end{align*}
 Plugging this back into the first and third terms of $R_2^{\text{num}}(\bar{P},P)$ above, we have:
 \begin{align*}
  R_2^{\text{num}}(\bar{P},P) &= \int \int K_h(a_0 - a) \frac{S(\bar{\pi}(a_0|x),t)}{\bar{\pi}(a_0|x)} \{ \pi(a_0|x) - \bar{\pi}(a_0|x) \} \{\mu(x,a_0) - \bar{\mu}(x,a_0)\} da_0 dP(x) \\ &+ \int \int K_h(a_0 - a) \{\mu(x,a_0) - \bar{\mu}(x,a_0)\}\{ S(\bar{\pi}(a_0|x), t) - S(\pi(a_0|x),t) \} da_0 dP(x) \\ &+\int \int K_h(a_0 - a) \bar{\mu}(x, a_0) \frac{\partial S(\bar{\pi}(a_0, x), t)}{\partial \bar{\pi}} \{ \pi(a_0|x) - \bar{\pi}(a_0|x) \} da_0 dP(x) \\ &+ \int \int K_h(a_0 - a) \bar{\mu}(x,a_0) \{ S(\bar{\pi}(a_0,x),t) - S(\pi(a_0,x),t) \} da_0 dP(x)
 \end{align*}
Now let's inspect the third and fourth terms---these are identical to the terms involved in $R_2^{\text{den}}(\bar{P},P)$, except there is a $\bar{\mu}(x,a_0)$ involved in the integral. Thus, following the same procedure as the proof of Lemma \ref{lemma:secondOrderDenomFixedT}, we have:
\begin{align*}
  R_2^{\text{num}}(\bar{P},P) &= \int \int K_h(a_0 - a) \frac{S(\bar{\pi}(a_0|x),t)}{\bar{\pi}(a_0|x)} \{ \pi(a_0|x) - \bar{\pi}(a_0|x) \} \{\mu(x,a_0) - \bar{\mu}(x,a_0)\} da_0 dP(x) \\ &+ \int \int K_h(a_0 - a) \{\mu(x,a_0) - \bar{\mu}(x,a_0)\}\{ S(\bar{\pi}(a_0|x), t) - S(\pi(a_0|x),t) \} da_0 dP(x) \\ &-\int \int K_h(a_0 - a) \bar{\mu}(x,a_0) \left[ \frac{1}{2}\frac{\partial^2 S(\bar{\pi}(a_0 | x), t)}{\partial \bar{\pi}^2} \left\{ \pi(a_0|x) - \bar{\pi}(a_0|x) \right\}^2 + R_3(\pi - \bar{\pi}) \right] da_0 dP(x) 
 \end{align*}
where $R_3(\bar{\pi} - \pi)$ denotes the higher-order remainder of the Taylor expansion of $S(\pi(a_0 | x), t)$ around $\bar{\pi}(a_0|x)$. Again, note that this Taylor expansion is:
 \begin{align*}
  S(\pi(a_0 | x), t) &= S(\bar{\pi}(a_0 | x), t) + \frac{\partial S(\bar{\pi}(a_0 | x), t)}{\partial \bar{\pi}} \left\{ \pi(a_0|x) - \bar{\pi}(a_0|x) \right\} \\&+ \frac{1}{2}\frac{\partial^2 S(\bar{\pi}(a_0 | x), t)}{\partial \bar{\pi}^2} \left\{ \pi(a_0|x) - \bar{\pi}(a_0|x) \right\}^2 + R_3(\pi - \bar{\pi})
 \end{align*}
 Thus:
 \begin{align*}
  S(\bar{\pi}(a_0 | x), t) - S(\pi(a_0 | x), t) &= \frac{\partial S(\bar{\pi}(a_0 | x), t)}{\partial \bar{\pi}} \left\{ \bar{\pi}(a_0|x) - \pi(a_0|x) \right\} \\&- \frac{1}{2}\frac{\partial^2 S(\bar{\pi}(a_0 | x), t)}{\partial \bar{\pi}^2} \left\{ \pi(a_0|x) - \bar{\pi}(a_0|x) \right\}^2 - R_3(\pi - \bar{\pi})
 \end{align*}
 Therefore, plugging this back into $R_2^{\text{num}}(\bar{P},P)$, we have:
 \begin{align*}
  R_2^{\text{num}}(\bar{P},P) &= \int \int K_h(a_0 - a) \frac{S(\bar{\pi}(a_0|x),t)}{\bar{\pi}(a_0|x)} \{ \pi(a_0|x) - \bar{\pi}(a_0|x) \} \{\mu(x,a_0) - \bar{\mu}(x,a_0)\} da_0 dP(x) \\ &+ \int \int K_h(a_0 - a) \frac{\partial S(\bar{\pi}(a_0|x),t)}{\partial \bar{\pi}} \{\mu(x,a_0) - \bar{\mu}(x,a_0)\}\{ \bar{\pi}(a_0|x) - \pi(a_0|x) \} da_0 dP(x) \\ &-\int \int K_h(a_0 - a) \mu(x,a_0) \left[ \frac{1}{2}\frac{\partial^2 S(\bar{\pi}(a_0 | x), t)}{\partial \bar{\pi}^2} \left\{ \pi(a_0|x) - \bar{\pi}(a_0|x) \right\}^2 + R_3(\pi - \bar{\pi}) \right] da_0 dP(x)
 \end{align*}
 Now by collapsing the first two terms, we have:
 \footnotesize{
 \begin{align*}
  R_2^{\text{num}}(\bar{P},P) &= \int \int K_h(a_0 - a) \left\{ \frac{S(\bar{\pi}(a_0|x),t)}{\bar{\pi}(a_0|x)} - \frac{\partial S(\bar{\pi}(a_0|x),t)}{\partial \bar{\pi}} \right\} \{ \pi(a_0|x) - \bar{\pi}(a_0|x) \} \{\mu(x,a_0) - \bar{\mu}(x,a_0)\} da_0 dP(x) \\ &-\int \int K_h(a_0 - a) \mu(x,a_0) \left[ \frac{1}{2}\frac{\partial^2 S(\bar{\pi}(a_0 | x), t)}{\partial \bar{\pi}^2} \left\{ \pi(a_0|x) - \bar{\pi}(a_0|x) \right\}^2 + R_3(\pi - \bar{\pi}) \right] da_0 dP(x) 
 \end{align*}
 }
 \normalsize
 This completes the proof.
\end{proof}

\begin{remark}
\label{remark:numDenomEIFs}
  Above we have shown that the remainder terms of the von Mises expansions for $\psi_h^{\text{num}}(a; t)$ and $\psi_h^{\text{den}}(a; t)$ are:
  \footnotesize{
  \begin{align*}
    R_2^{\text{num}}(\bar{P},P) &= \int \int K_h(a_0 - a) \left\{ \frac{S(\bar{\pi}(a_0|x),t)}{\bar{\pi}(a_0|x)} - \frac{\partial S(\bar{\pi}(a_0|x),t)}{\partial \bar{\pi}} \right\} \{ \pi(a_0|x) - \bar{\pi}(a_0|x) \} \{\mu(x,a_0) - \bar{\mu}(x,a_0)\} da_0 dP(x) \\ &-\int \int K_h(a_0 - a) \mu(x,a_0) \left[ \frac{1}{2}\frac{\partial^2 S(\bar{\pi}(a_0 | x), t)}{\partial \bar{\pi}^2} \left\{ \pi(a_0|x) - \bar{\pi}(a_0|x) \right\}^2 + R_3(\pi - \bar{\pi}) \right] da_0 dP(x), \text{ and} \\ \\
    R_2^{\text{den}}(\bar{P}, P) &= -\int \int K_h(a_0 - a) \left[ \frac{1}{2}\frac{\partial^2 S(\bar{\pi}(a_0 | x), t)}{\partial \bar{\pi}^2} \left\{ \pi(a_0|x) - \bar{\pi}(a_0|x) \right\}^2 + R_3(\pi - \bar{\pi}) \right] da_0 dP(x),
  \end{align*}
  }
  \normalsize
  where $\bar{\pi}$ and $\bar{\mu}$ denote the propensity score and outcome regression under probability measure $\bar{P}$, and $R_3(\bar{\pi} - \pi)$ denotes the third-order remainder of the Taylor expansion of $S(\pi(a_0 | x), t)$ around $\bar{\pi}(a_0|x)$. This demonstrates that sufficient conditions for the remainders of $\psi^{\text{num}}(a; t)$ and $\psi^{\text{den}}(a; t)$ to be $o_P(n^{-1/2})$ are that $\|\widehat{\pi}(a|x)-\pi(a|x)\|^2 = o_P(n^{-1/2})$ and $\|\widehat{\pi}(a|x)-\pi(a|x)\| \cdot \| \widehat{\mu}(x,a) - \mu(x,a)\| = o_P(n^{-1/2})$.
\end{remark}

\section{Proof of Proposition \ref{prop:inferenceFixedT}} \label{as:propositionProof}

To prove Proposition \ref{prop:inferenceFixedT}, it is sufficient to prove that
\begin{align}
  \sqrt{n}(\widehat{\psi}^{\text{num}}_h(a; t) - \psi^{\text{num}}_h(a; t) ) &\rightarrow N(0, \text{Var}\{ \varphi_h^{\text{num}}(a; t) \}), \text{ and} \label{eqn:propProofNum} \\
  \sqrt{n}(\widehat{\psi}^{\text{den}}_h(a; t) - \psi^{\text{den}}_h(a; t) ) &\rightarrow N(0, \text{Var}\{ \varphi_h^{\text{den}}(a; t) \}) \label{eqn:propProofDenom}
\end{align}
because, given (\ref{eqn:propProofNum}) and (\ref{eqn:propProofDenom}), Proposition \ref{prop:inferenceFixedT} immediately follows by the delta method.

First, note that $\widehat{\psi}_h^{\text{num}}(a; t) = \mathbb{P}_n\{ \widehat{\varphi}_h^{\text{num}}(a; t) \}$ and $\psi_h^{\text{num}}(a; t) = \mathbb{P}\{ \varphi_h^{\text{num}}(a; t) \}$. Thus, we have the following decomposition by adding and subtracting terms:
\begin{align}
  \widehat{\psi}^{\text{num}}_h(a; t) - \psi^{\text{num}}_h(a; t) = &\left( \mathbb{P}_n - \mathbb{P} \right) \{ \varphi_h^{\text{num}}(a; t) \} + \left( \mathbb{P}_n - \mathbb{P} \right) \{ \widehat{\varphi}_h^{\text{num}}(a; t) - \varphi_h^{\text{num}}(a; t) \} + R_2^{\text{num}}(\mathbb{P}_n, \mathbb{P}), \label{eqn:propProofDecomp}
\end{align}
where $R_2^{\text{num}}(\mathbb{P}_n, \mathbb{P}) = \mathbb{P} \{ \widehat{\varphi}_h^{\text{num}}(a; t) - \varphi_h^{\text{num}}(a; t) \}$ is the remainder term.
The first term in (\ref{eqn:propProofDecomp}) is the difference between a sample average and the expectation of a fixed function, and thus by the central limit theorem will be asymptotically Normal with mean zero and variance $\text{Var}(\varphi_h^{\text{num}})/n$ up to error $o_P(n^{-1/2})$; this is the same distribution in (\ref{eqn:propProofNum}). Thus, in order to prove (\ref{eqn:propProofNum}), we need to establish that the second and third terms in (\ref{eqn:propProofDecomp}) are $o_P(n^{-1/2})$.

By assumption, the sample used to compute $\widehat{\varphi}_h^{\text{num}}(a;t)$ is independent of the sample that defines the above empirical measure $\mathbb{P}_n\{ \cdot \}$. Thus, for the second term in (\ref{eqn:propProofDecomp}), we can write:
\begin{align*}
  \left( \mathbb{P}_n - \mathbb{P} \right) \{ \widehat{\varphi}_h^{\text{num}}(a; t) - \varphi_h^{\text{num}}(a; t) \} = O_p\left( n^{-1/2} \| \widehat{\varphi}_h^{\text{num}}(a; t) - \varphi_h^{\text{num}}(a; t) \| \right) = o_P(n^{-1/2}),
\end{align*}
where the first equality follows by Lemma 2 in \cite{kennedy2020sharp}, and the second equality follows by our assumption that $\| \widehat{\varphi}_h^{\text{num}}(a; t) - \varphi_h^{\text{num}}(a; t) \| = o_P(1)$.

Lastly, by assumption, $\|\widehat{\pi}(a|X) - \pi(a|X)\|^2 = o_P(n^{-1/2})$ and $\|\widehat{\mu}(X,a) - \mu(X,a)\| \cdot \|\widehat{\pi}(a|X) - \pi(a|X)\| = o_P(n^{-1/2})$ and thus, as discussed in Remark \ref{remark:numDenomEIFs}, $R_2^{\text{num}}(\mathbb{P}_n, \mathbb{P}) = o_P(n^{-1/2})$.

This proves (\ref{eqn:propProofNum}). The proof for (\ref{eqn:propProofDenom}) is exactly the same, with the only difference being that $\varphi_h^{\text{num}}(a; t)$ and $\widehat{\varphi}_h^{\text{num}}(a; t)$ are replaced with $\varphi_h^{\text{den}}(a; t)$ and $\widehat{\varphi}_h^{\text{den}}(a; t)$.

\section{Proof of Theorem \ref{thm:eifErrorEstT}} \label{as:proofTheorem2}

 First, note that $\widehat{\psi}_{h}^{\text{num}}(a; t) = \mathbb{P}_n\{ \varphi_{h}^{\text{num}}(a; t, \widehat{\eta}) \}$ and $\psi_{h}^{\text{num}}(a; t_0) = \mathbb{P}\{ \varphi_{h}^{\text{num}}(a; t_0, \eta_0) \}$. Thus,
 \begin{align}
  \widehat{\psi}^{\text{num}}_{h}(a; \widehat{t}) - \psi^{\text{num}}_{h}(a; t_0) &= \mathbb{P}_n\{ \varphi^{\text{num}}_{h}(a; \widehat{t}, \widehat{\eta}) \} - \mathbb{P}\{ \varphi^{\text{num}}_{h}(a; t_0, \eta_0) \} \notag \\
  &= (\mathbb{P}_n - \mathbb{P})\{ \varphi^{\text{num}}_{h}(a; t_0, \eta_0) \} + (\mathbb{P}_n - \mathbb{P})\{ \varphi^{\text{num}}_{h}(a; \widehat{t}, \widehat{\eta}) - \varphi^{\text{num}}_{h}(a; t_0, \widehat{\eta}) \} \label{eqn:thm2ProofLine1} \\
  &+ (\mathbb{P}_n - \mathbb{P})\{ \varphi^{\text{num}}_{h}(a; t_0, \widehat{\eta}) - \varphi^{\text{num}}_{h}(a; t_0, \eta_0) \} \label{eqn:thm2ProofLine2} \\
  &+ \mathbb{P}\{ \varphi^{\text{num}}_{h}(a; \widehat{t}, \widehat{\eta}) - \varphi^{\text{num}}_{h}(a; t_0, \widehat{\eta}) \} + \mathbb{P}\{ \varphi^{\text{num}}_{h}(a; t_0, \widehat{\eta}) - \varphi^{\text{num}}_{h}(a; t_0, \eta_0) \} \label{eqn:thm2ProofLine3}
 \end{align}
 which follows by adding and subtracting terms. We will analyze each of the three terms above, (\ref{eqn:thm2ProofLine1}), (\ref{eqn:thm2ProofLine2}), and (\ref{eqn:thm2ProofLine3}), in turn.

 The first term in (\ref{eqn:thm2ProofLine1}) is the difference between a sample average and the expectation of a fixed function, and thus by the central limit theorem will be asymptotically Normal. Meanwhile, because $\varphi^{\text{num}}_{h}(a; t, \widehat{\eta})$ is Donsker in $t$ by assumption and $\widehat{t}$ is consistent by assumption, the second term in (\ref{eqn:thm2ProofLine1}) is $o_P(n^{-1/2})$ according to Lemma 19.24 of \cite{van2000asymptotic}.

 Meanwhile, because $\widehat{\eta}$ is consistent and is estimated on an independent sample by assumption, (\ref{eqn:thm2ProofLine2}) is $o_P(n^{-1/2})$ by Lemma 2 of \cite{kennedy2020sharp}.

 Finally, because the map $t \mapsto \psi^{\text{num}}_{h}(a; t, \eta)$ is differentiable at $t_0$ uniformly in $\eta$, and noting that $\mathbb{P}\{ \varphi^{\text{num}}_{h}(a; t_0, \widehat{\eta}) \} = \psi^{\text{num}}_{h}(a; t_0, \widehat{\eta})$, we can use a Taylor expansion to write the first term of (\ref{eqn:thm2ProofLine3}) as
 \begin{align*}
  \mathbb{P}\{ \varphi^{\text{num}}_{h}(a; \widehat{t}, \widehat{\eta}) - \varphi^{\text{num}}_{h}(a; t_0, \widehat{\eta}) \} &= \left\{ \frac{\partial}{\partial t} \psi^{\text{num}}_{h}(a; t_0, \widehat{\eta}) \right\} (\widehat{t} - t_0) + o_P(\|\widehat{t} - t_0\|) \\
  &= \left\{ \frac{\partial}{\partial t} \psi^{\text{num}}_{h}(a; t_0, \eta_0) \right\} (\widehat{t} - t_0) + o_P(\|\widehat{t} - t_0\|)
 \end{align*}
 where the second equality follows by the assumption that $\frac{\partial}{\partial t} \psi^{\text{num}}_{h}(a; t_0, \widehat{\eta}) \xrightarrow{p} \frac{\partial}{\partial t} \psi^{\text{num}}_{h}(a; t_0, \eta_0)$.

 Putting these results together, we have
 \begin{equation}
 \begin{aligned}
  \widehat{\psi}^{\text{num}}_{h}(a; \widehat{t}) - \psi^{\text{num}}_{h}(a; t_0) &= (\mathbb{P}_n - \mathbb{P})\{ \varphi^{\text{num}}_{h}(a; t_0, \eta_0) \} + \left\{ \frac{\partial}{\partial t} \psi^{\text{num}}_{h}(a; t_0, \eta_0) \right\} (\widehat{t} - t_0) \\ &+ O_P(R^{\text{num}}_2) + o_P(\|\widehat{t} - t_0\|) + o_P(n^{-1/2}) \label{eqn:thm2ProofLine4}
 \end{aligned}
 \end{equation}
 where $R^{\text{num}}_2 = \mathbb{P}\{ \varphi^{\text{num}}_{h}(a; t_0, \widehat{\eta}) - \varphi^{\text{num}}_{h}(a; t_0, \eta_0) \} $ is the last term in (\ref{eqn:thm2ProofLine3}).

 Now recall that $\widehat{t}$ is estimated as
 \begin{align*}
  \widehat{t} = \inf \{t:\widehat{\psi}_{h}^{\text{den}}(a; t) \leq 1-\gamma\}
\end{align*}
In other words, $t$ is estimated via the estimating equation
\begin{align*}
  1 - \gamma - \varphi^{\text{den}}_{h}(a; t, \eta) = 0 
\end{align*}
First, note that for the true $t_0$ and $\eta_0$, $\mathbb{P}\{ 1 - \gamma - \varphi^{\text{den}}_{h}(a; t_0, \eta_0) \} = 0$, because $\varphi^{\text{den}}_{h}(a; t, \eta)$ is the uncentered EIF for $\psi^{\text{den}}_{h}(a; \eta, t)$, and the true $t_0$ is defined such that $\psi^{\text{den}}_{h}(a; t_0, \eta_0) = 1 - \gamma$. Meanwhile, $\widehat{t}$ is chosen such that $\mathbb{P}_n\{ 1 - \gamma - \varphi^{\text{den}}_{h}(a; \widehat{t}, \widehat{\eta}) \} = 0$, because $\mathbb{P}_n\{ \varphi^{\text{den}}_{h}(a; t, \widehat{\eta}) \} = \widehat{\psi}^{\text{den}}_{h}(a; t)$. Given this and Assumptions 1-4, we have by Lemma 3 of \cite{kennedy2021semiparametric} that
\begin{align}
  \widehat{t} - t_0 &= \left\{ \frac{\partial}{\partial t} \psi^{\text{den}}_{h}(a; t_0, \eta_0) \right\}^{-1} (\mathbb{P}_n - \mathbb{P}) \{ 1 - \gamma - \varphi^{\text{den}}_{h}(a; t_0, \eta_0) \} + O_P (R_2^{\text{den}}) + o_P(n^{-1/2}) \notag \\
  &= -\left\{ \frac{\partial}{\partial t} \psi^{\text{den}}_{h}(a; t_0, \eta_0) \right\}^{-1} (\mathbb{P}_n - \mathbb{P}) \{ \varphi^{\text{den}}_{h}(a; t_0, \eta_0) \} + O_P (R_2^{\text{den}}) + o_P(n^{-1/2}) \label{eqn:thm2ProofLine5}
\end{align}
where $R_2^{\text{den}} = \mathbb{P} \{ \varphi^{\text{den}}_{h}(a; t_0, \widehat{\eta}) - \varphi^{\text{den}}_{h}(a; t_0, \eta_0) \}$, and furthermore that
\begin{align}
  o_P(\|\widehat{t} - t_0\|) = o_P(n^{-1/2}) + o_P(R_2^{\text{den}}) \label{eqn:thm2ProofLine6}
\end{align}
Plugging this into (\ref{eqn:thm2ProofLine4}), we have
\begin{align*}
  \widehat{\psi}^{\text{num}}_{h}(a; \widehat{t}) - \psi^{\text{num}}_{h}(a; t_0, \eta_0) &= (\mathbb{P}_n - \mathbb{P})\{ \varphi^{\text{num}}_{h}(a; t_0, \eta_0) \} - \frac{\partial \psi^{\text{num}}_{h}(a; t_0, \eta_0)/\partial t}{\partial \psi^{\text{den}}_{h}(a; t_0, \eta_0)/\partial t}  (\mathbb{P}_n - \mathbb{P})\{ \varphi^{\text{den}}_{h}(a; t_0, \eta_0) \} \\ &+ \left\{ \frac{\partial}{\partial t} \psi^{\text{num}}_{h}(a; t_0, \eta_0) \right\} \{O_P (R_2^{\text{den}}) + o_P(n^{-1/2})\} \\ &+ O_P(R^{\text{num}}_2) + o_P(\|\widehat{t} - t_0\|) + o_P(n^{-1/2}) \\ \\
  &= (\mathbb{P}_n - \mathbb{P})\{ \varphi^{\text{num}}_{h}(a; t_0, \eta_0) \} - \frac{\partial \psi^{\text{num}}_{h}(a; t_0, \eta_0)/\partial t}{\partial \psi^{\text{den}}_{h}(a; t_0, \eta_0)/\partial t}  (\mathbb{P}_n - \mathbb{P})\{ \varphi^{\text{den}}_{h}(a; t_0, \eta_0) \} \\ &+ \left\{ \frac{\partial}{\partial t} \psi^{\text{num}}_{h}(a; t_0, \eta_0) \right\} \{O_P (R_2^{\text{den}}) + o_P(n^{-1/2})\} \\ &+ O_P(R^{\text{num}}_2) + o_P(R_2^{\text{den}}) + o_P(n^{-1/2}) \\ \\
  &= (\mathbb{P}_n - \mathbb{P})\{ \varphi^{\text{num}}_{h}(a; t_0, \eta_0) \} - \frac{\partial \psi^{\text{num}}_{h}(a; t_0, \eta_0)/\partial t}{\partial \psi^{\text{den}}_{h}(a; t_0, \eta_0)/\partial t}  (\mathbb{P}_n - \mathbb{P})\{ \varphi^{\text{den}}_{h}(a; t_0, \eta_0) \} \\ &+ O_P(R^{\text{num}}_2) + O_P(R_2^{\text{den}}) + o_P(n^{-1/2})
 \end{align*}
 where the first equality follows by plugging (\ref{eqn:thm2ProofLine5}) into (\ref{eqn:thm2ProofLine4}), the second equality follows by plugging in (\ref{eqn:thm2ProofLine6}), and the third equality follows by assuming $|\frac{\partial}{\partial t} \psi^{\text{num}}_{h}(a; t_0, \eta_0)|$ is bounded and recognizing that $O_p(R_2^{\text{den}}) + o_p(R_2^{\text{den}}) = O_p(R_2^{\text{den}})$.
 This completes the proof.

\section{Efficient Influence Function for the Ratio $\psi_{h}(a; t)$} \label{s:appendixEIFRatio}

The main text derives estimators for STATE $\psi_{h}(a; t)$ defined in (\ref{eqn:causalEstimandSmoothed}). Specifically, we wrote $\psi_{h}(a; t) \equiv \frac{\psi_{h}^{\text{num}}(a; t)}{\psi_{h}^{\text{den}}(a; t)}$, and then derived efficient influence functions (EIFs) for $\psi_{h}^{\text{num}}(a; t)$ and $\psi_{h}^{\text{den}}(a; t)$. An alternative, asymptotically equivalent approach is to instead derive the EIF of $\psi_{h}(a; t)$ itself. The following corollary communicates the EIF $\psi_{h}(a; t)$ when the trimming threshold $t$ is fixed.

\begin{corollary} \label{corollary:trimmedEffectEIF}
  Let $\psi_{h}(a; t) = \psi_{h}^{\text{num}}(a; t)/\psi_{h}^{\text{den}}(a; t)$, where $\psi_{h}^{\text{num}}(a; t)$ and $\psi_{h}^{\text{den}}(a; t)$ are defined in (\ref{eqn:funcNum}) and (\ref{eqn:funcDenom}). Then, if the trimming threshold $t$ is fixed and $\psi_{h}^{\text{den}}(a; t) \geq b > 0$ for some constant $b$, the EIF of $\psi_{h}(a; t)$ is
\begin{align*}
  \varphi_{h}(a; t) &= \frac{1}{\psi_{h}^{\text{den}}(a; t)} \bigg[ K_h(A - a) \frac{\{Y - \mu(X,A)\} S(\pi(A|X), t)}{\pi(A|X)} \\
  &+ K_h(A - a) \frac{\partial S(\pi(A | X), t) }{\partial \pi} \big\{ \mu(X,A) - \psi_{h}(a; t) \big\} \\
  &+ \int K_h(a_0 - a) \left\{ S(\pi(a_0 | X), t) - \frac{\partial S(\pi(a_0 | X), t) }{\partial \pi} \pi(a_0|X) \right\} \big\{ \mu(X,a_0) - \psi_{h}(a; t) \big\} da_0 \bigg]  .
\end{align*}
\end{corollary}

\begin{proof}
First, let $\psi_{h}(a; t) \equiv \frac{\psi_{h}^{\text{num}}(a; t)}{\psi_{h}^{\text{den}}(a; t)}$, where 
\begin{align*}
  \psi_{h}^{\text{num}}(a; t) &= \int \int \mu(x,a_0) S(\pi(a_0|x), t) K_h(a_0 - a) da_0 dP(x) \\
  \psi_{h}^{\text{den}}(a; t) &= \int \int S(\pi(a_0|x), t) K_h(a_0 - a) da_0 dP(x)
\end{align*}
Our goal here is to derive the efficient influence function for $\psi_{h}(a; t)$. First, by the influence function ``product rule,'' we have:
\begin{align}
  \mathbb{IF}(\psi_{h}(a; t)) &= \frac{\mathbb{IF}(\psi_{h}^{\text{num}}(a; t))}{\psi_{h}^{\text{den}}(a; t)} - \left( \frac{\psi_{h}^{\text{num}}(a; t)}{\psi_{h}^{\text{den}}(a; t)} \right) \frac{\mathbb{IF}(\psi_{h}^{\text{den}}(a; t))}{\psi_{h}^{\text{den}}(a; t)} \\
  &= \frac{1}{\psi_{h}^{\text{den}}(a; t)} \left[ \mathbb{IF}(\psi_{h}^{\text{num}}(a; t)) - \psi_{h}(a; t) \mathbb{IF}(\psi_{h}^{\text{den}}(a; t))  \right] \label{eqn:ifFrac}
\end{align}
In Theorem \ref{thm:numDenomEIFs} we found $\mathbb{IF}(\psi_{h}^{\text{num}}(a; t))$ and $\mathbb{IF}(\psi_{h}^{\text{den}}(a; t))$, which can be plugged into (\ref{eqn:ifFrac}) to get the following:
\begin{align*}
  \mathbb{IF}(\psi_{h}(a; t)) &= \frac{1}{\psi_{h}^{\text{den}}(a; t)} \left[ \mathbb{IF}(\psi_{h}^{\text{num}}(a; t)) - \psi_{h}(a; t) \mathbb{IF}(\psi_{h}^{\text{den}}(a; t))  \right] \\
  &= \frac{1}{\psi_{h}^{\text{den}}(a; t)} \bigg[ K_h(A - a) \frac{(Y - \mu(X,A)) S(\pi(A|X), t)}{\pi(A|X)} + K_h(A - a) \mu(X,A) \frac{\partial S(\pi(A | X), t) }{\partial \pi} \\
  &+ \int K_h(a_0 - a) \mu(X, a_0) \{ S(\pi(a_0 | X), t) - \frac{\partial S(\pi(a_0 | X), t) }{\partial \pi} \pi(a_0|X) \} da_0 \\
  &- \psi_{h}(a; t) \big\{ K_h(A - a) \frac{\partial S(\pi(A | X), t) }{\partial \pi} \\ &+ \int K_h(a_0 - a) \{ S(\pi(a_0 | X), t) - \frac{\partial S(\pi(a_0 | X), t) }{\partial \pi} \pi(a_0|X) \} da_0 \big\}  \bigg]
\end{align*}
Note that the $\psi_{h}(a; t)$ terms in $\mathbb{IF}(\psi_{h}^{\text{num}}(a; t))$ and $\mathbb{IF}(\psi_{h}^{\text{den}}(a; t))$ end up cancelling after distributing the outer $\frac{1}{\psi_{h}^{\text{den}}(a; t)}$. We can then combine some terms to obtain:
\begin{align*}
  \mathbb{IF}(\psi_{h}(a; t)) &= \frac{1}{\psi_{h}^{\text{den}}(a; t)} \bigg[ K_h(A - a) \frac{(Y - \mu(X,A)) S(\pi(A|X), t)}{\pi(A|X)} \\
  &+ K_h(A - a) \frac{\partial S(\pi(A | X), t) }{\partial \pi} \big\{ \mu(X,A) - \psi_{h}(a; t) \big\} \\
  &+ \int K_h(a_0 - a) \{ S(\pi(a_0 | X), t) - \frac{\partial S(\pi(a_0 | X), t) }{\partial \pi} \pi(a_0|X) \} \big\{ \mu(X,a_0) - \psi_{h}(a; t) \big\} da_0 \bigg]
\end{align*}
This completes the proof.
\end{proof}

\subsection{Analyzing the von Mises Remainder Term of Corollary \ref{corollary:trimmedEffectEIF}} \label{s:trimmedEffectEIFRemainder}

Note that in the proof of Corollary \ref{corollary:trimmedEffectEIF}, we found that the EIF $\varphi_{h}(a; t)$ is:
\begin{align*}
  \varphi_{h}(a; t) = \frac{\varphi_{h}^{\text{num}}(a; t)}{\psiden(a; t)} - \psi_{h}(a; t) \frac{\varphi_{h}^{\text{den}}(a; t)}{\psiden(a; t)}
\end{align*}
where $\varphi_{h}^{\text{num}}(a; t)$ and $\varphi_{h}^{\text{den}}(a; t)$ are the EIFs provided in Theorem \ref{thm:numDenomEIFs}.

In order for the von Mises expansion (\ref{eqn:vonMises2}) to be valid, we must have that $\mathbb{E}_P[\varphi_{t,h,\epsilon}(a; P)] = 0$. In the proof of Theorem \ref{thm:numDenomEIFs}, we already found that $\mathbb{E}_P[\varphi_{t,h,\epsilon}^{\text{num}}(a; P)] = 0$ and $\mathbb{E}_P[\varphi_{t,h,\epsilon}^{\text{den}}(a; P)] = 0$, and so it immediately follows that $\mathbb{E}_P[\varphi_{t,h,\epsilon}(a; P)] = 0$.

Now we can study the above von Mises expansion for the proposed EIF $\varphi_{h}(a; t)$. First, note that the von Mises expansions for $\varphi_{h}^{\text{num}}(a; t)$ and $\varphi_{h}^{\text{den}}(a; t)$ are
\begin{align*}
  \mathbb{E}_P[\varphi_{t,h,\epsilon}^{\text{num}}(a;\bar{P})] &= R_2^{\text{num}}(\bar{P},P) + \psinum(a; P) - \psinum(a; \bar{P}), \text{ and} \\
  \mathbb{E}_P[\varphi_{t,h,\epsilon}^{\text{den}}(a;\bar{P})] &= R_2^{\text{den}}(\bar{P},P) + \psiden(a; P) - \psiden(a; \bar{P})
\end{align*}
Therefore, we have:
\begin{align*}
  \mathbb{E}_P[\varphi_{t,h,\epsilon}(a; \bar{P})] &= \mathbb{E}_P \left[ \frac{\xi_{h}^{\text{num}}(a; t; \bar{P})}{\psiden(a; \bar{P})} - \psi_{t,h,\epsilon}(a; \bar{P}) \frac{\xi_{h}^{\text{den}}(a; t; \bar{P})}{\psiden(a; \bar{P})} \right] \\
  &= \frac{\mathbb{E}_P[\xi_{h}^{\text{num}}(a; t; \bar{P})]}{\psiden(a; \bar{P})} - \frac{\psi_{t,h,\epsilon}(a; \bar{P})}{\psiden(a; \bar{P})} \mathbb{E}_P[\xi_{h}^{\text{den}}(a; t; \bar{P})] \\
  &= \frac{R_2^{\text{num}}(\bar{P},P) + \psinum(a; P) - \psinum(a; \bar{P})}{\psiden(a; \bar{P})} \\ &- \frac{\psi_{t,h,\epsilon}(a; \bar{P})}{\psiden(a; \bar{P})} \left\{ R_2^{\text{den}}(\bar{P},P) + \psiden(a; P) - \psiden(a; \bar{P}) \right\} \\
  &= \frac{R_2^{\text{num}}(\bar{P},P) + \psinum(a; P)}{\psiden(a; \bar{P})} - \frac{\psi_{t,h,\epsilon}(a; \bar{P})}{\psiden(a; \bar{P})} \left\{ R_2^{\text{den}}(\bar{P},P) + \psiden(a; P) \right\} \\
  &= \frac{R_2^{\text{num}}(\bar{P},P)}{\psiden(a; \bar{P})} + \frac{\psi_{t,h,\epsilon}(a; P) \psiden(a; P)}{\psiden(a; \bar{P})} - \frac{\psi_{t,h,\epsilon}(a; \bar{P}) R_2^{\text{den}}(\bar{P},P)}{\psiden(a; \bar{P})} - \frac{\psi_{t,h,\epsilon}(a; \bar{P}) \psiden(a; P)}{\psiden(a; \bar{P})} \\
  &= \frac{\psiden(a; P)}{\psiden(a; \bar{P})} \{\psi_{t,h,\epsilon}(a; P) - \psi_{t,h,\epsilon}(a; \bar{P})\} + \frac{R_2^{\text{num}}(\bar{P},P)}{\psiden(a; \bar{P})} - \frac{\psi_{t,h,\epsilon}(a; \bar{P}) R_2^{\text{den}}(\bar{P},P)}{\psiden(a; \bar{P})}
\end{align*}
Now plugging this into the von Mises expansion, we have:
\begin{align*}
  R_2(\bar{P},P) &= \psi_{t,h,\epsilon}(a; \bar{P}) - \psi_{t,h,\epsilon}(a; P) + \mathbb{E}_P\left[ \varphi_{t,h,\epsilon}(a; \bar{P}) \right] \\
  &= \psi_{t,h,\epsilon}(a; \bar{P}) - \psi_{t,h,\epsilon}(a; P) + \frac{\psiden(a; P)}{\psiden(a; \bar{P})} \{\psi_{t,h,\epsilon}(a; P) - \psi_{t,h,\epsilon}(a; \bar{P})\} \\ &+ \frac{R_2^{\text{num}}(\bar{P},P)}{\psiden(a; \bar{P})} - \frac{\psi_{t,h,\epsilon}(a; \bar{P}) R_2^{\text{den}}(\bar{P},P)}{\psiden(a; \bar{P})} \\
  &= \left\{ \psi_{t,h,\epsilon}(a; \bar{P}) - \psi_{t,h,\epsilon}(a; P) \right\} \left( 1 - \frac{\psiden(a; P)}{\psiden(a; \bar{P})} \right) + \frac{R_2^{\text{num}}(\bar{P},P)}{\psiden(a; \bar{P})} - \frac{\psi_{t,h,\epsilon}(a; \bar{P}) R_2^{\text{den}}(\bar{P},P)}{\psiden(a; \bar{P})} \\
  &= \frac{1}{\psiden(a;\bar{P})} \bigg[ \left\{ \psi_{t,h,\epsilon}(a; \bar{P}) - \psi_{t,h,\epsilon}(a; P) \right\} \left\{ \psiden(a; \bar{P}) - \psiden(a; P) \right\} \\&+ R_2^{\text{num}}(\bar{P}, P) - \psi_{t,h,\epsilon}(a; \bar{P}) R_2^{\text{den}}(\bar{P}, P) \bigg]
\end{align*}
We already confirmed in our proof of Theorem \ref{thm:numDenomEIFs} that $R_2^{\text{num}}(\bar{P}, P)$ and $R_2^{\text{den}}(\bar{P}, P)$ are second-order. Thus, at this point we have shown that $R_2(\bar{P}, P)$ is also second-order, since the first term above involves the product of errors in estimating $\psi_{h}(a; t)$ and $\psi^{\text{den}}_{h}(a; t)$.

\section{Details of Simulation Study} \label{s:appendixSimDetails}

In the simulation study in Section \ref{s:simulations}, we generated datasets where $A|X \sim N(m(X), 0.2^2)$ and $Y|X, A \sim N(\mu(X), 0.5^2)$. The functions $m(X)$ and $\mu(X)$ are visualized in Figure \ref{fig:simFunctions} and are defined as:
\begin{align*}
  m(X) &= \begin{cases}
    0.05 &\mbox{if $X < 0.25$} \\
    0.15 - 24(0.25 - X)(0.5 - X) &\mbox{if $0.25 \leq X < 0.5$} \\
    X &\mbox{if $X \geq 0.5$}
  \end{cases} \\
  \mu(X) &= \begin{cases}
    0.5 -4(X - 0.2)^2 &\mbox{if $X \leq 0.25$} \\
    0.25 + 2(X - 0.2)^2 &\mbox{if $0.25 < X \leq 0.75$} \\
    1.25 - X &\mbox{if $X > 0.75$}
  \end{cases}
\end{align*}
These were chosen as arbitrary functions that were somewhat complex, thereby motivating nonparametric models, and that were constrained between 0 and 1 for ease of interpretation.

Meanwhile, in the simulation study we considered the estimands $\psi_h(a)$, $\psi(a; t)$, and $\psi_h(a; t)$, shown in Table \ref{tab:ests}, for treatment values $a \in \{0, 0.05, \dots, 0.95, 1\}$. In the simulation study, we set the bandwidth $h = 0.1$ and the smoothing parameter $\epsilon = 0.01$. In Section \ref{ss:simFixedT} we set the trimming threshold to $t = 0.1$. The resulting estimands are visualized in Figure \ref{fig:simEstimands}. The non-trimmed estimand $\psi_h(a)$ is flat, because $\mathbb{E}[Y | X, A] = \mu(X)$ is not a function of $A$. However, the trimmed estimands $\psi(a; t)$ and $\psi_h(a; t)$ are not flat; this is because the non-trimmed population changes across $a$. That said, the discrepancy between these estimands is somewhat small, given the small $y$-axis scale in Figure \ref{fig:simEstimands}, although this discrepancy could have been made larger by changing the data-generating process. This would not invalidate our simulation results, but nonetheless researchers may find these trimmed estimands difficult to interpret, given that the non-trimmed population changes across $a$. Thus, in practice, we recommend assessing how the non-trimmed population changes across $a$, as we illustrated in our application (Section \ref{s:application}).

\begin{figure}
  \centering
  \includegraphics[scale=0.45]{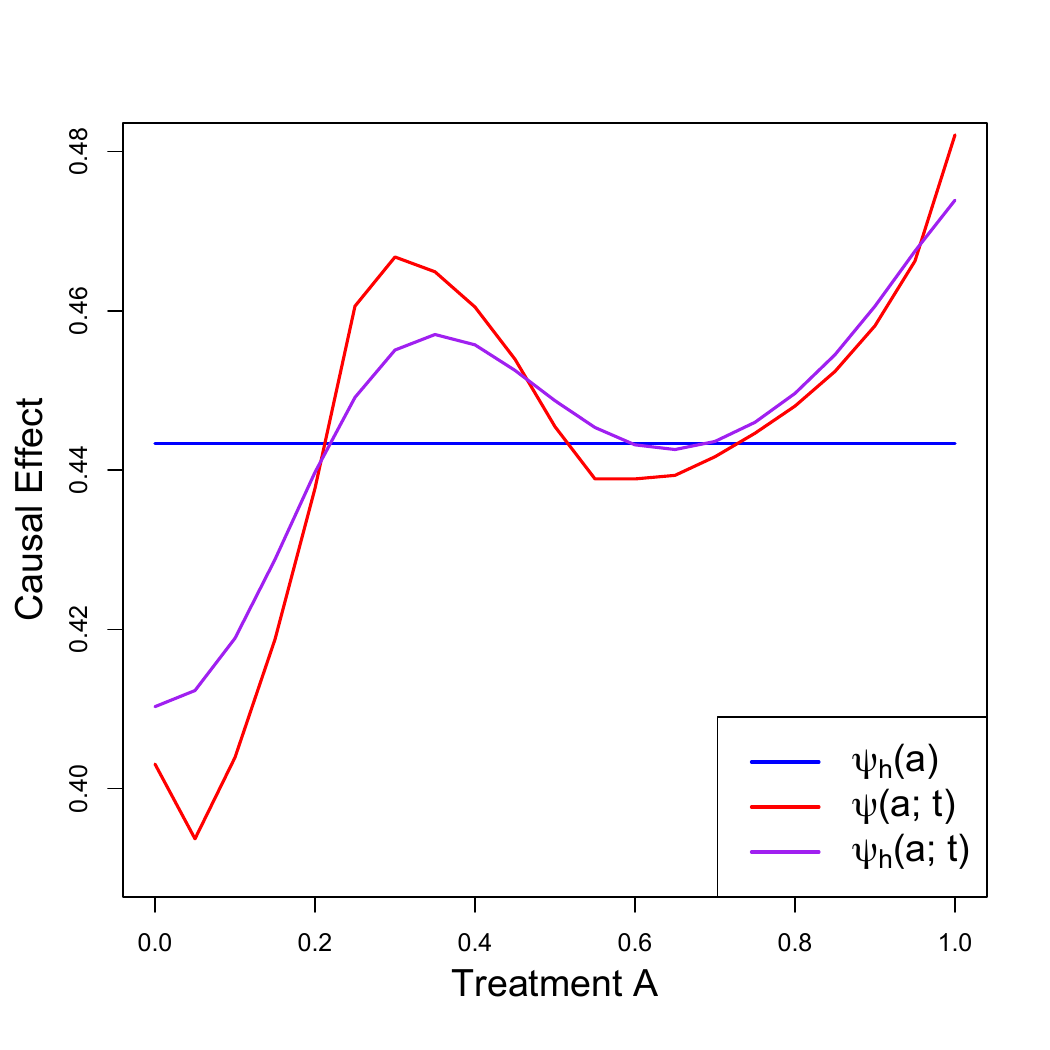}
  \caption{The SATE, TATE, and STATE for $a \in \{0, 0.05, \dots, 0.95, 1\}$ when $t = 0.1$, $h = 0.1$, and $\epsilon = 0.01$. To compute these estimands, we simulated $X \sim \text{Unif}(0,1)$ $n = 10^5$ times, computed the true $\pi(a|X)$ and $\mu(X,a)$, and empirically computed the SATE, TATE, and STATE. To evaluate the inner integrals in the SATE and STATE, we considered values $a_0 \in [-0.5, 1.5]$ in 0.05 increments, such that $\int \int K_h(a_0 - a) da_0 dP(x) \approx 1$ for all $a \in [0,1]$.}
  \label{fig:simEstimands}
\end{figure}

Meanwhile, in Section \ref{ss:simEstT} we considered simulation results when the trimming threshold is estimated. Specifically, we used the trimming estimator $\widehat{t}$ in (\ref{eqn:estimatedT}) for our estimator and the plug-in estimator $\widehat{t}_{\text{pi}} = \widehat{F}_a^{-1}(\gamma)$, where $\widehat{F}_a^{-1}(\gamma)$ is the inverse empirical CDF of $\widehat{\pi}(a|X)$, for the plug-in estimators, where the proportion of trimmed subjects was set at $\gamma = 0.2$. These estimators correspond to true thresholds $t_0$ and $t_0^{\text{pi}}$, respectively. In other words, $t_0$ is defined as the $t$ such that $\psi_h^{\text{den}}(a; t) = 1-\gamma$, and $t_0^{\text{pi}}$ is defined as the $t$ such that $\mathbb{E}[\mathbb{I}(\pi(a | X) > t)] = 1 - \gamma$. Figure \ref{fig:trueThresholds} displays $t_0$ and $t_0^{\text{pi}}$ for this simulation setup, which makes clear that the threshold changes across $a$ in this case. We see that $t_0^{\text{pi}}$ is smaller on the edges of the dose-response curve but higher in the middle of the curve; this is because $t_0$ corresponds to a quantile that is kernel-smoothed across $a$, and this smoothing moderates the magnitude of the threshold. Furthermore, these thresholds are smaller than $t = 0.1$ for extreme treatment values but larger for moderate values. Meanwhile, Figure \ref{fig:estimandsEstT} visualizes the corresponding TATE and STATE estimands, $\psi(a; t_0^{\text{pi}})$ and $\psi_h(a; t_0)$.

\begin{figure}
  \centering
  \begin{subfigure}[b]{0.475\textwidth} 
  \centering
  \includegraphics[scale=0.45]{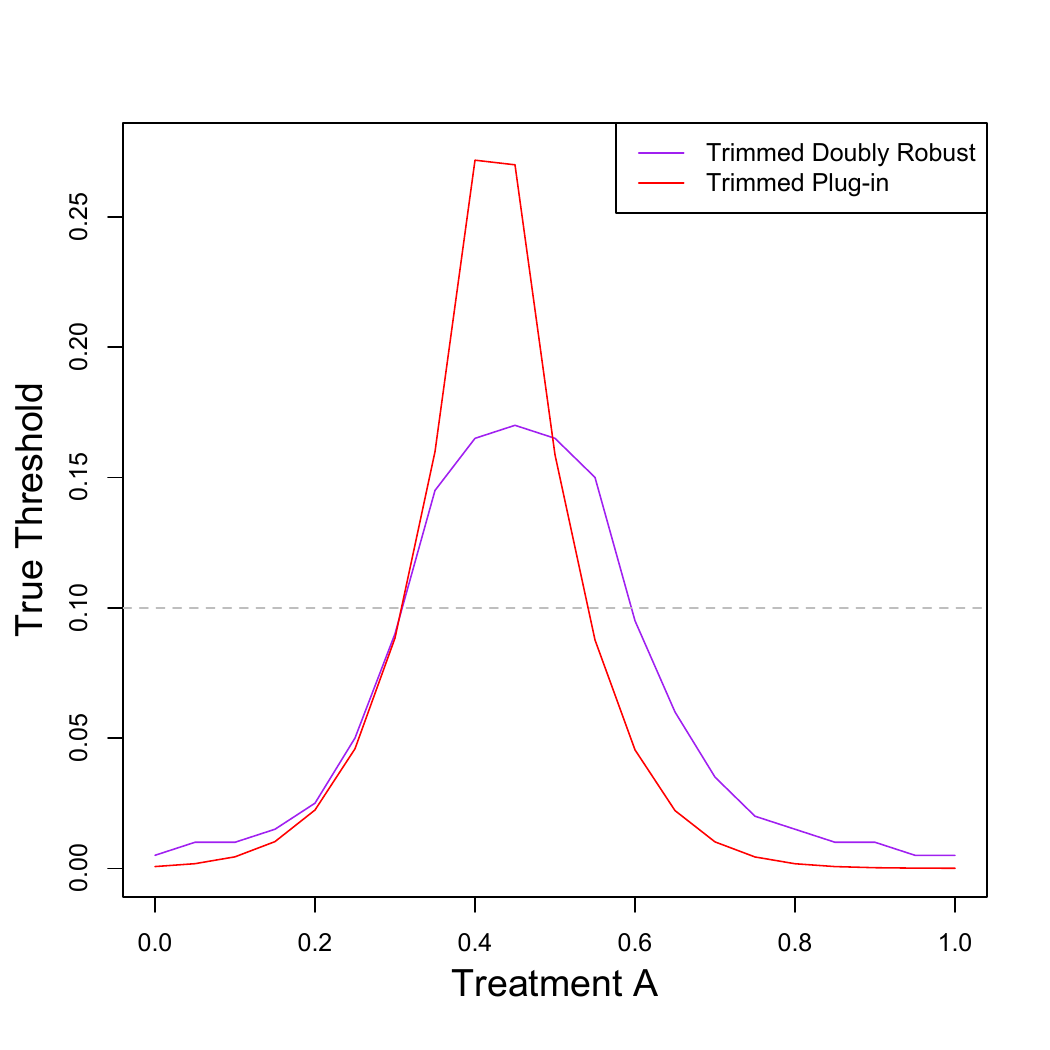}
  \caption{The true trimming thresholds $t_0^{\text{pi}}$ (red) and $t_0$ (purple).}
  \label{fig:trueThresholds}
  \end{subfigure}
  \begin{subfigure}[b]{0.475\textwidth} 
  \centering
  \includegraphics[scale=0.45]{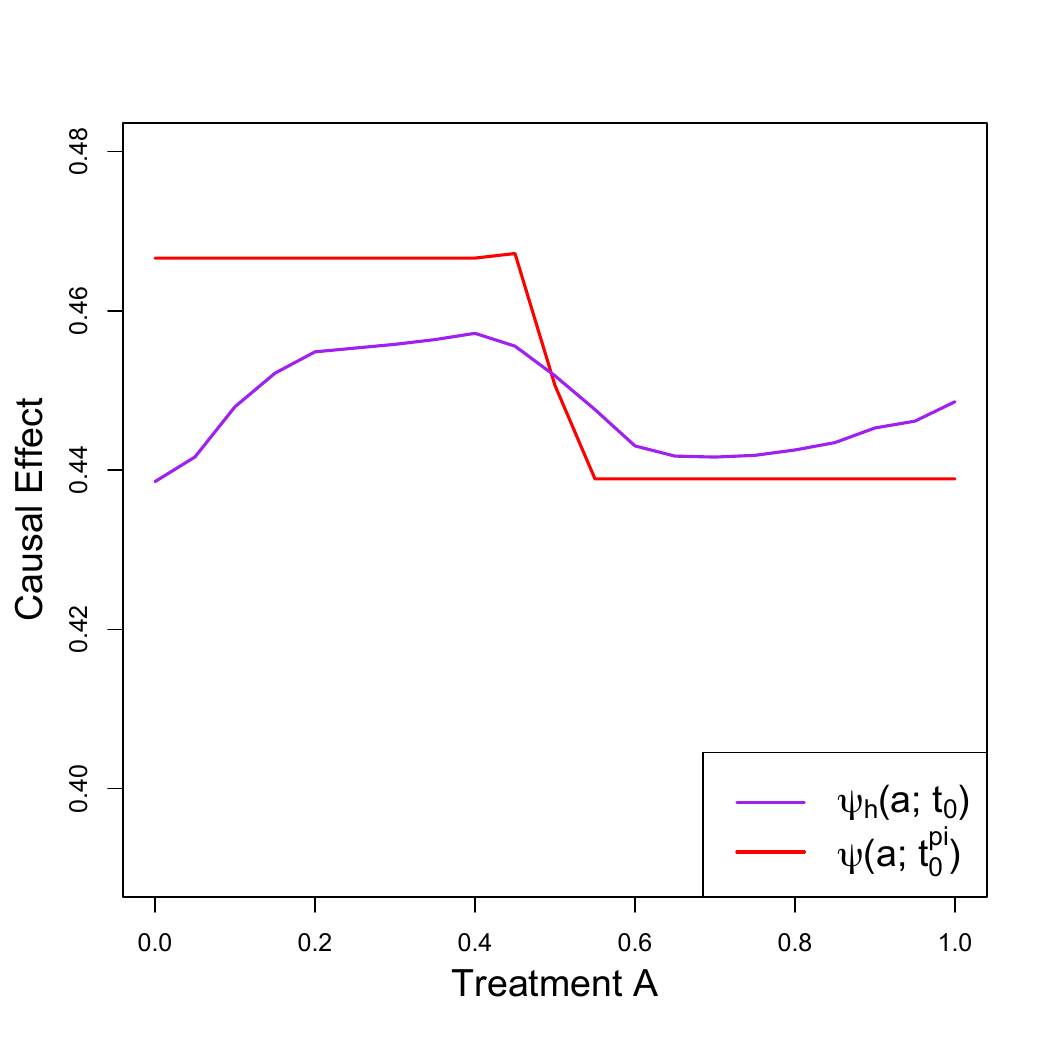}
  \caption{The TATE (red) and STATE (purple) estimands for thresholds $t_0^{\text{pi}}$ and $t_0$.}
  \label{fig:estimandsEstT}
  \end{subfigure}
  \caption{The true trimming thresholds (left) and their corresponding dose-response curves (right) when $\gamma = 0.2$.}
\end{figure}

\section{Smoothing Parameter Selection for Application} \label{as:smoothingParameters}

In Section \ref{s:application}, we selected the bandwidth $h$ using the risk-minimization approach from Section \ref{ss:bandwidth}. Specifically, we computed the estimated risk (\ref{eqn:estimatedRisk}) for $h \in \{0.05, 0.06, \dots, 1.99, 2\}$. The estimated risk across bandwidth values is shown in Figure \ref{fig:estimatedBandwidth}; the $h$ that minimizes this risk is $h = 0.92$.  Meanwhile, to choose the smoothing parameter $\epsilon$, we computed the estimated entropy (\ref{eqn:estimatedEntropy}) for $\epsilon = 10^{-c}$ where $c \in \{1, 1.2, \dots, 4.8, 5\}$. The resulting entropy across $c$ is shown in Figure \ref{fig:estimatedEntropy}. In the application, we chose $\epsilon = 10^{-2}$ because it results in an entropy close to 0.05.

\begin{figure}
\centering
\begin{subfigure}[b]{0.45\textwidth}
\centering
  \includegraphics[scale=0.4]{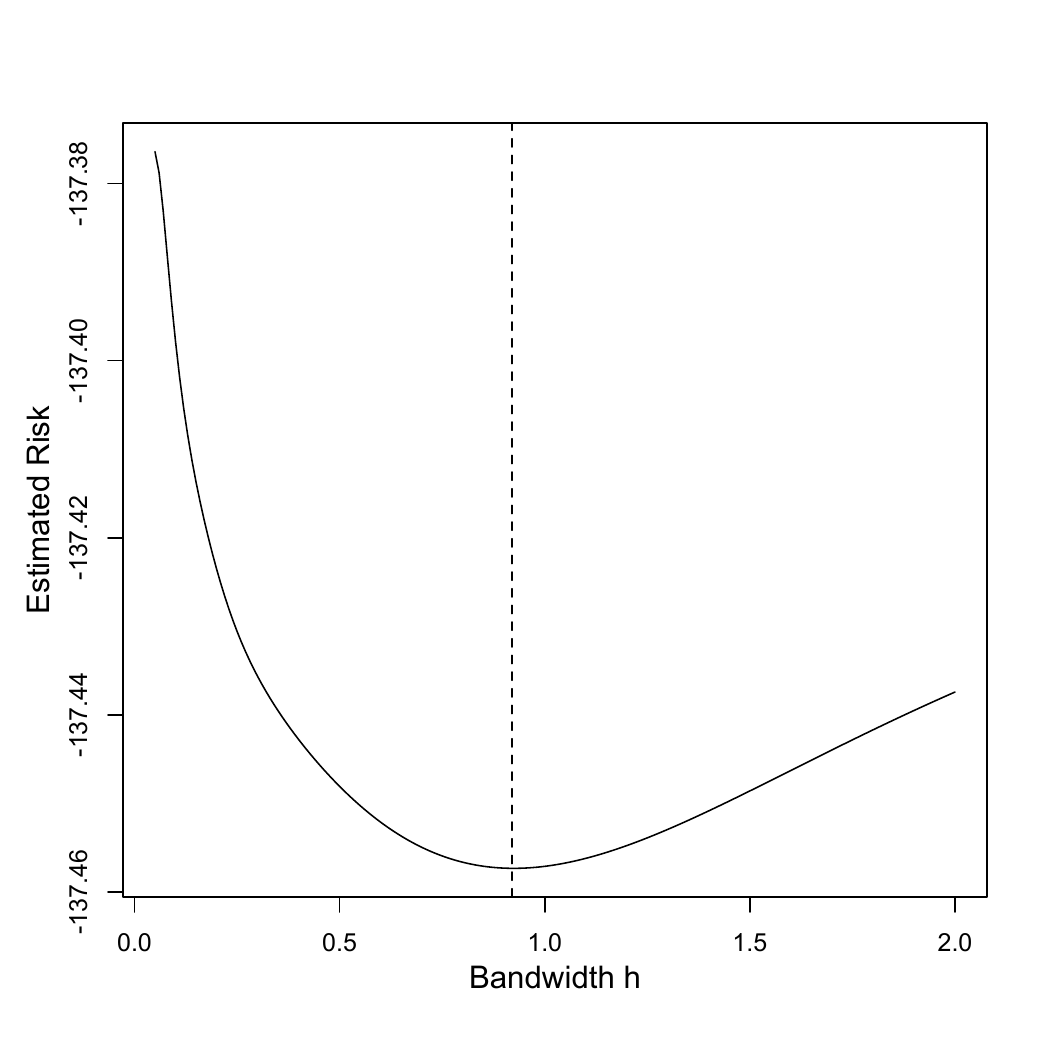}
  \caption{Estimated risk (\ref{eqn:estimatedRisk}) for $h \in \{0.05, 0.06, \dots, 1.99, 2\}$. Vertical line denotes $\widehat{h} = 0.92$.}
  \label{fig:estimatedBandwidth}
\end{subfigure}
\hspace{0.1in}
\begin{subfigure}[b]{0.45\textwidth}
\centering
  \includegraphics[scale=0.4]{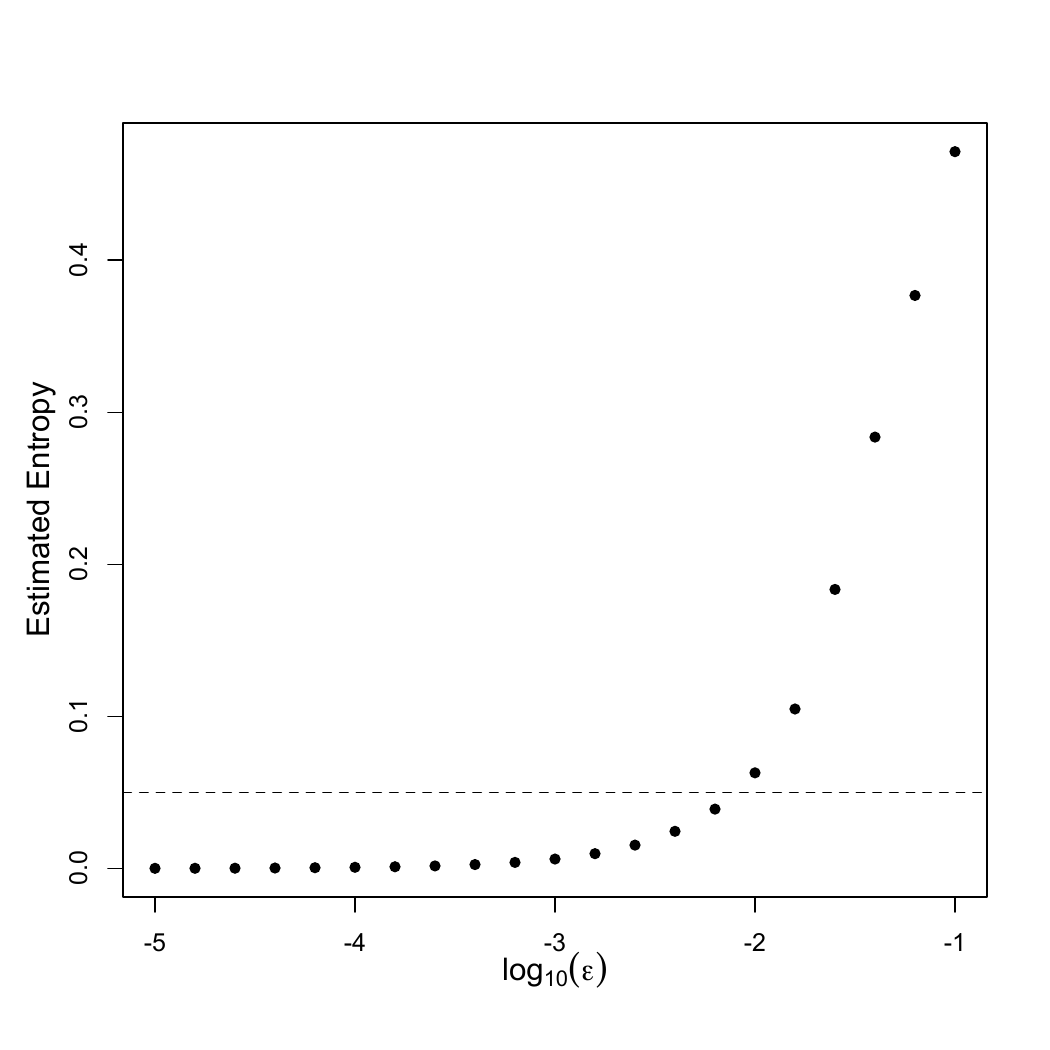}
  \caption{Estimated entropy (\ref{eqn:estimatedEntropy}) for $\epsilon = 10^{-c}$ where $c \in \{0.1, 0.12, \dots, 4.98, 5\}$. Horizontal line denotes 0.05.}
  \label{fig:estimatedEntropy}
\end{subfigure}
\caption{Selecting the bandwidth $h$ (left) and selecting smoothing parameter $\epsilon$ (right).}
\end{figure}

\section{Results for Discrete Treatments} \label{s:appendixBinaryTreatment}

\subsection{Estimands and Efficient Influence Functions}

In the main text, we considered the estimands $\psi_h^{\text{num}}(a; t)$ and $\psi_h^{\text{den}}(a; t)$, which incorporate kernel smoothing and a smoothed trimming indicator $S(\pi(a|x), t)$. Kernel smoothing is required to define an EIF in the case of a continuous treatment. In this section we instead consider the case where the treatment $A$ is discrete, such that $A$ only takes on a finite number of values. In this case, kernel smoothing is not needed to define estimands that have EIFs, although the smoothed trimming indicator $S(\pi(a|x), t)$ is still required. 

Consider the following estimands, for a given treatment value $a$:
  \begin{align*}
    \psi^{\text{num}}_{\epsilon}(a; t) = \int S(\pi(a|x), t) \mu(x,a) dP(x), \hspace{0.1in} \text{and} \hspace{0.1in} \psi^{\text{den}}_{\epsilon}(a; t) = \int S(\pi(a|x), t) dP(x),
  \end{align*}
  where $\epsilon$ is the smoothing parameter for $S(\pi(a|x), t)$. Following the same proof mechanics as that of Theorem \ref{thm:numDenomEIFs}, we have the following uncentered EIFs when $t$ is fixed.
  \begin{proposition}
  \label{prop:eifDiscreteTreatment}
    When the trimming threshold $t$ is fixed, the uncentered EIFs for $\psi^{\text{num}}_{\epsilon}(a; t)$ and $\psi^{\text{den}}_{\epsilon}(a; t)$ are
      \begin{align*}
    \varphi^{\text{num}}_{\epsilon}(a; t) &= \mathbb{I}(A = a) \frac{\{Y - \mu(X,A)\} S(\pi(A|X), t)}{\pi(A|X)} + S(\pi(a|X),t) \mu(X,a) \\ &+ \mathbb{I}(A = a) \mu(X,A) \frac{\partial S(\pi(A | X), t) }{\partial \pi} - \mu(X,a) \frac{\partial S(\pi(a | X), t) }{\partial \pi} \pi(a|X),
    \end{align*}
    and
    \begin{align*}
  \varphi^{\text{den}}_{\epsilon}(a; t) &= S(\pi(a | X), t)  + \mathbb{I}(A = a) \frac{\partial S(\pi(A | X), t) }{\partial \pi} - \frac{\partial S(\pi(a | X), t) }{\partial \pi} \pi(a|X)  .
  \end{align*}
\end{proposition}
The corresponding one-step estimators based on the above proposition are:
\begin{align*}
  \widehat{\psi}^{\text{num}}_{\epsilon}(a; t) = \mathbb{P}_n\{ \widehat{\varphi}^{\text{num}}_{\epsilon}(a; t) \}, \hspace{0.1in} \widehat{\psi}^{\text{den}}_{\epsilon}(a; t) = \mathbb{P}_n\{ \widehat{\varphi}^{\text{den}}_{\epsilon}(a; t) \},
\end{align*}
where $\widehat{\varphi}^{\text{num}}_{\epsilon}(a; t)$ and $\widehat{\varphi}^{\text{den}}_{\epsilon}(a; t)$ are equal to $\varphi^{\text{num}}_{\epsilon}(a; t)$ and $\varphi^{\text{den}}_{\epsilon}(a; t)$, but with $\pi(a|x)$ and $\mu(x,a)$ replaced with estimators $\widehat{\pi}(a|x)$ and $\widehat{\mu}(x,a)$. Then, one can compute the ratio $\widehat{\psi}^{\text{num}}_{\epsilon}(a; t)/\widehat{\psi}^{\text{den}}_{\epsilon}(a; t)$ to estimate the trimmed causal effect $\psi^{\text{num}}_{\epsilon}(a; t)/\psi^{\text{den}}_{\epsilon}(a;t)$.

In the specific case of a binary treatment, researchers are typically interested in the contrast $a = 1$ versus $a = 0$. Using the above results in this case, the estimand would be:
\begin{align*}
  \frac{\psi^{\text{num}}_{\epsilon}(1; t)}{\psi^{\text{den}}_{\epsilon}(1;t)} - \frac{\psi^{\text{num}}_{\epsilon}(0; t)}{\psi^{\text{den}}_{\epsilon}(0;t)} = \frac{\int S(\pi(1|x), t) \mu(x,1) dP(x)}{\int S(\pi(1|x), t) dP(x)} - \frac{\int S(\pi(0|x), t) \mu(x,0) dP(x)}{\int S(\pi(0|x), t) dP(x)}.
\end{align*}
The above estimand trims treatment and control separately, such that the non-trimmed populations for treatment and control may differ.  Instead, researchers may wish to construct a single non-trimmed population. One can consider the following trimmed estimand:
\begin{align}
  \psi_{\epsilon}(t) = \frac{\psi^{\text{num}}_{\epsilon}(t)}{\psi^{\text{den}}_{\epsilon}(t)} = \frac{\int S(\pi(1|x), t) \{ \mu(x, 1) - \mu(x, 0)\} dP(x)}{\int S(\pi(1|x), t) dP(x)}. \label{eqn:binaryTrimmedEstimand}
\end{align}
Here, $S(\pi(1|x), t)$ corresponds to a smoothed analog of the trimming indicator $\mathbb{I}(t < \pi(1|x) < 1 - t)$. For example, \cite{yang2018asymptotic} defined $S(\pi(1|x), t)$ as:
\begin{align*}
  S(\pi(1|x), t) = \Phi_{\epsilon}\{\pi(1|x) - t\} \Phi_{\epsilon}\{t - \pi(1|x)\}.
\end{align*}
The following proposition establishes the EIFs for $\psi^{\text{num}}_{\epsilon}(t)$ and $\psi^{\text{den}}_{\epsilon}(t)$ in (\ref{eqn:binaryTrimmedEstimand}).

\begin{proposition}
\label{prop:eifBinaryTreatment}
  When the trimming threshold $t$ is fixed, the uncentered EIFs of $\psi^{\text{num}}_{\epsilon}(t)$ and $\psi^{\text{den}}_{\epsilon}(t)$ are
  \begin{align*}
    \varphi^{\text{num}}_{\epsilon}(t) &= S(\pi(1|X), t) \{ \mu(X,1) - \mu(X,0) \} + \frac{\partial S(\pi(1|X), t)}{\partial \pi} (A - \pi(1|X)) \{ \mu(X, 1) - \mu(X, 0)\} \\
    &+ S(\pi(1|X), t) \left\{ \frac{A (Y - \mu(X,1))}{\pi(1|X)} - \frac{(1-A)(Y - \mu(X,0))}{(1 - \pi(1|X))} \right\},
    \end{align*}
    and
    \begin{align*}
    \varphi^{\text{den}}_{\epsilon}(t) &= S(\pi(1|X), t) + \frac{\partial S(\pi(1|X), t)}{\partial \pi} (A - \pi(1|X)).
  \end{align*}
\end{proposition}
\begin{proof}
  First, recall that we have the following influence functions when $X$ is discrete:
  \begin{align*}
    \mathbb{IF}\{ \mu(x, 1) - \mu(x, 0) \} &= \frac{A \mathbb{I}(X=x)}{\pi(1|x)p(x)}(Y - \mu(x,1)) - \frac{(1-A)\mathbb{I}(X=x)}{(1 - \pi(1|x))p(x)}(Y - \mu(x,0)) \\
    \mathbb{IF}\{ S(\pi(1|x), t) \} &= \frac{\partial S(\pi(1|x), t)}{\partial \pi} \frac{\mathbb{I}(X = x)}{P(X=x)}(A - \pi(1|x)) \\
   \mathbb{IF}\{p(x)\} &= \mathbb{I}(X = x) - p(x)
  \end{align*}
  Thus, we have, by the discretization trick and product rule:
  \begin{align*}
    \mathbb{IF}\{ \psi^{\text{den}}_{\epsilon}(t) \} &= \sum_x \mathbb{IF}\{ S(\pi(1|x), t) \} p(x) + S(\pi(1|x), t) \mathbb{IF}\{p(x)\} \\
    &= \frac{\partial S(\pi(1|X), t)}{\partial \pi} (A - \pi(1|X)) + S(\pi(1|X), t) - \psi^{\text{den}}_{\epsilon}(t)
  \end{align*}
  Similarly, we have:
  \begin{align*}
    \mathbb{IF}\{ \psi^{\text{num}}_{\epsilon}(t) \} &= \sum_x \mathbb{IF}\{ S(\pi(1|x), t) \} \{ \mu(x, 1) - \mu(x, 0)\} p(x) \\ &+ S(\pi(1|x), t) \mathbb{IF}\{ \mu(x, 1) - \mu(x, 0)\} p(x) \\ &+ S(\pi(1|x), t) \{ \mu(x, 1) - \mu(x, 0)\} \mathbb{IF}\{p(x)\} \\ \\
    &= \frac{\partial S(\pi(1|X), t)}{\partial \pi} (A - \pi(1|X)) \{ \mu(X, 1) - \mu(X, 0)\} \\
    &+ S(\pi(1|X), t) \left\{ \frac{A (Y - \mu(X,1))}{\pi(1|X)} - \frac{(1-A)(Y - \mu(X,0))}{(1 - \pi(1|X))} \right\} \\
    &+ S(\pi(1|X), t) \{ \mu(X,1) - \mu(X,0) \} - \psi^{\text{num}}_{\epsilon}(t)
  \end{align*}
  This completes the proof.
\end{proof}
The corresponding estimator for $\psi_{\epsilon}(t)$ in (\ref{eqn:binaryTrimmedEstimand}) is $\widehat{\psi}_{\epsilon}(t) = \widehat{\psi}^{\text{num}}_{\epsilon}(t)/\widehat{\psi}^{\text{den}}_{\epsilon}(t)$, where
\begin{align*}
  \widehat{\psi}^{\text{num}}_{\epsilon}(t) = \mathbb{P}_n\{ \widehat{\varphi}^{\text{num}}_{\epsilon}(t) \}, \hspace{0.1in} \widehat{\psi}^{\text{den}}_{\epsilon}(t) = \mathbb{P}_n\{ \widehat{\varphi}^{\text{den}}_{\epsilon}(t) \},
\end{align*}
where $\widehat{\varphi}^{\text{num}}_{\epsilon}(t)$ and $\widehat{\varphi}^{\text{den}}_{\epsilon}(t)$ are equal to $\varphi^{\text{num}}_{\epsilon}(t)$ and $\varphi^{\text{den}}_{\epsilon}(t)$, but with $\pi(1|x)$, $\mu(x,1)$, and $\mu(x,0)$ replaced with estimators $\widehat{\pi}(1|x)$, $\widehat{\mu}(x,1)$, and $\widehat{\mu}(x,0)$.

\subsection{Simulation Results for Binary Treatments}

Similar to the simulation study in the main text, we consider an illustrative example with a single covariate $X$ that ranges from 0 to 1 and a continuous outcome $Y$. However, in this example, the treatment $A$ is binary. For simplicity, we focus on estimating the average treatment potential outcome $\psi(1) = \mathbb{E}[Y(1)]$, which is identified as $\psi(1) = \mathbb{E}[\mu(X,1)]$ when consistency, unconfoundedness, and positivity hold.

The propensity score $\pi(1|X) = P(A = 1 | X)$ and outcome regression $\mu(X,1) = \mathbb{E}[Y | X, A = 1]$ were specified as $m(X)$ and $\mu(X)$, respectively, both visualized in Figure \ref{fig:simFunctions} (see Section \ref{s:appendixSimDetails} for the exact specifications for these functions). The propensity score is small for $X < 0.25$, such that $\psi(1)$ is difficult to estimate. The trimmed estimands are defined as:
\begin{align*}
  \psi(1; t) = \frac{\mathbb{E}[\mu(X,1) \mathbb{I}\{ \pi(1|X) \geq t \}]}{\mathbb{E}[\mathbb{I}\{ \pi(1|X) \geq t \}]}, \hspace{0.1in} \text{and }
  \psi_{\epsilon}(1; t) = \frac{\mathbb{E}[\mu(X,1)S(\pi(1|X), t)]}{\mathbb{E}[S(\pi(1|X), t)]},
\end{align*}
where $S(\pi(1|X),t) = \Phi_{\epsilon}\{ \pi(1|X) - t\}$ is a smoothed version of the trimming indicator $\mathbb{I}\{ \pi(1|X) \geq t \}$. In this section, we call $\psi(1; t)$ the trimmed average treatment effect (TATE) and $\psi_{\epsilon}(1; t)$ the smoothed TATE (STATE). The TATE and STATE will be close when $\epsilon$ is small; for simplicity we fix $\epsilon = 0.01$.

Similar to the simulation study in the main text, we generated 1000 datasets containing $(X, A, Y)$, where $X \sim \text{Unif}(0, 1)$, $A \sim \text{Bern}(\pi(1|X))$, and $Y \sim N(\mu(X,1), 0.5^2)$. Each dataset contained $n = 1000$ subjects. For each dataset we simulated estimators as
\begin{align*}
  \widehat{\pi}(1|X) &\sim \pi(1 | X) + \text{expit}[ \text{logit}\{\pi(1 | X)\} + N(n^{-\alpha}, n^{-2 \alpha})] \\
  \widehat{\mu}(X, 1) &\sim \mu(X, 1) + N(n^{-\alpha}, n^{-2 \alpha}),
\end{align*}
such that the root mean squared error (RMSE) of $\widehat{\pi}(1 | X)$ and $\widehat{\mu}(X, 1)$ are $O_P(n^{-\alpha})$, and thus we can control the estimators' convergence rate via the rate parameter $\alpha$. We considered convergence rates $\alpha \in \{0.1, 0.2, \dots, 0.5\}$. As an example, Figure \ref{fig:binaryPSHat} displays $\widehat{\pi}(1 | X)$ for the first simulated dataset when $\alpha = 0.1$.

\begin{figure}
\centering
    \includegraphics[scale=0.5]{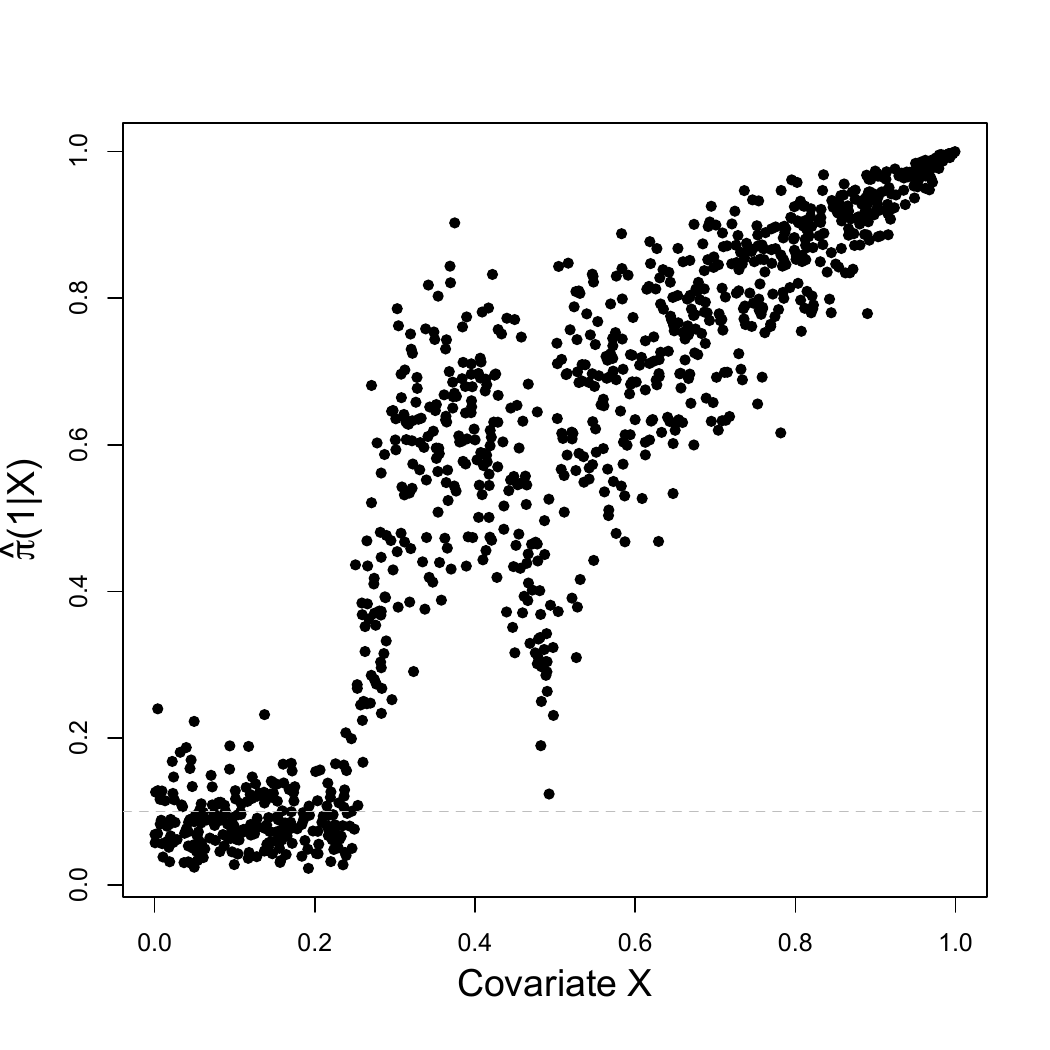}
    \caption{Example of $\widehat{\pi}(1|X)$ for the first simulated dataset when convergence rate $\alpha = 0.1$. Horizontal dashed line denotes $t = 0.1$.}
    \label{fig:binaryPSHat}
\end{figure}

Once $\widehat{\pi}(1 | X)$ and $\widehat{\mu}(X, 1)$ are generated, we computed the following four estimators:
\begin{enumerate}
  \item \textbf{Doubly Robust}: $\widehat{\psi}(1) = \mathbb{P}_n\{ \widehat{\varphi}(1)\}$, where $\widehat{\varphi}(1) = \widehat{\mu}(X,1) + \frac{A\{Y - \widehat{\mu}(X,1)\}}{\widehat{\pi}(1|X)}$ are the estimated influence function values for the typical doubly robust estimator. This targets the ATE $\psi(1)$.
  \item \textbf{Trimmed Plug-in}: $\widehat{\psi}(1; t) = \frac{\mathbb{P}_n[\widehat{\mu}(X,1) \mathbb{I}\{ \widehat{\pi}(1|X) \geq t \}]}{\mathbb{P}_n[\mathbb{I}\{ \widehat{\pi}(1|X) \geq t \}]}$. This targets the TATE and is based on the identification result for $\psi(1; t)$.
  \item \textbf{EIF-based Trimmed Plug-in}: $\widehat{\psi}^{\text{alt}}(1; t) = \frac{\mathbb{P}_n[\widehat{\varphi}(1) \mathbb{I}\{ \widehat{\pi}(1|X) \geq t \}]}{\mathbb{P}_n[\mathbb{I}\{ \widehat{\pi}(1|X) \geq t \}]}$. This also targets the TATE and averages the typical doubly robust estimator on the trimmed sample.
  \item \textbf{Trimmed Doubly Robust}: $\widehat{\psi}_{\epsilon}(1; t) = \frac{\widehat{\psi}^{\text{num}}_{\epsilon}(1; t)}{\widehat{\psi}^{\text{den}}_{\epsilon}(1; t)} = \frac{\mathbb{P}_n\{ \widehat{\varphi}^{\text{num}}_{\epsilon}(1; t) \} }{\mathbb{P}_n\{ \widehat{\varphi}^{\text{den}}_{\epsilon}(1; t) \} }$, where $\varphi^{\text{num}}_{\epsilon}(1; t)$ and $\varphi^{\text{den}}_{\epsilon}(1; t)$ are defined in Proposition \ref{prop:eifDiscreteTreatment} and represent the EIFs for the numerator and denominator of $\psi_{\epsilon}(1; t)$. This targets the STATE.
\end{enumerate}
Figure \ref{fig:binaryPerformance} displays the estimators' performance in terms of RMSE and 95\% confidence interval coverage when $t = 0.1$ or when it is estimated as the $\gamma = 0.2$ quantile of the propensity score, similar to what was studied in the main text. Note that RMSE and coverage were computed according to each estimator's respective estimand. For example, our estimator targets $\psi_{\epsilon}(1; t)$ when $t = 0.1$ and $\psi_{\epsilon}(1; t^0)$ when $t$ is estimated, where $t_0$ is the threshold such that $\psi_{\epsilon}^{\text{den}}(1; t_0) = 1 - \gamma = 0.8$. The results are very similar to the results in the main text simulation study. In terms of RMSE, the trimmed plug-in estimator $\widehat{\psi}(1;t)$ performs poorly when the convergence rate is slow, whereas our trimmed doubly robust estimator $\widehat{\psi}_{\epsilon}(1; t)$ and the EIF-based trimmed plug-in estimator $\widehat{\psi}^{\text{alt}}(1; t)$ outperform the typical doubly robust estimator $\widehat{\psi}(1)$. Meanwhile, in terms of coverage, our estimator performs slightly better than the EIF-based trimmed plug-in estimator, although both approach the nominal level as the convergence rate increases.

\begin{figure}
\centering
  \begin{subfigure}[b]{0.49\textwidth}
  \centering
    \includegraphics[scale=0.45]{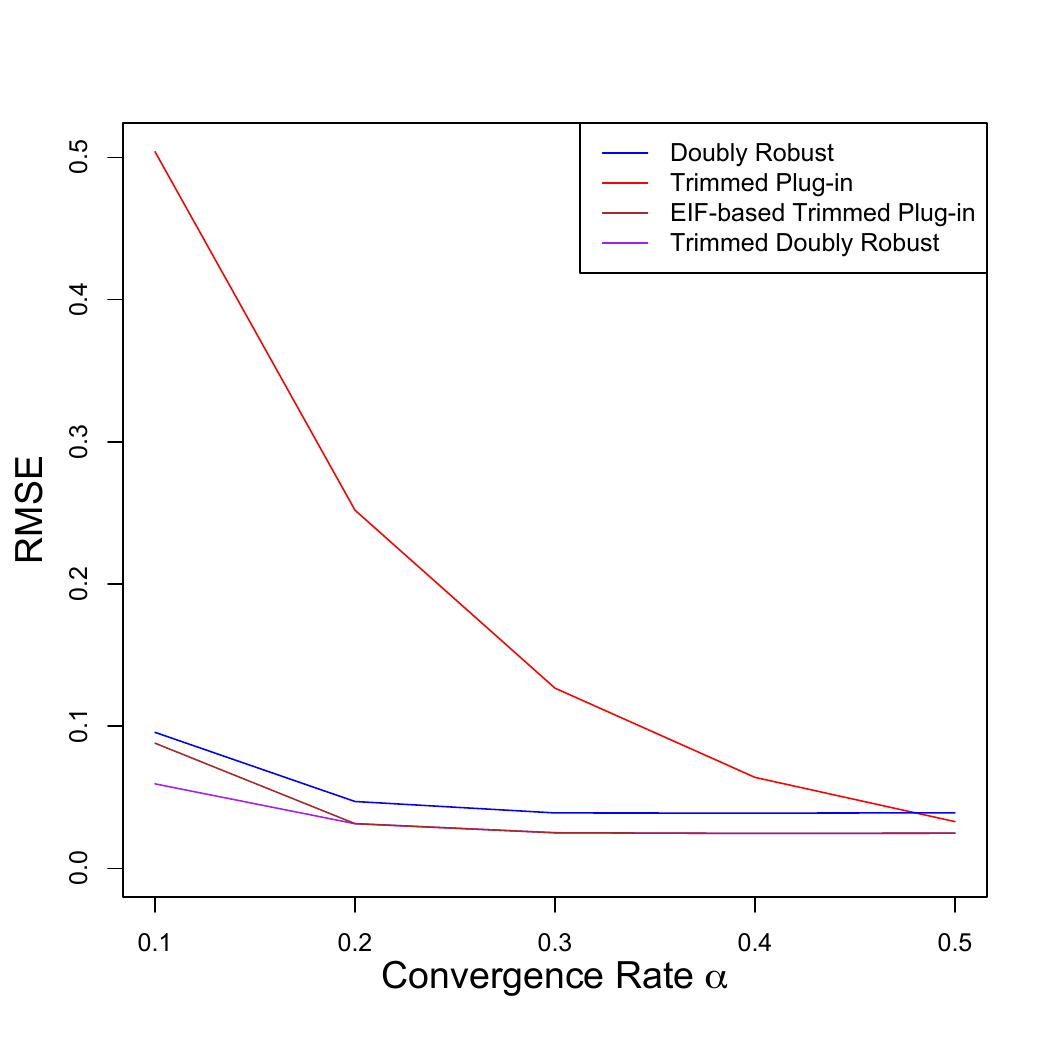}
    \caption{Estimators' RMSE when $t = 0.1$.}
    \label{fig:binaryEstFixedT}
  \end{subfigure}
  \begin{subfigure}[b]{0.49\textwidth}
  \centering
    \includegraphics[scale=0.45]{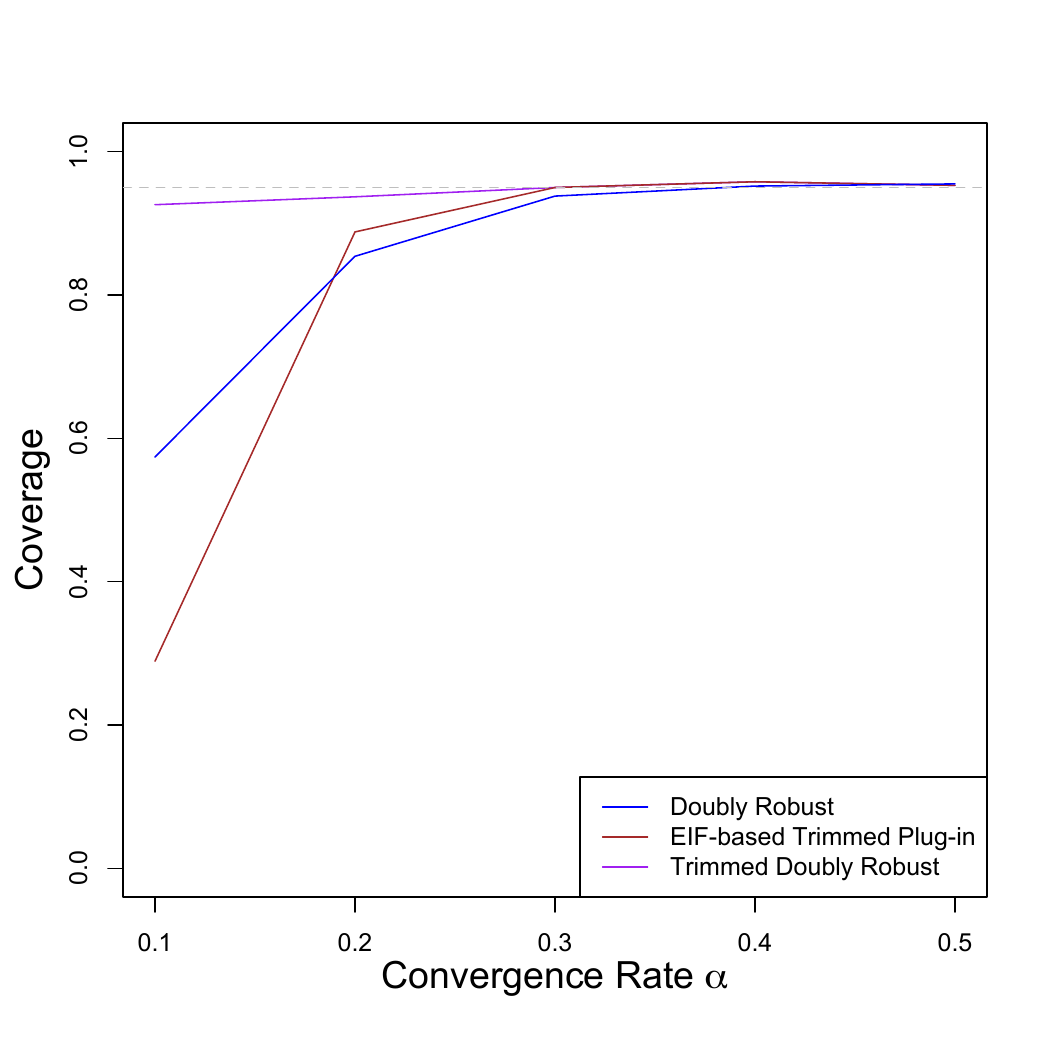}
    \caption{Estimators' coverage when $t = 0.1$.}
    \label{fig:binaryCIFixedT}
  \end{subfigure}
  \begin{subfigure}[b]{0.49\textwidth}
  \centering
    \includegraphics[scale=0.45]{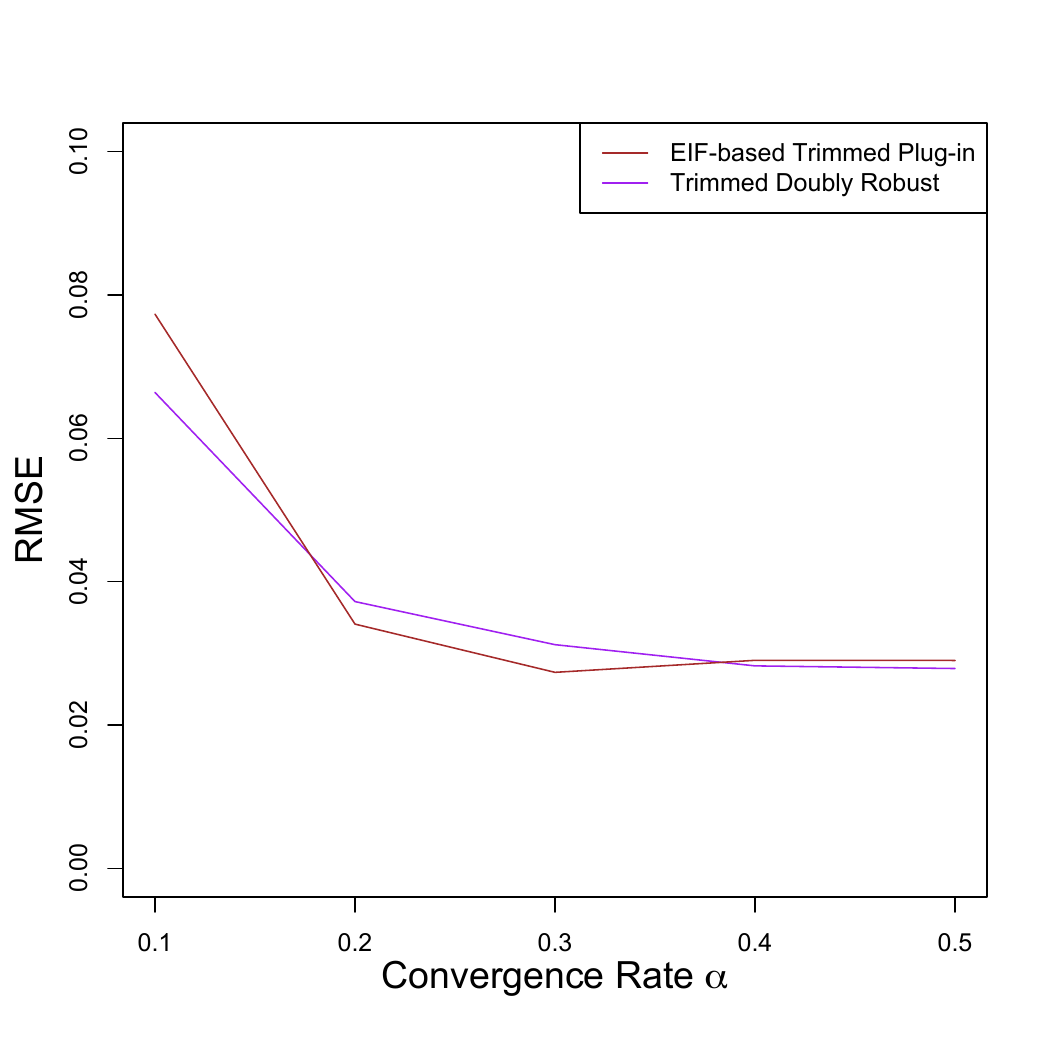}
    \caption{Estimators' RMSE when $t$ is estimated.}
    \label{fig:binaryEstEstT}
  \end{subfigure}
  \begin{subfigure}[b]{0.49\textwidth}
  \centering
    \includegraphics[scale=0.45]{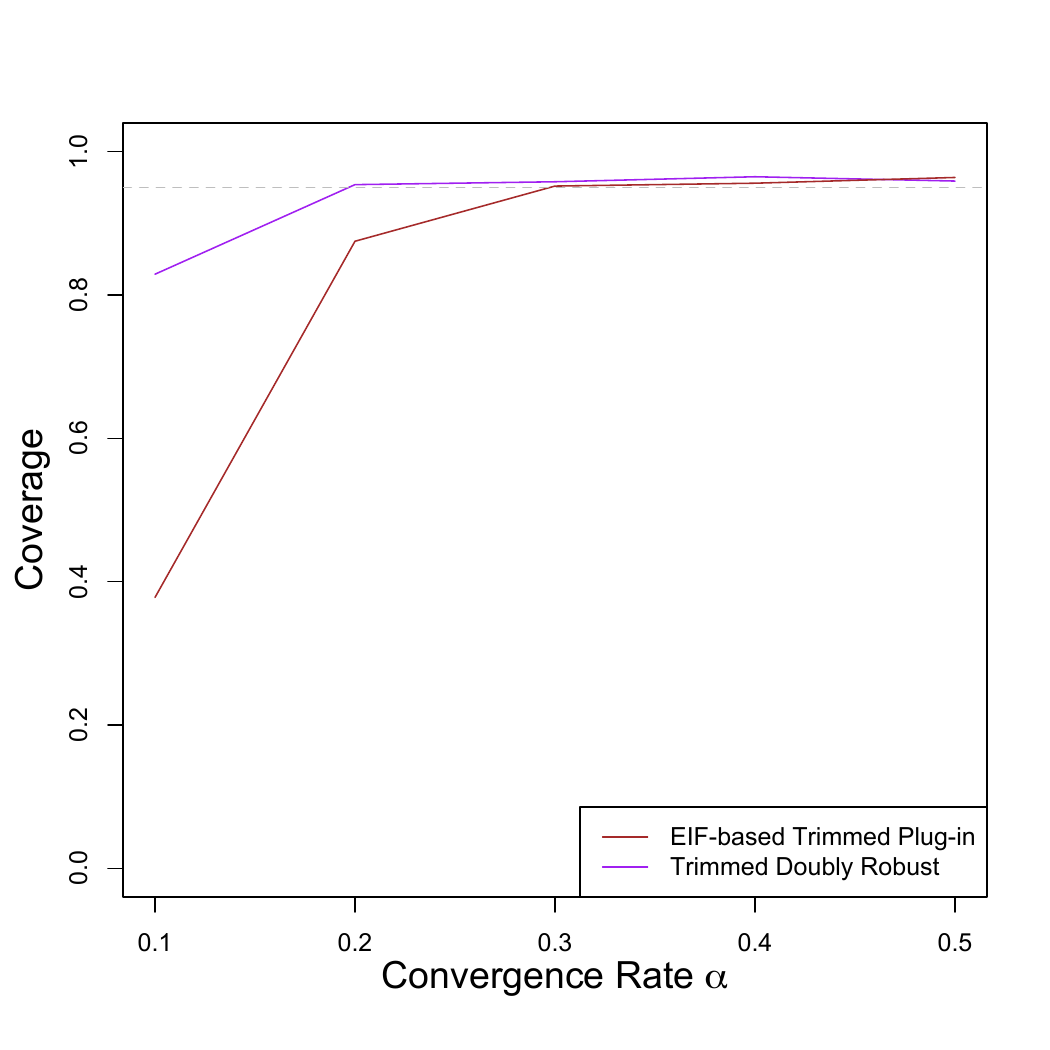}
    \caption{Estimators' coverage when $t$ is estimated.}
    \label{fig:binaryCIEstT}
  \end{subfigure}
  \caption{RMSE (left) and 95\% confidence interval coverage (right) of the estimators when the threshold is fixed at $t = 0.1$ (top) or estimated as the $\gamma = 0.2$ quantile of $\pi(1|X)$ (bottom). Results displayed across different convergence rates $\alpha$ for $\widehat{\pi}(1|X)$ and $\widehat{\mu}(1,X)$.}
  \label{fig:binaryPerformance}
\end{figure}

The fact that $\widehat{\psi}^{\text{alt}}(1; t)$ performs similarly to our estimator may at first seem surprising, because this estimator uses a plug-in estimator for the denominator of the TATE, $\mathbb{E}[\mathbb{I}\{ \pi(1|X) \geq t \}]$, and thus we may expect this estimator to inherit the convergence rate of $\widehat{\pi}(1|X)$. However, even when $\pi(1|X)$ is estimated poorly, the trimming indicator $\mathbb{I}\{ \pi(1|X) \geq t\}$ may nonetheless be estimated well. Reconsider Figure \ref{fig:binaryPSHat}, which displays an example of $\widehat{\pi}(1|X)$ when $\alpha = 0.1$. We see that $\widehat{\pi}(1|X)$ is closer to $\pi(1|X)$ when $\pi(1|X)$ is small or large; this is a by-product of constraining $\widehat{\pi}(1|X)$ to be between 0 and 1 to be a proper propensity score estimator. As a result, even though $\widehat{\pi}(1|X)$ is estimated poorly, the set of trimmed subjects is close to correct; in other words, $\mathbb{I}\{ \pi(1|X) \geq t\}$ is still well-estimated.

Lastly, we can note that, in this simulation setup, the typical doubly robust estimator does not degrade as severely as in the simulation setup in the main text. This is likely because the propensity scores are less severe, compared to the simulation study in the main text. As shown in Figure \ref{fig:binaryPSHat}, only some subjects' propensity scores were close to zero in this simulation setup, whereas when the treatment was continuous, subjects' propensity scores were close to zero for many treatment values. In other words, positivity violations at $A = 1$ in this simulation study are less severe than positivity violations across treatment values in the simulation study in the main text. This is also likely why the EIF-based trimmed plug-in estimator performed similarly to our estimator regardless of whether the trimming threshold was fixed or estimated in this simulation study, whereas it could perform worse than our estimator in the continuous case. This illustrates that positivity violations can be a greater concern when the treatment is continuous.

\end{document}